\documentclass[nocopyrightspace, preprint]{sigplanconf}

\usepackage{amsmath, amsthm, amssymb}
\usepackage{mathabx}
\usepackage[latin1]{inputenc}
\usepackage[english]{babel}
\usepackage{times}

\usepackage{tablefootnote}
\usepackage{mathtools}
\usepackage{graphicx}
\usepackage{caption}
\usepackage{subcaption}
\usepackage{longtable}
\usepackage{booktabs}
\usepackage{tabu}
\usepackage[referable]{threeparttablex}
\usepackage{multirow}
\usepackage{pdfpages}
\usepackage{thmtools}
\usepackage{thm-restate}

\usepackage[capitalise]{cleveref}

\crefname{figure}{Figure}{Figure}
\crefformat{footnote}{#2\footnotemark[#1]#3}
\usepackage{tikz}
\usepackage{comment}
\usepackage{parskip}

\usepackage{todonotes}
\usepackage[ruled, linesnumbered]{algorithm2e}
\usepackage{tikz}
\usetikzlibrary{arrows,automata,shapes,decorations,decorations.markings,calc, matrix,decorations.pathmorphing, patterns}

\usepackage{bm}
\usepackage{pgfplots}
\usepackage{enumerate,paralist}
\usepackage{pbox}

\usepackage{appendix}

\usepackage{calc}
\usepackage{fp}
\usepackage{varwidth}

\newtheorem{definition}{Definition}
\newtheorem{theorem}{Theorem}
\newtheorem{corollary}{Corollary}
\newtheorem{lemma}{Lemma}

\newtheorem{claim}{Claim}

\theoremstyle{remark}
\newtheorem{remark}{Remark}

\theoremstyle{definition}

\DeclareMathAlphabet{\mathpzc}{OT1}{pzc}{m}{it}

\usepackage{graphicx} 
\newcommand{\Nats}{\mathbb{N}}

\newcommand{\Reals}{\mathbb{R}}

\newcommand{\Level}{\mathsf{Lv}}
\newcommand{\Fwd}{\mbox{\sc Fwd}}
\newcommand{\Bwd}{\mbox{\sc Bwd}}
\newcommand{\Fwdstar}{\Fwd^+}
\newcommand{\Bwdstar}{\Bwd^+}

\newcommand{\ov}{\overline}

\newcommand{\PNode}[1]{\ov{#1}}
\newcommand{\Refines}{\sqsubseteq}
\newcommand{\Strictrefines}{\sqsubset}
\newcommand{\Distance}{d}
\newcommand{\Distancecomp}{d'}
\newcommand{\Weight}{\mathsf{wt}}

\newcommand{\Mergealgo}{\mathsf{Merge}}
\newcommand{\Splitalgo}{\mathsf{Split}}
\newcommand{\Nsep}{2/\delta}
\newcommand{\Seplist}{\mathcal{X}}
\newcommand{\Complist}{\mathcal{Y}}

\newcommand{\Bag}{B}
\newcommand{\RBag}{\mathcal{B}}
\newcommand{\RHBag}{\widehat{\mathcal{B}}}
\newcommand{\Tree}{\mathrm{Tree}}
\newcommand{\restr}[1]{[#1]}

\newcommand{\Rankalgo}{\mathsf{Rank}}

\newcommand{\Concpreprocessalgo}{\mathsf{ConcurPreprocess}}
\newcommand{\ConcurAP}{\mathsf{ConcurAP}}

\newcommand{\Concqueryalgo}{\mathsf{ConcurQuery}}

\newcommand{\Nhood}{\mathsf{Nh}}
\newcommand{\Nhoodnodes}{\mathsf{NhV}}
\newcommand{\Baghood}{\mathsf{Bh}}
\newcommand{\Comp}{\mathcal{C}}
\newcommand{\Zero}{\overline{\mathbf{0}}}
\newcommand{\One}{\overline{\mathbf{1}}}
\newcommand{\True}{\mathsf{True}}
\newcommand{\False}{\mathsf{False}}

\newcommand{\Ranktree}{\mathsf{R_G}}
\newcommand{\Rankhtree}{\widehat{\mathsf{R_G}}}
\newcommand{\Rep}{\mathsf{Replace}}

\newcommand{\Time}{\mathcal{T}}
\newcommand{\Space}{\mathcal{S}}

\newcommand{\Path}{\rightsquigarrow}

\newcommand{\Weightmap}{{\mathpzc{wt}}}

\newcommand{\ProductTree}{\mathsf{ConcurTree}}
\newcommand\numberthis{\addtocounter{equation}{1}\tag{\theequation}}
\newcommand{\From}{\Fwdstar}
\newcommand{\To}{\Bwdstar}
\newcommand{\AncestV}[1]{\mathpzc{#1}}

\newcommand{\Bagfactor}{4\cdot \lambda /\delta}
\newcommand{\Balfactor}{((1+\delta)/2)^{\lambda-1}}
\newcommand{\Levelfactor}{\lambda}

\newcommand{\Philosophers}{\mathsf{Dining Philosophers}}

\newcommand{\Prod}[1]{\langle #1 \rangle}
\newcommand{\eps}{\epsilon}

\newcommand{\Null}{\mathsf{null}}

\newcommand{\Children}{[2]}

\makeatletter
\newcommand{\removelatexerror}{\let\@latex@error\@gobble}
\makeatother

\renewcommand{\baselinestretch}{1}
\title{Algorithms for Algebraic Path Properties in \\ Concurrent Systems 
of Constant Treewidth Components
}

\authorinfo{Krishnendu Chatterjee \and Amir Kafshdar Goharshady \and Rasmus Ibsen-Jensen\and Andreas Pavlogiannis}
           {IST Austria (Institute of Science and Technology Austria) Klosterneuburg, Austria}
           {\{krish.chat, goharshady, ribsen, pavlogiannis\}@ist.ac.at}
           
\date{}

\begin{document}

\maketitle

\begin{abstract}
We study algorithmic questions for concurrent systems where the 
transitions are labeled from a complete, closed semiring, and 
path properties are algebraic with semiring operations.
The algebraic path properties can model dataflow analysis problems, 
the shortest path problem, and many other natural problems that arise
in program analysis.
We consider that each component of the concurrent system is a graph 
with constant treewidth, a property satisfied by the controlflow graphs of most programs.
We allow for multiple possible queries, which arise naturally in demand 
driven dataflow analysis.
The study of multiple queries allows us to consider the tradeoff 
between the resource usage of the \emph{one-time} preprocessing and
for \emph{each individual} query.
The traditional approach constructs the product graph of all components and 
applies the best-known graph algorithm on the product. 
In this approach, even the answer to a single query requires the 
transitive closure (i.e., the results of all possible queries), 
which provides no room for tradeoff between preprocessing and query time.

Our main contributions are algorithms that significantly improve the
worst-case running time of the traditional approach, and provide various
tradeoffs depending on the number of queries.
For example, in a concurrent system of two components, the traditional approach
requires hexic time in the worst case for answering one query as well as computing the 
transitive closure, whereas we show that with one-time preprocessing in almost cubic time, 
each subsequent query can be answered in at most linear time, 
and even the transitive closure can be computed in almost quartic time.
Furthermore, we establish conditional optimality results showing that
the worst-case running time of our algorithms cannot be improved without 
achieving major breakthroughs in graph algorithms (i.e., improving 
the worst-case bound for the shortest path problem in general graphs).
Preliminary experimental results show that our algorithms perform
favorably on several real-world benchmarks.
\end{abstract}

%%\category{D.3.4}{Programming Languages}{Processors---Optimization}

%\terms Algorithms, Languages

% general terms are not compulsory anymore, 
% you may leave them out
%%\terms Algorithms, Languages


%%\category{CR-number}{subcategory}{third-level}

% general terms are not compulsory anymore, 
% you may leave them out
%%\terms
%%term1, term2

%\keywords
%%keyword1, keyword2

\section{Introduction\label{sec:intro}}

In this work we consider concurrent finite-state systems where each 
component is a constant-treewidth graph, 
and the algorithmic question is to determine algebraic path properties between pairs of nodes in the system.
Our main contributions are algorithms which significantly improve
the worst-case running time of the existing algorithms.
We establish conditional optimality results for some of our algorithms 
in the sense that they cannot be improved without achieving major 
breakthroughs in the algorithmic study of graph problems.
Finally, we provide a prototype implementation of our algorithms which 
significantly outperforms the existing algorithmic methods on several benchmarks.
%%well in practice. 

\smallskip\noindent{\em Concurrency and algorithmic approaches.}
The analysis of concurrent systems is one of the fundamental problems in 
computer science in general, and programming languages in particular.
A finite-state concurrent system consists of several components, 
each of which is a finite-state graph, and the whole system is a composition 
of the components.
Since errors in concurrent systems are hard to reproduce by simulations 
due to combinatorial explosion in the number of interleavings, 
formal methods are necessary to analyze such systems. 
In the heart of the formal approaches are graph algorithms, which 
provide the basic search procedures for the problem.
The basic graph algorithmic approach is to construct the product graph 
(i.e., the product of the component systems) and then apply the 
best-known graph algorithms on the product graph.
While there are many practical approaches for the analysis of concurrent 
systems, a fundamental theoretical question is whether special properties of 
graphs that arise in analysis of programs can be exploited to develop asymptotically faster 
algorithms as compared to the basic approach.

\smallskip\noindent{\em Special graph properties for programs.}
A very well-studied notion in graph theory is the concept of {\em treewidth} 
of a graph, which is a measure of how similar a graph is to a tree 
(a graph has treewidth~1 precisely if it is a tree)~\cite{Robertson84}. 
The treewidth of a graph is defined based on a {\em tree decomposition} of 
the graph~\cite{Halin76}, see Section~\ref{sec:definitions} for a formal 
definition.
On one hand the treewidth property provides a mathematically elegant way 
to study graphs, and on the other hand there are many classes of graphs which 
arise in practice and have constant treewidth. 
The most important example is that the controlflow graph for goto-free 
programs for many programming languages are of constant 
treewidth~\cite{Thorup98}, and it was also shown in~\cite{Gustedt02} that 
typically all Java programs have constant treewidth. 

%%which has also been demonstrated experimentally even at the presence of gotos~
%%(\cite{Gustedt02,Burgstaller04}).

\smallskip\noindent{\em Algebraic path properties.}
To specify properties of traces of concurrent systems we consider a very 
general framework, where edges of the system are labeled from a complete, 
closed semiring (which subsumes bounded and finite distributive semirings), 
and we refer to the labels of the edges as weights. 
For a given path, the weight of the path is the semiring product of the weights
on the edges of the path, and the weights of different paths are combined using the
semiring plus operator. 
For example, (i)~the Boolean semiring (with semiring product as AND, and 
semiring plus as OR) expresses the reachability property; 
(ii)~the tropical semiring (with real numbers as edge weights, semiring product as 
standard sum, and semiring plus as minimum) expresses the shortest path 
property; and (iii)~with letter labels on edges, semiring product as string 
concatenation and semiring plus as union we can express the regular expression 
of reaching from one node to another.
The algebraic path properties subsumes the dataflow 
analysis of the IFDS/IDE frameworks~\cite{Reps95,Sagiv96} in the 
intraprocedural setting,
which consider compositions of distributive dataflow functions, and meet-over-all-paths 
as the semiring plus operator.
Since IFDS/IDE is a special case of our framework, a large and important class of 
dataflow analysis problems that can be expressed in IFDS/IDE
can also be expressed in our framework. 
However, the IFDS/IDE framework works for sequential interprocedural analysis,
whereas we focus on intraprocedural analysis, but in the concurrent setting.

\smallskip\noindent{\em Expressiveness of algebraic path properties.}
The algebraic path properties provide an expressive framework with rich modeling power.
%A well-studied case is that of dataflow analysis of distributive flow functions~\cite{Reps95, Sagiv96, Farzan07, Chugh08, CIPG15}.
Here we elaborate on three important classes.

\begin{compactenum}
\item {\em Weighted shortest path.} 
The algebraic paths framework subsumes several problems on weighted graphs.
The most well-known such problem is the shortest path problem~\cite{Floyd62, Warshall62, Bellman58, Ford56, Johnson77},
phrased on the tropical semiring.
For example, the edge weights (positive and negative) can express energy 
consumptions, and the shortest path problem asks for the least energy 
consuming path.
Another important quantitative property is the \emph{mean-payoff} property, 
where each edge weight represents a reward or cost, and the problem asks for 
a path that minimizes the average of the weights along a path. 
Many quantitative properties of relevance for program analysis (e.g., to 
express performance or resource consumption) can be modeled as mean-payoff 
properties~\cite{Chatterjee15B,CHR13}.
The mean-payoff and other fundamental problems on weighted graphs (e.g., the most probable path and the minimum initial 
credit problem) can be reduced to the 
shortest-path problem~\cite{Viterbi67, L76, Karp78, Bouyer08, CDH10, CHR13, W08, Chatterjee15C}.

\item {\em Dataflow problems.} 
A wide range of dataflow problems has an algebraic paths formulation, expressed as a ``meet-over-all-paths'' analysis~\cite{Kildall73}.
Perhaps the most well-known case is that of distributive flow functions considered in the IFDS framework~\cite{Reps95, Sagiv96}.
Given a finite domain $D$ and a universe $F$ of distributive dataflow functions $f:2^D\rightarrow 2^D$,
a weight function $\Weight:E\rightarrow F$ associates each edge of the controlflow graph with a flow function. 
The weight of a path is then defined as the composition of the flow functions along its edges,
and the dataflow distance between two nodes $u$, $v$ is the meet $\sqcap$ (union or intersection) of the
weights of all $u\Path v$ paths. 
The problem can be formulated on the meet-composition semiring $(F, \sqcap, \circ, \emptyset,  I)$, where $I$ is the identity function. 
We note, however, that the IFDS/IDE framework considers interprocedural paths in sequential programs.
In contrast, the current work focuses on intraprocedural analysis of concurrent programs.
The dataflow analysis of concurrent programs has been a problem of intensive study (e.g.~\cite{Grunwald93, Knoop96, Farzan07, Chugh08, Kahlon09, De11}), where (part of) the underlying analysis is based on an algebraic, ``meet-over-all-paths'' approach.

\item {\em Regular expressions.}
Consider the case that each edge is annotated with an observation or action.
Then the regular expression to reach from one node to another represents all the 
sequences of observable actions that lead from the start node to the target.
The regular languages of observable actions have provided useful formulations
in the analysis and synthesis of concurrent systems~\cite{Dwyer04, Farzan13,Cerny15}.
Regular expressions have also been used as algebraic relaxations of 
interprocedurally valid paths in sequential and concurrent systems~\cite{Yan11, Bouajjani03}.
%For example, the regular expression can represent a set of error traces
%(instead of a single error trace), and can help in accelerating the 
%abstraction-refinement loop by ruling out several traces instead of a 
%single one~\cite{Cerny15}.
%Hence computing regular expression is another important algebraic path 
%property.
\end{compactenum}

%For a detailed discussion, see~\cite{TR}.
For a detailed discussion, see \cref{sec:example}.

\smallskip\noindent{\em The algorithmic problem.}
In graph theoretic parlance, graph algorithms typically consider two types of 
queries: (i)~a \emph{pair query} given nodes $u$ and $v$ 
(called $(u,v)$-pair query) asks for the algebraic path property 
from $u$ to $v$; and 
(ii)~a \emph{single-source} query given a node $u$ asks for the 
answer of $(u,v)$-pair queries for all nodes $v$.
In the context of concurrency, in addition to the classical pair 
and single-source queries, we also consider {\em partial} queries. 
Given a concurrent system with $k$ components, a node in the product graph 
is a tuple of $k$ component nodes. 
A \emph{partial} node $\PNode{u}$ in the product only specifies nodes of a nonempty 
strict subset of all the components.
Our work also considers partial pair and partial single-source queries, 
where the input nodes are partial nodes.
Queries on partial nodes are very natural, as they capture 
properties between local locations in a component, 
that are shaped by global paths in the whole concurrent system.
For example, constant propagation and dead code elimination are local properties in a program,
but their analysis requires analyzing the concurrent system as a whole.

%Say that most of these works offer no theoretical improvement on the on-demand queries
\smallskip\noindent{\em Preprocess vs query.}
A topic of widespread interest in the programming languages community is that of on-demand analysis
~\cite{WayneA78,Zadeck84,Horwitz95,Duesterwald95,Reps95b,Sagiv96,Reps97,Yuan97,Naeem10,CIPG15}.
Such analysis has several advantages, such as (quoting from~\cite{Horwitz95,Reps97})
(i)~narrowing down the focus to specific points of interest,
(ii)~narrowing down the focus to specific dataflow facts of interest,
(iii)~reducing work in preliminary phases,
(iv)~sidestepping incremental updating problems, and
(v)~offering demand analysis as a user-level operation.
%\begin{compactenum}
%\item Narrowing down the focus to specific points of interest
%\item Narrowing down the focus to specific dataflow facts of interest
%\item Reducing work in preliminary phases
%\item Sidestepping incremental updating problems
%\item Offering demand analysis as a user-level operation
%\end{compactenum}
For example, in alias analysis, the question is whether two pointers may point to the same 
object, which is by definition modeled as a question between a pair of nodes.
Similarly, in constant propagation a relevant question is whether 
some variable remains constant between a pair of controlflow locations.
The problem of on-demand analysis allows us to distinguish between a single preprocessing phase 
(one time computation), and a subsequent query phase, where queries are answered on demand.
The two extremes of the preprocessing and query phase are:
(i)~\emph{complete preprocessing} (aka \emph{transitive closure} computation) 
where the result is precomputed for every possible query, and hence queries are 
answered by simple table lookup; and (ii)~\emph{no preprocessing} where 
every query requires a new computation.
However, in general, there can be a tradeoff between the preprocessing and 
query computation.
%Write that on demand offer practical guarantee, and the theoretical is only for all queries.
Most of the existing works for on-demand analysis do not make a formal distinction between preprocessing and query phases,
as the provided complexities only guarantee the \emph{same worst-case 
complexity} property, namely that the total time for handling any sequence of queries 
is no worse than the complete preprocessing.
Hence most existing tradeoffs are practical, without any theoretical guarantees.

\begin{table*}[!h]
\centering
\begin{ThreePartTable}
\small
\renewcommand{\arraystretch}{1.2}
\begin{tabular}{|c||c|c||c|c|c|c|}
\hline
& \multicolumn{2}{c||}{\textbf{Preprocess}} & \multicolumn{4}{c|}{\textbf{Query time}}\\
\cline{2-7}
%\tnotex{tn:space}
& Time & Space  & Single-source & Pair & Partial single-source & Partial pair\\
\hline
\hline
Previous results~\cite{Lehmann77,Floyd62,Warshall62,Kleene56} & $O(n^6)$ & $O(n^4)$ & $O(n^2)$ & $O(1)$ & $O(n^2)$ & $O(1)$\\
\hline
%\hline
%No preprocess~\cite{Lehmann77,Floyd62,Warshall62,Kleene56} & $-$ & $O(n^4)$ & $O(n^6)$ & $O(n^6)$ & $O(n^6)$\\
{\bf Our result \cref{cor:ld}}~$(\eps>0)$ & $\mathbf{O(n^{3})}$ & $\mathbf{O(n^{2+\eps})}$ & $\mathbf{O(n^{2+\eps})}$ & $\mathbf{O(n^{2})}$ & $\mathbf{O(n^{2+\eps})}$ & $\mathbf{O(n^{2})}$\\
\hline
{\bf Our result \cref{them:concurrent}}~$(\eps>0)$  & $\mathbf{O(n^{3+\eps})}$ & $\mathbf{O(n^3)}$ & $\mathbf{O(n^{2+\eps})}$ & $\mathbf{O(n)}$ & $\mathbf{O(n^2)}$ & $\mathbf{O(1)}$\\
\hline
{\bf Our result \cref{cor:closure}}~$(\eps>0)$  & $\mathbf{O(n^{4+\eps})}$ & $\mathbf{O(n^4)}$ & $\mathbf{O(n^{2})}$ & $\mathbf{O(1)}$ & $\mathbf{O(n^2)}$ & $\mathbf{O(1)}$\\
%{\bf \cref{them:tradeoff}}~$(\eps\in [0,\Epsbound])$ & $\mathbf{O(n^{3+\eps})}$ & $\mathbf{O(n^{2.4})}$ & $\mathbf{O(n^{2.4})}$ & $\mathbf{O(n^{2-\eps})}$ & $\mathbf{O(n^{2-3\cdot \eps})}$\\
\hline
\end{tabular}
\caption{The algorithmic complexity for computing algebraic path queries wrt a closed, complete semiring on a concurrent graph 
$G$ which is the product of two constant-treewidth graphs $G_1$, $G_2$, with $n$ nodes each.}
\label{tab:intro1}
%\begin{tablenotes}
%\item[a] \label{tn:space} Space complexity is measured as the total space usage in preprocessing and one query.
%\end{tablenotes}
\end{ThreePartTable}
\end{table*}

\smallskip\noindent{\em Previous results.} 
In this work we consider finite-state concurrent systems, where each 
component graph has constant treewidth, and the trace properties are specified as algebraic path properties. 
Our framework can model a large class of problems: typically the 
controlflow graphs of programs have constant treewidth~\cite{Thorup98,Gustedt02,Burgstaller04}, 
and if there is a constant number of synchronization variables with constant-size domains, 
then each component graph has constant treewidth. 
Note that this imposes little practical restrictions, as typically 
synchronization variables, such as locks, mutexes and condition variables 
have small (even binary) domains (e.g. locked/unlocked state).
The best-known graph algorithm for the algebraic path property problem is the 
classical Warshall-Floyd-Kleene~\cite{Lehmann77,Floyd62,Warshall62,Kleene56} style dynamic programming, which requires cubic time.
Two well-known special cases of the algebraic paths problem are 
(i)~computing the shortest path from a source to a target node in a weighted 
graph, and (ii)~computing the regular expression from a source to a target 
node in an automaton whose edges are labeled with letters from a 
finite alphabet.
In the first case, the best-known algorithm is the Bellman-Ford algorithm with
time complexity $O(n\cdot  m)$.
%(where $n$ is the number of vertices and $m$ is
%the number of edges, and since the number of edges in the concurrent system is upper bounded by the
%product of the number of edges of the components,
% $m$ can be $n$, the worst-case bound is $O(n^2)$).
In the second case, the well-known construction of Kleene's~\cite{Kleene56} theorem requires cubic time.
The only existing algorithmic approach for the problem we consider is to first 
construct the product graph (thus if each component graph has size $n$, and
there are $k$ components, then the product graph has size $O(n^k)$), and then 
apply the best-known graph algorithm (thus the overall time complexity is 
$O(n^{3\cdot k})$ for algebraic path properties).
Hence for the important special case of two components we obtain a hexic-time (i.e., $O(n^6)$) algorithm.
Moreover, for algebraic path properties the current best-known algorithms for one pair query (or one 
single-source query) computes the entire transitive closure. 
Hence the existing approach does not allow a tradeoff of preprocessing and
query as even for one query the entire transitive closure is computed.

\smallskip\noindent{\em Our contributions.} 
Our main contributions are improved algorithmic upper bounds,
proving several optimality results of our algorithms, and experimental 
results.
Below all the complexity measures (time and space) are in the number of basic 
machine operations and number of semiring operations. 
We elaborate our contributions below.

\begin{compactenum}
\item {\em Improved upper bounds.}
%%The results are stated for a data-structure that operates on input
%%of a concurrent graph $G$ that consists of $k$ components.
%%The data-structure has a preprocessing phase, follows by a query phase
%%in which it answers algebraic path queries.
We present improved upper bounds both for general $k$ components,
and the important special case of two components.

\begin{compactitem}
\item \emph{General case.} 
We show that for $k\geq 3$ components with $n$ nodes each, after 
$O(n^{3\cdot (k-1)})$ preprocessing time, we can answer 
(i)~single-source queries in $O(n^{2\cdot (k-1)})$ time, 
(ii)~pair queries in $O(n^{k-1})$ time, 
(iii)~partial single-source queries in $O(n^k)$ time, and 
(iv)~partial pair queries in $O(1)$ time; 
while using at all times $O(n^{2\cdot k -1})$ space.
In contrast, the existing methods~\cite{Lehmann77,Floyd62,Warshall62,Kleene56}
compute the transitive closure even for a single query, and thus require 
$O(n^{3\cdot k})$ time and $O(n^{2\cdot k})$ space.
 
\item \emph{Two components.} 
For the important case of two components, the existing methods 
require $O(n^6)$ time and $O(n^4)$ space even for one query.
In contrast, we establish a variety of tradeoffs between preprocessing and 
query times, and the best choice depends on the number of expected queries.
In particular, for any fixed $\eps>0$, we establish the following three results.

\smallskip\noindent{\em Three results.}
First, we show (\cref{cor:ld}) that with 
$O(n^3)$ preprocessing time and using $O(n^{2+\eps})$ space, we can answer
single-source queries in $O(n^{2+\eps})$ time, and 
pair and partial pair queries require $O(n^2)$ time.
Second, we show (\cref{them:concurrent}) that 
with $O(n^{3+\eps})$ preprocessing time and using $O(n^3)$ space,
we can answer pair and partial pair queries in time 
$O(n)$ and $O(1)$, respectively.
Third, we show (\cref{cor:closure}) that 
the transitive closure can be computed using $O(n^{4+\eps})$ preprocessing time 
and $O(n^{4})$ space, after which single-source queries require $O(n^2)$ time, 
and pair and partial pair queries require $O(1)$ time
(i.e., all queries require linear time in the size of the output).

\smallskip\noindent{\em Tradeoffs.}
Our results provide various tradeoffs:
The first result is best for answering $O(n^{1+\eps})$ pair and partial pair 
queries;
the second result is best for answering between $\Omega(n^{1+\eps})$ and  
$O(n^{3+\eps})$ pair queries, and $\Omega(n^{1+\eps})$ partial pair queries; 
and the third result is best when answering $\Omega(n^{3+\eps})$ pair queries.
Observe that the transitive closure computation is preferred when the number 
of queries is large, in sharp contrast to the existing methods that compute 
the transitive closure even for a single query.
Our results are summarized in~\cref{tab:intro1} and the tradeoffs are 
pictorially illustrated in~\cref{fig:outline}.
\end{compactitem}

\begin{figure}
%\usepgfplotslibrary{fillbetween}
\centering

\newcommand{\pgfplotsdrawaxis}{\pgfplots@draw@axis}
\pgfplotsset{axis line on top/.style={
  axis line style=transparent,
  ticklabel style=transparent,
  tick style=transparent,
  axis on top=false,
  after end axis/.append code={
    \pgfplotsset{axis line style=opaque,
      ticklabel style=opaque,
      tick style=opaque,
      grid=none}
    \pgfplotsdrawaxis}
  }
}

\begin{tikzpicture}[scale=1]

\begin{axis}[
unit vector ratio*=1 0.875 1,
axis background/.style={fill=red!20},
axis lines=left,
ticks=none,
compat=newest,
xmin=0,xmax=8,
ymin=0,ymax=8
]
\addplot[fill=green!20, draw=green!20] coordinates{
(0,8)
(5.6,8)
(5.6,0)
(0,0)
}\closedcycle;
\addplot[color=black, very thick,] coordinates {
(5.6,0)
(5.6,8)
};
\addplot [color=black, very thick, fill=blue!20] coordinates {
(0,3)
(3,3)
(3,0)
}\closedcycle;
\end{axis}

\begin{axis}[
unit vector ratio*=1 0.875 1,
xlabel=$i$ pair queries,
ylabel=$j$ partial pair queries,
axis lines=left,
ultra thick,
compat=newest,
xmin=0,xmax=8,
ymin=0,ymax=8,
xtick={3,  5.6, 7.8},
xticklabels={$n^{1+\eps}$, $n^{3+\eps}$, $n^4$},
ytick={3, 7.8},
yticklabels={$n^{1+\eps}$, $n^3$},
]
\node[align=center] at (axis cs:1.5,1.5) {$n^3+(i+j)\cdot n^2$\\ ~\\ \cref{cor:ld}};
\node[align=center] at (axis cs:3,5) {$n^{3+\eps}+i\cdot n+j$\\ ~\\ \cref{them:concurrent}};
\node[align=center] at (axis cs:6.8,4) {$n^{4+\eps}+i+j$\\ ~\\ \cref{cor:closure}};
\end{axis}

\end{tikzpicture}
\caption{
Given a concurrent graph $G$ of two constant-treewidth graphs of $n$ nodes each,
the figure illustrates the time required by the variants of our algorithms to preprocess $G$, 
and then answer $i$ pair queries and $j$ partial pair queries.
The different regions correspond to the best variant for handling different number of such queries.
In contrast, the current best solution requires $O(n^6+i+j)$ time.
For ease of presentation we omit the $O(\cdot)$ notation.
}\label{fig:outline}
\end{figure}
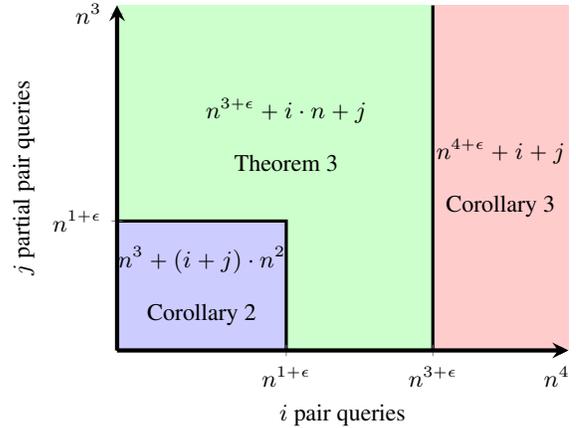

\item {\em Optimality of our results.}
Given our significant improvements for the case of two components, a very 
natural question is whether the algorithms can be improved further. 
While presenting matching  bounds for polynomial-time graph algorithms 
to establish optimality is very rare in the whole of computer science, 
we present \emph{conditional lower bounds} which show that our combined
preprocessing and query time cannot be improved without
achieving a major breakthrough in graph algorithms.

\begin{compactitem}

\item {\em Almost optimality.} 
First, note that in the first result (obtained from~\cref{cor:ld})
our space usage and single-source query time are arbitrarily close to 
optimal, as both the input and the output have size $\Theta(n^2)$.
Moreover, the result is achieved with preprocessing time less than 
$\Omega(n^4)$, which is a lower bound for computing the transitive closure 
(which has $n^4$ entries).
Furthermore, in our third result (obtained from~\cref{cor:closure}) 
the $O(n^{4 + \eps})$ preprocessing time is arbitrarily close to optimal, and the $O(n^4)$ preprocessing 
space is indeed optimal, 
as the transitive closure computes the distance among all $n^4$ pairs of 
nodes (which requires $\Omega(n^4)$ time and space).

\item {\em Conditional lower bound.}
In recent years, the conditional lower bound problem has received
vast attention in complexity theory, where under the assumption that 
certain problems (such as matrix multiplication, all-pairs shortest path)
cannot be solved faster than the existing upper bounds, lower bounds for other problems (such as 
dynamic graph algorithms) are obtained~\cite{AW14,AWY15,HKNS15}.
The current best-known algorithm for algebraic path properties for general 
(not constant-treewidth) graphs is cubic in the number of nodes.
Even for the special case of shortest paths with positive and negative 
weights, the best-known algorithm (which has not been improved over five 
decades) is $O(n\cdot  m)$, where $m$ is the number of edges. 
Since $m$ can be $\Omega(n^2)$, the current best-known worst-case complexity 
is cubic in the number of nodes.
We prove that pair queries require more time in a concurrent graph of two 
constant-treewidth graphs, with $n$ nodes each, than in general graphs with $n$ nodes.
This implies that improving the $O(n^3)$ combined preprocessing and query time over our result (from~\cref{cor:ld}) for answering $r$ queries, for $r= O(n)$, 
would yield the same improvement over the $O(n^3)$ time for answering $r$ pair queries in general graphs. 
%Note that improving the $O(n^3)$ time for general graphs, even for a single 
%pair query, would be a major breakthrough. 
That is, the combination of our preprocessing and query time 
(from~\cref{cor:ld}) cannot be improved without equal improvement on the long standing 
cubic bound for the shortest path and the algebraic 
path problems in general graphs.
Additionally, our result (from \cref{them:concurrent}) 
cannot be improved much further even for $n^2$ queries,
as the combined time for preprocessing and answering $n^2$ 
queries is $O(n^{3+\epsilon})$ using \cref{them:concurrent}, while
the existing bound is $O(n^3)$ for general graphs. 
%Note that in general graphs with $n$ nodes there are at most $n^2$ distinct 
%pair queries. 
\end{compactitem}

\item {\em Experimental results.}
We provide a prototype implementation of our algorithms which significantly 
outperforms the baseline methods on several benchmarks.

\end{compactenum}

\smallskip\noindent{\em Technical contributions.}
The results of this paper rely on several novel technical contributions.
\begin{compactenum}

\item \emph{Upper bounds}.
Our upper bounds depend on a series of technical results.
\begin{compactenum}
\item 
The first key result is an algorithm 
for constructing a \emph{strongly balanced} tree-decomposition $T$.
A tree is called $(\beta, \gamma)$-balanced if for every node $u$ and descendant $v$ of $u$
that appears $\gamma$ levels below,
the size of the subtree of $T$ rooted at $v$ is at most a $\beta$ fraction of the size
of the subtree of $T$ rooted at $u$.
For any fixed $\delta>0$ and $\lambda\in \Nats$ with $\lambda\geq 2$,
let $\beta=\Balfactor$ and $\gamma=\Levelfactor$.
We show that a $(\beta,\gamma)$-balanced tree decomposition of a constant-treewidth graph 
with $n$ nodes can be constructed in $O(n\cdot \log n)$ time and $O(n)$ space.
To our knowledge, this is the first algorithm that constructs a tree decomposition 
with such a strong notion of balance.
This  property is crucial for achieving the resource bounds of our algorithms
for algebraic paths.
The construction is presented in \cref{sec:prelim_tree_dec}.
\item 
Given a concurrent graph $G$ obtained from $k$ constant-treewidth graphs $G_i$, 
we show how a tree-decomposition of $G$ can be constructed from the strongly balanced 
tree-decompositions $T_i$ of the components $G_i$, in time that is linear in the size of the output.
We note that $G$ can have large treewidth, and thus determining the treewidth of $G$ can be computationally 
expensive.
Instead, our construction avoids computing the treewidth of $G$, and 
directly constructs a tree-decomposition of $G$ from the strongly balanced tree decompositions $T_i$.
The construction is presented in \cref{sec:product_tree_dec}.
\item
Given the above tree-decomposition algorithm for concurrent graphs $G$,
in \cref{sec:concurrent} we present the algorithms for handling algebraic path queries.
In particular, we introduce the \emph{partial expansion} $\PNode{G}$ of $G$ for additionally handling partial queries,
and describe the algorithms for preprocessing and querying $\PNode{G}$ in the claimed time and space bounds.
\end{compactenum}
\item \emph{Lower bound}.
Given an arbitrary graph $G$ (not of constant treewidth) of $n$ nodes, we show how to construct
a constant-treewidth graph $G''$ of $2\cdot n$ nodes, and a graph $G'$ that is 
the product of $G''$ with itself, such that algebraic path queries in $G$ coincide with 
such queries in $G'$.
This construction requires quadratic time on $n$.
The conditional optimality of our algorithms follows, as improvement over our algorithms 
must achieve the same improvement for algebraic path properties on arbitrary graphs.
\end{compactenum}
All our algorithms are simple to implement and provided as pseudocode 
in \cref{sec:algorithms}. Several technical proofs are also relegated 
to the full version due to lack of space.

\subsection{Related Works}

\smallskip\noindent{\em Treewidth of graphs.}
The notion of treewidth of graphs as an elegant mathematical tool
to analyze graphs was introduced in~\cite{Robertson84}.
The significance of constant treewidth in graph theory is large 
mainly because several problems on graphs become complexity-wise easier.
Given a tree decomposition of a graph with low treewidth $t$, 
many NP-complete problems for arbitrary graphs can be solved in time polynomial 
in the size of the graph, but exponential in $t$~\cite{Arnborg89,Bern1987216,Bodlaender88,Bodlaender93,Bodlaender05}. 
Even for problems that can be solved in polynomial time, 
faster algorithms can be obtained for low treewidth graphs,
e.g., for the distance problem~\cite{Chaudhuri95}.
The constant-treewidth property of graphs has also been used in the context
of logic: Monadic Second Order (MSO) logic is a very expressive logic, and a 
celebrated result of~\cite{Courcelle91} showed that for constant-treewidth graphs the 
decision questions for MSO can be solved in polynomial time; and 
the result of~\cite{Elberfeld10} shows that this can even be achieved in deterministic
log-space.
Various other models (such as probabilistic models of Markov decision processes
and games played on graphs for synthesis) with the constant-treewidth restriction 
have also been considered~\cite{CL13,Obdrzalek03}.
The problem of computing a balanced tree decomposition for a constant treewidth 
graph was considered in~\cite{Reed92}.
More importantly, in the context of programming languages, it was shown in~\cite{Thorup98}
that the controlflow graphs of goto-free programs in many programming languages have 
constant treewidth.
This theoretical result was subsequently followed up in several practical 
approaches, and although in the presence of gotos the treewidth is not 
guaranteed to be bounded, it has been shown that programs in several 
programming languages have typically low treewidth 
 ~\cite{Gustedt02,Burgstaller04}.
The constant-treewidth property of graphs has been used to develop faster 
algorithms for sequential interprocedural analysis~\cite{CIPG15},
and on the analysis of automata with auxiliary storage (e.g., stacks and queues)~\cite{Madhusudan11}.
These results have been followed in practice, and some compilers (e.g., SDCC) implement tree-decomposition-based algorithms
for performance optimizations~\cite{Krause13}.
%The techniques of~\cite{CIPG15} do not extend to the concurrent setting because 
%even though the component graphs have constant treewidth, the graph which is 
%the product of the components has treewidth that depends on the number of nodes.
%Thus none of the previous works exploit the constant-treewidth property
%for faster algorithms in the concurrent setting. 

\smallskip\noindent{\em Concurrent system analysis.}
The problem of concurrent system analysis has been considered in several 
works, both for intraprocedural as well context-bounded interprocedural
analysis~\cite{Harel97,Alur99,Farzan13,Qadeer05,Bouajjani05,LaTorre08, Lal08,Lal09, Kahlon13}, 
and many practical tools have been developed as well~\cite{Qadeer05, Lal09,Suwimonteerabuth08,Lal12}.
In this work we focus on the intraprocedural analysis with constant-treewidth
graphs, and present algorithms with better asymptotic complexity.
To our knowledge, none of the previous works consider the constant-treewidth property, nor do they
improve the asymptotic complexity of the basic algorithm for the algebraic 
path property problem.

\section{Definitions}\label{sec:definitions}
In this section we present definitions related to semirings, 
graphs, concurrent graphs, and tree decompositions.
We start with some basic notation on sets and sequences.

\smallskip\noindent{\bf Notation on sets and sequences.}
Given a number $r\in \Nats$, we denote by $[r]=\{1,2,\dots, r\}$ the natural numbers from $1$ to $r$.
Given a set $X$ and a $k\in \Nats$, we denote by $X^k=\prod_{i=1}^kX$, the
$k$ times Cartesian product of $X$.
A sequence $x_1,\dots x_k$ is denoted for short by $(x_i)_{1\leq i\leq k}$,
or $(x_i)_i$ when $k$ is implied from the context.
Given a sequence $Y$, we denote by $y\in Y$ the fact that $y$ appears in $Y$.

\subsection{Complete, closed semirings}

\begin{definition}[Complete, closed semirings]
\normalfont
We fix a complete {\em semiring} $S=(\Sigma, \oplus, \otimes, \Zero, \One)$ where $\Sigma$ is a countable set, 
$\oplus$ and $\otimes$ are binary operators on $\Sigma$, and $\Zero, \One\in \Sigma$, and the following properties hold:
\begin{compactenum}
\item $\oplus$ is infinitely associative, commutative, and $\Zero$ is the neutral element,
\item $\otimes$ is associative, and $\One$ is the neutral element,
\item $\otimes$ infinitely distributes over $\oplus$,
%\item $\oplus$ is infinitely associative,
%\item $\otimes$ infinitely distributes over $\oplus$,
\item $\Zero$ absorbs in multiplication, i.e., $\forall a\in \Sigma: a\otimes \Zero=\Zero$.
\end{compactenum}
Additionally, we consider that $S$ is equipped with a \emph{closure} operator $^*$,
such that $\forall s\in \Sigma:~s^*=\One\oplus (s\otimes s^*) =\One\oplus (s^*\otimes s)$
(i.e., the semiring is {\em closed}).
\end{definition}

\subsection{Graphs and tree decompositions}

\smallskip\noindent{\bf Graphs and weighted paths.}
Let $G=(V,E)$ be a weighted finite directed graph (henceforth called simply a graph) where $V$ is a set of $n$ nodes and
$E\subseteq V\times V$ is an edge relation, along with a weight
function $\Weight:E\rightarrow \Sigma$ that assigns to each edge of $G$ an element from $\Sigma$. 
Given a set of nodes $X\subseteq V$, we denote by $G\restr{X}=(X, E\cap (X\times X))$
the subgraph of $G$ induced by $X$.
A path $P:u\Path v$ is a sequence of nodes $(x_1,\dots, x_k)$
such that $x_1=u$, $x_k=v$, and for all $1\leq i < k$ we have $(x_i,x_{i+1})\in E$.
The length of $P$ is $|P|=k-1$, and a single node is itself a $0$-length path.
A path $P$ is \emph{simple} if no node repeats in the path (i.e., it does not
contain a cycle). 
Given a path $P=(x_1,\dots , x_k)$, the weight of $P$ is $\otimes(P)=\bigotimes(\Weight(x_i,x_{i+1}))_i$ if 
$|P|\geq 1$ else $\otimes(P)=\One$. Given nodes $u,v\in V$, the {\em semiring distance} $\Distance(u,v)$
is defined as $\Distance(u,v)=\bigoplus_{P:u\Path v}\otimes(P)$, and $\Distance(u,v)=\Zero$
if no such $P$ exists.

\smallskip\noindent{\bf Trees.}
A \emph{tree} $T=(V,E)$ is an undirected graph with a root node $u_0$, 
such that between every two nodes there is a unique simple path. 
For a node $u$ we denote by $\Level(u)$ the \emph{level} of $u$ which is 
defined as the length of the simple path from $u_0$ to $u$.
A \emph{child} of a node $u$ is a node $v$ such that $\Level(v)= \Level(u) +1$
and $(u,v) \in E$, and then $u$ is the \emph{parent} of $v$.
For a node $u$, any node (including $u$ itself) that appears in the path 
from $u_0$ to $u$ is an \emph{ancestor} of $u$, and if $v$ is an ancestor
of $u$, then $u$ is a \emph{descendant} of $v$. 
Given two nodes $u,v$, the \emph{lowest common ancestor (LCA)} is the common 
ancestor of $u$ and $v$ with the highest level.
Given a tree $T$, a \emph{contiguous subtree} is subgraph $(X,E')$ of $T$ 
such that $E'=E \cap (X\times X)$ and for every pair $u,v \in X$, 
every node that appears in the unique path from $u$ to $v$ belongs to $X$.
A tree is {\em $k$-ary} if every node has at most $k$-children 
(e.g., a binary tree has at most two children for every node).
In a {\em full $k$-ary} tree, every node that $0$ or $k$-children.

\smallskip\noindent{\bf Tree decompositions.}
A \emph{tree-decomposition} $\Tree(G)=T=(V_T, E_T)$ of a graph $G$ is a tree,
where every node $B_i$ in $T$ is a subset of nodes of $G$ such that  
the following conditions hold:
\begin{compactenum}
\item[C1] $V_T=\{\Bag_0,\dots, \Bag_{b}\}$ with $\Bag_i\subseteq V$, and $\bigcup_{\Bag_i\in V_T}\Bag_i=V$ (every node is covered).
\item[C2] For all $(u,v)\in E$ there exists $\Bag_i\in V_T$ such that $u,v\in \Bag_i$ (every edge is covered).
\item[C3] For all $i,j,k$ such that there is a bag $\Bag_k$ that appears in the simple path $\Bag_i\Path \Bag_j$ in $\Tree(G)$, 
we have $\Bag_i\cap \Bag_j\subseteq \Bag_k$ (every node appears in a contiguous subtree of $T$).
\end{compactenum}
The sets $\Bag_i$ which are nodes in $V_T$ are called \emph{bags}.
We denote by $|T|=|V_T|$ the number of bags in $T$.
Conventionally, we call $\Bag_0$ the root of $T$, and denote by $\Level(\Bag_i)$
the level of $\Bag_i$ in $\Tree(G)$. 
%%defined as the length of the unique simple path from $\Bag_0$ to $\Bag_i$.
For a bag $\Bag$ of $T$, we denote by $T(\Bag)$ the subtree of $T$ rooted at $\Bag$.
A bag $\Bag$ is called the \emph{root bag} of a node $u$ if $u\in \Bag$
and every $\Bag'$ that contains $u$ appears in $T(\Bag)$.
We often use $\Bag_u$ to refer to the root bag of $u$,
and define $\Level(u)=\Level(\Bag_u)$.
%In this work we assume w.l.o.g. that every bag $\Bag$ of a tree-decomposition $T$ is the root bag of eat least one node
%(if not, $\Bag$ can be replaced by its children in $T$).
Given a bag $\Bag$, we denote by
\begin{compactenum}
\item $V_T(\Bag)$ the nodes of $G$ that appear in bags in $T(\Bag)$,
%\item $\RootV{V}_T(\Bag)$ the nodes of $G$ whose root bags are in $T(\Bag)$, and
\item $\AncestV{V}_T(\Bag)$ the nodes of $G$ that appear in $\Bag$ and its ancestors in $T$.
\end{compactenum}
The \emph{width} of the tree-decomposition $T$ is the size of the largest bag minus $1$.
The treewidth $t$ of $G$ is the smallest width among the widths of all tree decompositions of $G$.
Note that if $T$ achieves the treewidth of $G$, we have
%$|\RootV{V}_T(\Bag)|\leq |V_T(\Bag)|\leq (t+1)\cdot |T(\Bag)|$.
$|V_T(\Bag)|\leq (t+1)\cdot |T(\Bag)|$.
Given a graph $G$ with treewidth $t$ and a fixed $\alpha\in \Nats$,
a tree-decomposition $\Tree(G)$ is called \emph{$\alpha$-approximate}
if it has width at most $\alpha\cdot (t+1)-1$.
%%\cref{sec:prelim_tree_dec} presents some basic results on tree-decompositions and their relevance to the problem we study in this paper.
\cref{fig:G} illustrates the above definitions on a small example.

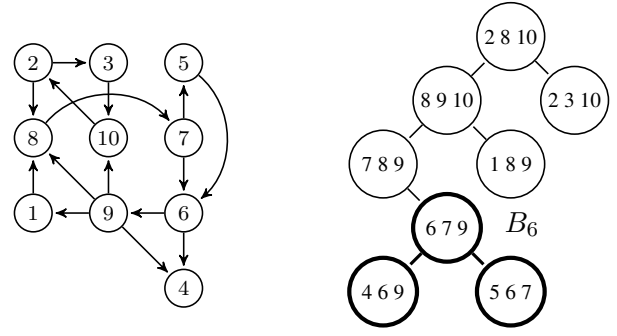
\begin{figure}
\newcommand{\distone}{5cm*0.2}
\centering
\begin{tikzpicture}[->,>=stealth',shorten >=1pt,auto,node distance=\distone,
                    semithick,scale=1 ]
      
\tikzstyle{every state}=[fill=white,draw=black,text=black,font=\small , inner sep=0.05cm, minimum size=0.5cm]
\tikzstyle{invis}=[fill=white,draw=white,text=white,font=\small , inner sep=-0.05cm]

\node[state] (v1) at (0,0.5) {$1$};
\node[state,above of=v1] (v8) {$8$};
\node[state,right of=v1] (v9) {$9$};
\node[state,above of=v8] (v2) {$2$};
\node[state,above of=v9] (v10) {$10$};
\node[state,above of=v10] (v3) {$3$};
\node[state,right of=v9] (v6) {$6$};
\node[state,below of=v6] (v4) {$4$};
\node[state,above of=v6] (v7) {$7$};
\node[state,above of=v7] (v5) {$5$};

\path (v1) edge (v8)
(v8) edge[bend left=50] (v7)
(v9) edge (v8) edge (v1) edge (v4) edge (v10)
(v2) edge (v3) edge (v8)
(v10) edge (v2)
(v3) edge (v10)
(v6) edge (v9) edge (v4)
(v7) edge (v5) edge (v6)
(v5) edge[bend left=50] (v6);

%\renewcommand{\distone}{5cm*0.38}
%\node[below of =v5] (ra) {\Large$\Rightarrow$};

\tikzstyle{every state}=[fill=white,draw=black,text=black,font=\small , inner sep=0.05cm, minimum size=0.9cm]
\renewcommand{\distone}{4cm*0.3}

%\node[state,below of=v9,node distance=sqrt(2)*\dist] (v8910) {8 9 10};
\node[state,] (v8910) at (5.5,2) {8 9 10};
\node[state,below right of=v8910] (v189) {1 8 9};
\node[state,above right of=v8910] (v2810) {2 8 10};
\node[state,below right of=v2810] (v2310) {2 3 10};
\node[state,below left of=v8910] (v789) {7 8 9};
\node[state,  ultra thick,below right of=v789] (v679) {6 7 9};
\node[state, ultra thick, below left of=v679] (v469) {4 6 9};
\node[state,ultra thick, below right of=v679] (v567) {5 6 7};

\node[] at (6.5,0.35) {\large $\Bag_6$};

\path[-] (v189) edge (v8910)
(v8910) edge (v2810) edge (v789)
(v2810) edge (v2310)
(v789) edge (v679)
(v679) edge[very thick] (v469) edge[very thick] (v567);

\end{tikzpicture}
\caption{A graph $G$ with treewidth $2$ (left) and a corresponding tree-decomposition $T=\Tree(G)$ of $8$ bags and width $2$ (right).
The distinguished bag $\Bag_6$ is the root bag of node $6$. 
We have $V_T(\Bag_6)=\{6,7,9,4,5\}$ and $\AncestV{V}_T(\Bag_6)=\{6,7,9,8,10,2\}$.
The subtree $T(\Bag_6)$ is shown in bold.
}\label{fig:G}
\end{figure}

\subsection{Concurrent graphs}

\noindent{\bf Product graphs.}
A graph $G_{\mathsf{p}}=(V_{\mathsf{p}},E_{\mathsf{p}})$ is said to be the {\em product graph} of $k$ graphs $(G_i=(V_i, E_i))_{1\leq i\leq k}$
if $V_{\mathsf{p}}=\prod_i V_i$ and $E_{\mathsf{p}}$ is such that for all $u,v\in V_{\mathsf{p}}$ 
with $u=\Prod{u_i}_{1\leq i\leq k}$ and $v=\Prod{v_i}_{1\leq i\leq k}$,
we have $(u,v)\in E_{\mathsf{p}}$ iff there exists a set $\mathcal{I}\subseteq [k]$ such that
(i)~$(u_i, v_i)\in E_i$ for all $i\in \mathcal{I}$, and  (ii)~$u_{i}=v_{i}$ for all  $i \not\in \mathcal{I}$.
In words, an edge $(u,v)\in E_{\mathsf{p}}$ is formed in the product graph
by traversing a set of edges $\{(u_i,v_i)\in E_i\}_{i\in \mathcal{I}}$ in some component graphs $\{G_i\}_{i\in \mathcal{I}}$,
and traversing no edges in the remaining $\{G_{i}\}_{i\not\in \mathcal{I}}$.
We say that $G_{\mathsf{p}}$ is the $k$-{\em self-product} of a graph $G'$ if $G_i=G'$ for all $1\leq i\leq k$.

\smallskip\noindent{\bf Concurrent graphs.} 
A graph $G=(V,E)$ is called a  \emph{concurrent graph} of $k$ graphs $(G_i=(V_i, E_i))_{1\leq i\leq k}$
if $V=V_{\mathsf{p}}$ and $E\subseteq E_{\mathsf{p}}$, where $G_{\mathsf{p}}=(V_{\mathsf{p}}, E_{\mathsf{p}})$
is the product graph of $(G_i)_{i}$.
Given a concurrent graph $G=(V,E)$ and a node $u\in V$, we will denote by $u_i$ the $i$-th constituent of $u$.
We say that $G$ is a $k$-{\em self-concurrent} of a graph $G'$ if $G_{\mathsf{p}}$ is the $k$-{\em self-product} of $G'$.

\smallskip\noindent{\bf Various notions of composition.}
%Since our main contributions will be improved algorithms, 
The framework we consider is quite general as it captures various different notions of concurrent composition.
Indeed, the edge set of the concurrent graph is any possible subset of the edge set of the corresponding product graph.
Then, two  well-known composition notions can be modeled as follows.
For any edge $(u,v)\in E$ of the concurrent graph $G$, let $\mathcal{I}_{u,v}=\{i\in [k]: (u_i,v_i)\in E_i\}$ denote the components that execute a 
transition in $(u,v)$.
\begin{compactenum}
\item 
In \emph{synchronous composition} at every step all components make one move 
each simultaneously.
This is captured by $\mathcal{I}_{u,v}=[k]$ for all  $(u,v)\in E$.
\item 
In \emph{asynchronous composition} at every step only one component makes a 
move. 
This is captured by $|\mathcal{I}_{u,v}|=1$ for all $(u,v)\in E$.
\end{compactenum}
Thus the framework we consider is not specific to any particular notion of composition,
and all our results apply to various different notions of concurrent composition that exist in the literature.

\newcommand{\Fork}{\mathsf{fork}}
\newcommand{\Knife}{\mathsf{knife}}
\newcommand{\Attain}{\mathsf{acquire}}
\newcommand{\Release}{\mathsf{release}}
\newcommand{\Lock}{\mathsf{lock}}
\newcommand{\Unlock}{\mathsf{unlock}}
\newcommand{\Mine}{\mathsf{mine}}
\newcommand{\Dine}{\mathsf{dine}}
\newcommand{\Sleep}{\mathsf{discuss}}

\begin{figure*}
\renewcommand{\thealgocf}{}
\removelatexerror
\centering
\begin{tikzpicture}[thick, >=latex,
pre/.style={<-,shorten >= 1pt, shorten <=1pt, thick},
post/.style={->,shorten >= 1pt, shorten <=1pt,  thick},
und/.style={very thick, draw=gray},
bag/.style={ellipse, minimum height=7mm,minimum width=14mm,draw=gray!80, line width=1pt, inner sep=1},
internal/.style={circle,draw=black!80, inner sep=1, minimum size=4.5mm, very thick},
entry/.style={isosceles triangle, shape border rotate=-90, isosceles triangle stretches, minimum width=8mm, minimum height=3.6mm, draw=black!80, inner sep=0},
exit/.style={isosceles triangle, shape border rotate=90, isosceles triangle stretches, minimum width=8mm, minimum height=3.6mm, draw=black!80, inner sep=0},
call/.style={isosceles triangle, shape border rotate=0, isosceles triangle stretches, minimum width=8mm, minimum height=3.6mm, draw=black!80, inner sep=0},
return/.style={isosceles triangle, shape border rotate=180, isosceles triangle stretches, minimum width=8mm, minimum height=3.6mm, draw=black!80, inner sep=0},
%exit/.style={circle,draw=black!80, inner sep=2, minimum size=4pt},
%call/.style={circle,draw=black!80, inner sep=2, minimum size=4pt},
%return/.style={circle,draw=black!80, inner sep=2, minimum size=4pt},
virt/.style={circle,draw=black!50,fill=black!20, opacity=0}]

\node	[text width = 5.3cm] at (0,0)	{%

\begin{algorithm}[H]
%\TitleOfAlgo{$\Dotvectoralgo$}
\small
\SetInd{0.6em}{0.0em}
\SetAlgoNoLine
\SetAlgorithmName{Method}{method}
\DontPrintSemicolon
%\setstretch{1.05}
\caption{$\Philosophers$}\label{algo:philosophers}
%\KwIn{$x,y\in \Reals^n$}
%\KwOut{The dot product $x^\top y$}
\While{$\True$}{
\While{$\Fork$ not $\Mine$ or $\Knife$ not $\Mine$}{
\If{$\Fork$ is free}{
$\Lock(\ell)$\\
$\Attain(\Fork)$\\
$\Unlock(\ell)$\\
}
\If{$\Knife$ is free}{
$\Lock(\ell)$\\
$\Attain(\Knife)$\\
$\Unlock(\ell)$\\
}
}
$\Dine(\Fork, \Knife)$\tcp{for some time}
$\Lock(\ell)$\\
$\Release(\Fork)$\\
$\Release(\Knife)$\\
$\Unlock(\ell)$\\
$\Sleep()$\tcp{for some time}
}

\end{algorithm}
};

\newcommand{\globalxdisposition}{4.75}
\newcommand{\globalydisposition}{3.4}
\newcommand{\ynodestep}{-0.8}
\newcommand{\xnodestep}{0.5}

\node	[internal]		(x1)	at	(\globalxdisposition+0*\xnodestep,\globalydisposition+0*\ynodestep)		{$1$};
\node	[internal]		(x20)	at	(\globalxdisposition+3*\xnodestep,\globalydisposition+0*\ynodestep)		{$20$};
\node	[internal]		(x2)	at	(\globalxdisposition+1*\xnodestep,\globalydisposition+1*\ynodestep)		{$2$};
\node	[internal]		(x13)	at	(\globalxdisposition-1*\xnodestep,\globalydisposition+2*\ynodestep)		{$13$};
\node	[internal]		(x14)	at	(\globalxdisposition-1*\xnodestep,\globalydisposition+3*\ynodestep)		{$14$};
\node	[internal]		(x15)	at	(\globalxdisposition-1*\xnodestep,\globalydisposition+4*\ynodestep)		{$15$};
\node	[internal]		(x16)	at	(\globalxdisposition-1*\xnodestep,\globalydisposition+5*\ynodestep)		{$16$};
\node	[internal]		(x17)	at	(\globalxdisposition-1*\xnodestep,\globalydisposition+6*\ynodestep)		{$17$};
\node	[internal]		(x18)	at	(\globalxdisposition-1*\xnodestep,\globalydisposition+7*\ynodestep)		{$18$};
\node	[internal]		(x19)	at	(\globalxdisposition-1*\xnodestep,\globalydisposition+8*\ynodestep)		{$19$};

\node	[internal]		(x3)	at	(\globalxdisposition+2*\xnodestep,\globalydisposition+2*\ynodestep)		{$3$};
\node	[internal]		(x4)	at	(\globalxdisposition+2*\xnodestep,\globalydisposition+3*\ynodestep)		{$4$};
\node	[internal]		(x5)	at	(\globalxdisposition+2*\xnodestep,\globalydisposition+4*\ynodestep)		{$5$};
\node	[internal]		(x6)	at	(\globalxdisposition+2*\xnodestep,\globalydisposition+5*\ynodestep)		{$6$};
\node	[internal]		(x7)	at	(\globalxdisposition+2*\xnodestep,\globalydisposition+6*\ynodestep)		{$7$};

\node	[internal]		(x8)	at	(\globalxdisposition+5*\xnodestep,\globalydisposition+2*\ynodestep)		{$8$};
\node	[internal]		(x9)	at	(\globalxdisposition+5*\xnodestep,\globalydisposition+3*\ynodestep)		{$9$};
\node	[internal]		(x10)	at	(\globalxdisposition+5*\xnodestep,\globalydisposition+4*\ynodestep)		{$10$};
\node	[internal]		(x11)	at	(\globalxdisposition+5*\xnodestep,\globalydisposition+5*\ynodestep)		{$11$};
\node	[internal]		(x12)	at	(\globalxdisposition+5*\xnodestep,\globalydisposition+6*\ynodestep)		{$12$};

\draw [->, thick] (x1) to 	(x2);

\draw [->, thick] (x2) to 	(x13);
\draw [->, thick] (x2) to 	(x3);

\draw [->, thick] (x13) to 	(x14);
\draw [->, thick] (x14) to 	(x15);
\draw [->, thick] (x15) to 	(x16);
\draw [->, thick] (x16) to 	(x17);
\draw [->, thick] (x17) to 	(x18);
\draw [->, thick] (x18) to 	(x19);
\draw [->, thick] (x1) to 	(x20);

\draw [->, thick] (x3) to 	(x4);
\draw [->, thick] (x4) to 	(x5);
\draw [->, thick] (x5) to 	(x6);
\draw [->, thick] (x6) to 	(x7);

\draw [->, thick] (x8) to 	(x9);
\draw [->, thick] (x9) to 	(x10);
\draw [->, thick] (x10) to 	(x11);
\draw [->, thick] (x11) to 	(x12);

\draw [->, thick] (x7) .. controls ([xshift=1cm] x7) and ([xshift=-1cm] x8) .. (x8);
\draw [->, thick] (x12) .. controls ([xshift=-1cm] x12) and ([xshift=1.5cm] x2) .. (x2);
\draw [->, thick] (x19) .. controls ([xshift=-1cm] x19) and ([xshift=-1.5cm] x1) .. (x1);
\draw [->, thick] (x3) .. controls ([xshift=-0.7cm] x3) and ([xshift=-0.7cm] x7) .. (x7);
\draw [->, thick] (x8) .. controls ([xshift=0.7cm] x8) and ([xshift=0.7cm] x12) .. (x12);

\renewcommand{\globalxdisposition}{10.5}
\renewcommand{\globalydisposition}{3.3}
\renewcommand{\ynodestep}{-0.9}
\renewcommand{\xnodestep}{0.6}

\node	[bag]		(b1)	at	(\globalxdisposition + 1*\xnodestep,\globalydisposition+0*\ynodestep)		{$1, 19, 20$};
\node	[bag]		(b2)	at	(\globalxdisposition + 1*\xnodestep,\globalydisposition+1*\ynodestep)		{$1, 2, 19$};
\node	[bag]		(b13)	at	(\globalxdisposition + -1*\xnodestep,\globalydisposition+2*\ynodestep)		{$2, 13, 19$};
\node	[bag]		(b14)	at	(\globalxdisposition + -1*\xnodestep,\globalydisposition+3*\ynodestep)		{$13, 14, 19$};
\node	[bag]		(b15)	at	(\globalxdisposition + -1*\xnodestep,\globalydisposition+4*\ynodestep)		{$14, 15, 19$};
\node	[bag]		(b16)	at	(\globalxdisposition + -1*\xnodestep,\globalydisposition+5*\ynodestep)		{$15, 16, 19$};
\node	[bag]		(b17)	at	(\globalxdisposition + -1*\xnodestep,\globalydisposition+6*\ynodestep)		{$16, 17, 19$};
\node	[bag]		(b18)	at	(\globalxdisposition + -1*\xnodestep,\globalydisposition+7*\ynodestep)		{$17, 18, 19$};

\node	[bag]		(b7)	at	(\globalxdisposition + 3*\xnodestep,\globalydisposition+2*\ynodestep)		{$2, 7, 12$};

\node	[bag]		(b3)	at	(\globalxdisposition + 2*\xnodestep,\globalydisposition+3*\ynodestep)		{$2, 3, 7$};
\node	[bag]		(b4)	at	(\globalxdisposition + 2*\xnodestep,\globalydisposition+4*\ynodestep)		{$3, 4, 7$};
\node	[bag]		(b5)	at	(\globalxdisposition + 2*\xnodestep,\globalydisposition+5*\ynodestep)		{$4, 5, 7$};
\node	[bag]		(b6)	at	(\globalxdisposition + 2*\xnodestep,\globalydisposition+6*\ynodestep)		{$5, 6, 7$};

\node	[bag]		(b8)	at	(\globalxdisposition + 5*\xnodestep,\globalydisposition+3*\ynodestep)		{$7, 8, 12$};
\node	[bag]		(b9)	at	(\globalxdisposition + 5*\xnodestep,\globalydisposition+4*\ynodestep)		{$8, 9, 12$};
\node	[bag]		(b10)	at	(\globalxdisposition + 5*\xnodestep,\globalydisposition+5*\ynodestep)	   {$9, 10, 12$};
\node	[bag]		(b11)	at	(\globalxdisposition + 5*\xnodestep,\globalydisposition+6*\ynodestep)	   {$10, 11, 12$};

\draw [-, very thick] (b1) to 	(b2);
\draw [-, very thick] (b2) to 	(b13);
\draw [-, very thick] (b13) to 	(b14);
\draw [-, very thick] (b14) to 	(b15);
\draw [-, very thick] (b15) to 	(b16);
\draw [-, very thick] (b16) to 	(b17);
\draw [-, very thick] (b17) to 	(b18);

\draw [-, very thick] (b2) to 	(b7);
\draw [-, very thick] (b7) to 	(b8);
\draw [-, very thick] (b8) to 	(b9);
\draw [-, very thick] (b9) to 	(b10);
\draw [-, very thick] (b10) to 	(b11);

\draw [-, very thick] (b7) to 	(b3);
\draw [-, very thick] (b3) to 	(b4);
\draw [-, very thick] (b4) to 	(b5);
\draw [-, very thick] (b5) to 	(b6);

\end{tikzpicture}
\caption{A concurrent program (left), its controlflow graph (middle), and a tree decomposition of the controlflow graph (right).}\label{fig:fig_example}
\end{figure*}

\smallskip\noindent{\bf Partial nodes of concurrent graphs.}
A {\em partial node} $\PNode{u}$ of a concurrent graph $G$ is an element of $\prod_i(V_i\cup \{\bot\})$,
where $\bot \not\in \bigcup_i V_i$.
Intuitively, $\bot$ is a fresh symbol to denote that a component is unspecified.
A partial node $\PNode{u}$ is said to {\em refine} a partial node $\PNode{v}$,
denoted by $\PNode{u}\Refines \PNode{v}$ if for all $1\leq i\leq k$ either $\PNode{v}_i=\bot $ or $\PNode{v}_i=\PNode{u}_i$.
We say that the partial node $\PNode{u}$ \emph{strictly refines} $\PNode{v}$, denoted by $\PNode{u}\Strictrefines \PNode{v}$,
if $\PNode{u}\Refines \PNode{v}$ and $\PNode{u}\neq \PNode{v}$
(i.e., for at least one constituent $i$ we have $\PNode{v}_i=\bot$ but $\PNode{u}_i\neq \bot$).
A partial node $\PNode{u}$ is called \emph{strictly partial} if it is strictly refined by some node $u\in V$
(i.e., $\PNode{u}$ has at least one $\bot$).
The notion of semiring distances is extended to partial nodes, 
and for partial nodes $\PNode{u}, \PNode{v}$ of $G$ we define the semiring distance from $\PNode{u}$ to $\PNode{v}$ as
\[
\Distance(\PNode{u}, \PNode{v}) = \bigoplus_{u\Refines \PNode{u}, v\Refines \PNode{v}}\Distance(u,v)
\]
where $u,v\in V$.
In the sequel, a partial node $\PNode{u}$ will be either (i)~a node of $V$, or (ii)~a strictly partial node.
We refer to nodes of the first case as actual nodes, and write $u$ (i.e., without the bar).
Distances where one endpoint is a strictly partial node $\PNode{u}$ succinctly quantify
over all nodes of all the components for which the corresponding constituent of $\PNode{u}$ is $\bot$.
Observe that the distance still depends on the unspecified components.
%The relevance of the semiring distance between strictly partial nodes is as follows: 
%suppose for two components 
%(or in general for many components) we are interested in the semiring distance between two nodes
%of a component, such as whether one node is reachable from the other, irrespective of the 
%other component, then this is succinctly phrased as a problem for partial nodes.
%In other words, partial nodes allow us to succinctly quantify over the nodes of certain 
%components with the $\bot$ symbol.

\smallskip\noindent{\bf The algebraic paths problem on concurrent graphs of constant-treewidth components.}
In this work we are interested in the following problem.
Let $G=(V,E)$ be a concurrent graph of $k\geq 2$ constant-treewidth graphs $(G_i=(V_i, E_i))_{1\leq i\leq k}$,
and $\Weight:E\to \Sigma$ be a weight function that assigns to every edge of $G$
a weight from a set $\Sigma$ that forms a complete, closed semiring 
$S=(\Sigma, \oplus, \otimes, \Zero, \One)$.
The {\em algebraic path problem} on $G$ asks the following types of queries:
\begin{compactenum}
\item {\em Single-source query.} Given a partial node $\PNode{u}$ of $G$, return the distance $\Distance(\PNode{u},v)$ 
to every node $v\in V$. When the partial node $\PNode{u}$ is  an actual node of $G$, we have a traditional single-source query.
\item {\em Pair query.} Given two nodes $u,v\in V$, return the distance $\Distance(u,v)$.
\item {\em Partial pair query.} Given two partial nodes $\PNode{u}, \PNode{v}$ of $G$ where at least one is strictly partial,
return the distance $\Distance(\PNode{u}, \PNode{v})$.
\end{compactenum}

\cref{fig:fig_example} presents the notions introduced in this section on a toy example on the dining philosophers problem.
In \cref{sec:example} we discuss the rich modeling capabilities of the algebraic paths framework,
and illustrate the importance of pair and partial pair queries in the analysis of the dining philosophers program.

\smallskip\noindent{\bf Input parameters.}
For technical convenience, we consider a uniform upper bound $n$ on the number of nodes of each $G_i$ (i.e. $|V_i|\leq n$).
Similarly, we let $t$ be an upper bound on the treewidth of each $G_i$.
The number $k$ is taken to be fixed and independent of $n$.
The input of the problem consists of the graphs $(G_i)_{1\leq i\leq k}$, together with some 
representation of the edge relation $E$ of $G$. 

\smallskip\noindent{\bf Complexity measures.}
The complexity of our algorithms is measured as a function of $n$.
In particular, we ignore the size of the representation of $E$ when considering 
the size of the input.
This has the advantage of obtaining complexity bounds that are independent of the representation 
of $E$, which can be represented implicitly (such as synchronous or asynchronous composition) 
or explicitly, depending on the modeling of the problem under consideration. 
The time complexity of our algorithms is measured in number of operations,
with each operation being either a basic machine operation,
or an application of one of the operations of the semiring.

\section{Strongly Balanced Tree Decompositions}\label{sec:prelim_tree_dec}

In this section we introduce the notion of strongly balanced tree 
decompositions, and present an algorithm for computing them efficiently on 
constant-treewidth graphs.
Informally, a strongly balanced tree-decomposition is a binary tree-decomposition in which the number of descendants of each bag is typically approximately half of that of its parent.
The following sections make use of this construction.

\smallskip\noindent{\bf Strongly balanced tree decompositions.}
Given a binary tree-decomposition $T$ and constants $0<\beta < 1$, $\gamma\in \Nats^+$,
a bag $\Bag$ of $T$ is called \emph{$(\beta, \gamma)$-balanced} 
if for every descendant $\Bag_i$ of $\Bag$ with $\Level(\Bag_i)-\Level(\Bag)=\gamma$,
we have $|T(\Bag_i)|\leq \beta \cdot |T(\Bag)|$,
i.e., the number of bags in $T(\Bag_i)$ is at most a $\beta$-fraction of those in $T(\Bag)$.
A tree-decomposition $T$ is called a $(\beta, \gamma)$ tree-decomposition if
every bag of $T$ is $(\beta, \gamma)$-balanced.
A $(\beta, \gamma)$ tree-decomposition that is $\alpha$-approximate is called an
$(\alpha, \beta, \gamma)$ tree-decomposition.
The following theorem is central to the results obtained in this paper.
%The proof is technical and presented in the full version~\cite{TR},
The proof is technical and presented in \cref{sec:tree_decomp},
and here we provide a sketch of the algorithm for obtaining it.

\smallskip
\begin{restatable}{theorem}{balancedtreedec}\label{them:tree_decomp}
For every graph $G$ with $n$ nodes and constant treewidth, 
for any fixed $\delta>0$ and $\lambda\in \Nats$ with $\lambda\geq 2$,
let $\alpha=\Bagfactor$, $\beta=\Balfactor$, and $\gamma=\Levelfactor$.
A binary $(\alpha,\beta,\gamma)$ tree-decomposition $\Tree(G)$ with $O(n)$ bags
can be constructed in $O(n\cdot \log n)$ time and  $O(n)$ space.
\end{restatable}

\smallskip\noindent{\bf Sketch of \cref{them:tree_decomp}.}
The construction of \cref{them:tree_decomp} considers that a 
tree-decomposition $\Tree'(G)$ of width $t$ and $O(n)$ bags 
is given (which can be obtained using e.g.~\cite{Bodlaender96} in $O(n)$ time). 
Given two parameters $\delta>0$ and $\lambda\in \Nats$ with $\lambda\geq 2$,
$\Tree'(G)$ is turned to an $(\alpha, \beta, \gamma)$ tree-decomposition,
for $\alpha=\Bagfactor$, $\beta=\Balfactor$, and $\gamma=\Levelfactor$,
 in two conceptual steps.
\begin{compactenum}
\item A tree of bags $\Ranktree$ is constructed, which is $(\beta, \gamma)$-balanced.
\item $\Ranktree$ is turned to an $\alpha$-approximate tree 
decomposition of $G$.
\end{compactenum}
The first construction is obtained by a recursive algorithm $\Rankalgo$, which operates on inputs
$(\Comp, \ell)$, where $\Comp$ is a %(not necessarily connected) 
component of $\Tree'(G)$,
and $\ell\in [\Levelfactor]$ specifies the type of operation the algorithm performs on $\Comp$.
Given such a component $\Comp$, we denote by $\Nhood(\Comp)$ the \emph{neighborhood}
of $\Comp$, defined as the set of bags of $\Tree'(G)$ that are incident to $\Comp$.
Informally, on input $(\Comp, \ell)$, the algorithm partitions $\Comp$ into two sub-components $\ov{\Comp}_1$ and $\ov{\Comp}_2$
such that either 
(i)~the size of each $\ov{\Comp}_i$  is approximately half the size of $\Comp$, or
(ii)~the size of the neighborhood of each $\ov{\Comp}_i$ is approximately half the size of the neighborhood of $\Comp$.
More specifically, 
\begin{compactenum}
\item If $\ell>0$, then $\Comp$ is partitioned into components $\Complist=(\Comp_1,\dots, \Comp_{r})$,
by removing a list of bags $\Seplist=(\Bag_1,\dots \Bag_m)$, such that $|\Comp_i|\leq\frac{\delta}{2}\cdot |\Comp|$.
The union of $\Seplist$ yields a new bag $\RBag$ in $\Ranktree$.
Then $\Complist$ is merged into two components $\ov{\Comp}_1$, $\ov{\Comp}_2$ with $|\ov{\Comp}_1|\leq |\ov{\Comp}_2|\leq \frac{1+\delta}{2}\cdot |\Comp|$. 
Finally, each $\ov{\Comp}_i$ is passed on to the next recursive step with $\ell=(\ell+1)\mod \Levelfactor$.

\item If $\ell=0$, then $\Comp$ is partitioned into two components 
$\ov{\Comp}_1, \ov{\Comp}_{2}$
such that $|\Nhood(\ov{\Comp}_i)\cap \Nhood(\Comp)|\leq \frac{|\Nhood(\Comp)|}{2}$ by removing a single bag $\Bag$.
This bag becomes a new bag $\RBag$ in $\Ranktree$, and each $\ov{\Comp}_i$ is passed on to the next recursive step
with $\ell=(\ell+1)\mod \Levelfactor$.
\end{compactenum}
\cref{fig:nhood} provides an illustration.
The second construction is obtained simply by inserting in each bag $\RBag$ of $\Ranktree$
the nodes contained in the neighborhood $\Nhood(\Comp)$ of the component $\Comp$
upon which $\RBag$ was constructed.
%In \cref{sec:tree_decomp} we prove that these two steps yield an $(\alpha, \beta, \gamma)$ tree-decomposition of $G$.

\begin{figure}
\newcommand\irregularcircle[2]{% radius, irregularity
  \pgfextra {\pgfmathsetmacro\len{(#1)+rand*(#2)}}
  +(0:\len pt)
  \foreach \a in {10,25,...,350}{
    \pgfextra {\pgfmathsetmacro\len{(#1)+rand*(#2)}}
    -- +(\a:\len pt)
  } -- cycle
}

\newcommand\irregularellipse[3]{% a,b, irregularity
  \pgfextra {\pgfmathsetmacro\b{(#1)+rand*(#3)}}
  \pgfextra {\pgfmathsetmacro\c{(#2)+rand*(#3)}}
  \pgfextra {\pgfmathsetmacro\len{\b*\c/ ( (\c * cos(0))^2 + (\b * sin(0 ))^2)^(1/2)}}
  +(0:\len pt)
  \foreach \a in {10,20,...,350}{
    \pgfextra {\pgfmathsetmacro\b{(#1)+rand*(#3)}}
    \pgfextra {\pgfmathsetmacro\c{(#2)+rand*(#3)}}
    \pgfextra {\pgfmathsetmacro\len{\b*\c/ ( (\c * cos(\a))^2 + (\b * sin(\a ))^2)^(1/2)}}
    -- +(\a:\len  pt)
  } -- cycle
}

\centering
\begin{tikzpicture}
\pgfmathsetseed{1234}
\coordinate (c) at (0,0);
%\draw[very thick,rounded corners=1mm, fill=red!30] (c) \irregularcircle{1.6cm}{2mm};
%\draw[very thick,rounded corners=1mm, fill=gray!30] (c) \irregularcircle{1.2cm}{1.5mm};
\draw[very thick,rounded corners=1mm, fill=red!30] (c) \irregularellipse{38mm}{13mm}{1.2mm};
\draw[very thick,rounded corners=1mm, fill=gray!30] (c) \irregularellipse{23mm}{9mm}{1.2mm};
\node[] at (-3,0) {$\Nhood(\Comp)$};
\node[circle, draw=black, thick, fill=white] at (0,0) (bag)	{$\Bag$};
\node[] at (1,0.1) {$\ov{\Comp}_2$};
%\node[] at (1.3,-0.2) {$\Comp_3$};
\node[] at (-1.2,-0.1) {$\ov{\Comp}_1$};
%\draw[thick, decoration={random steps,segment length=2mm,amplitude=0.5mm}, decorate,] (bag) -- (2,1.2);
\draw[thick, decoration={random steps,segment length=2mm,amplitude=0.5mm}, decorate,] (bag) -- (1.1,-1.2);
\draw[thick, decoration={random steps,segment length=2mm,amplitude=0.5mm}, decorate,] (bag) -- (-1.8,1.15);
\end{tikzpicture}
\caption{Illustration of one recursive step of $\Rankalgo$ on a component $\Comp$ (gray).
$\Comp$ is split into two sub-components $\ov{\Comp}_1$ and $\ov{\Comp}_2$
by removing a list of bags $\Seplist=(\Bag_i)_i$.
Once every $\lambda$ recursive calls, $\Seplist$ contains one bag, 
such that the neighborhood $\Nhood(\ov{\Comp}_i)$ of each $\ov{\Comp}_i$ is at most half the size of $\Nhood(\Comp)$
(i.e., the red area is split in half).
In the remaining $\lambda-1$ recursive calls, $\Seplist$ contains $m$ bags,
such that the size of each $\ov{\Comp}_i,$
is at most $\frac{1+\delta}{2}$ fraction the size of $\Comp$.
(i.e., the gray area is split in almost half).
}\label{fig:nhood}
\end{figure}

\smallskip\noindent{\bf Use of $(\alpha, \beta, \gamma)$ tree-decompositions.}
For ease of presentation we consider that every $\Tree(G)$ is a full binary tree.
%Let $\alpha=\Bagfactor$, $\beta=\Balfactor$ and $\gamma=\Levelfactor$.
Since our tree decompositions are $(\beta,\gamma)$-balanced,
we can always attach empty children bags to those that have only one child,
while increasing the size of $\Tree(G)$ by a constant factor only.
In the sequel, $\Tree(G)$ will denote a full binary 
$(\alpha, \beta, \gamma)$ tree-decomposition of $G$.
The parameters  $\delta$ and $\lambda$ will be chosen appropriately in later sections.
%for the best asymptotic results of our algorithms.
%%In \cref{sec:experiments} we report values of $\delta$ and $\lambda$ that were found to work well in practice.

\smallskip
\begin{remark}\label{rem:balance_comparison}
The notion of balanced tree decompositions exists in the 
literature~\cite{Elberfeld10,BH98}, but balancing only requires that the height of the tree is logarithmic in its size.
Here we develop a stronger notion of balancing, which is crucial for proving the complexity results of the algorithms
presented in this work.
\end{remark}

\section{Concurrent Tree Decomposition}\label{sec:product_tree_dec}
In this section we present the construction of a tree-decomposition $\Tree(G)$
of a concurrent graph $G=(V,E)$ of $k$ constant-treewidth graphs.
In general, $G$ can have treewidth which depends on the number of its nodes
(e.g., $G$ can be a grid, which has treewidth $n$, obtained as the product of two lines, which have treewidth $1$).
While the treewidth computation for constant-treewidth graphs is linear 
time~\cite{Bodlaender96}, it is 
NP-complete for general graphs~\cite{Bodlaender93}. 
Hence computing a tree decomposition that achieves the treewidth of $G$ 
can be computationally expensive (e.g., exponential in the size of $G$).
Here we develop an algorithm $\ProductTree$ which constructs a 
tree-decomposition $\ProductTree(G)$ of $G$, given a $(\alpha,\beta,\gamma)$
tree-decomposition of the components, in $O(n^k)$ time and space 
(i.e., linear in the size of $G$), such that the 
following properties hold:
(i)~the width is $O(n^{k-1})$; and 
(ii)~for every bag in level at least $i\cdot \gamma$, the size of the bag is $O(n^{k-1}\cdot \beta^{i})$
(i.e., the size of the bags decreases geometrically along the levels).
%%Finally, we also show that there exist concurrent graphs $G$ for which the 
%%treewidth is $n^{k-1}$.

\smallskip\noindent{\bf Algorithm $\ProductTree$ for concurrent tree decomposition.}
Let $G$ be a concurrent graph of $k$ graphs $(G_i)_{1\leq i\leq k}$.
The input consists of a full binary tree-decomposition $T_i$ of constant width for every graph $G_i$.
In the following, $\Bag_i$ ranges over bags of $T_i$, 
and we denote by $\Bag_{i,r}$, with $r\in \Children$, the $r$-th child of $\Bag_i$.
We construct the {\em concurrent tree-decomposition} $T=\ProductTree(G)=(V_T, E_T)$ of $G$
using the recursive procedure $\ProductTree$, which operates as follows.
On input $(T_i(\Bag_i))_{1\leq i\leq k}$, return a tree decomposition where
\begin{compactenum}
\item The root bag $\Bag$ is
\begin{equation}\label{eq:bag}
\Bag=\bigcup_{1\leq i\leq k}\left( \left(\prod_{j< i} V_{T_j}\left(\Bag_j\right) \right) \times \Bag_i \times \left(\prod_{j> i} V_{T_j}\left(\Bag_j\right) \right)\right)
\end{equation}

\item If every $\Bag_i$ is a non-leaf bag of $T_i$,
for every choice of $\Prod{r_1,\dots, r_k}\in\Children^k$,
repeat the procedure for $(T_i(\Bag_{i,{r_i}}))_{1\leq i\leq k}$,
and let $\Bag'$ be the root of the returned tree. Make $\Bag'$ a child of $\Bag$.
\end{compactenum}
%The pseudocode of $\ProductTree$ can be found in \cref{sec:algorithms}. 
Let $\Bag_i$ be the root of the tree-decomposition $T_i$.
We denote by $\ProductTree(G)$ the application of the recursive procedure $\ProductTree$ on $(T_i(\Bag_i))_{1\leq i\leq k}$.
\cref{fig:product_tree_dec} provides an illustration.
%We first present the correctness of $\ProductTree$.

\begin{figure}
\newcommand{\distone}{5cm*0.3}
\def \darkgreen {black!50!green}
\def \darkred {black!30!red}
\centering
\begin{tikzpicture}[->,>=stealth',shorten >=1pt,auto,node distance=\distone,
                    very thick,scale=1 ]
      
\tikzstyle{every state}=[fill=white,draw=black,text=black,font=\small , inner sep=0.05cm, minimum size=0.5cm]
\tikzstyle{invis}=[fill=white,draw=white,text=white,font=\small , inner sep=-0.05cm]
\tikzstyle{every state}=[fill=white,draw=black,text=black,font=\small , inner sep=0.05cm, minimum size=0.9cm]
\renewcommand{\distone}{4cm*0.28}

\def\xbias{2}
\def\ybias{2.2}
\def\textbias{0.7}

\node[] (t1) at (-\xbias,\textbias-2.1) {$\Tree(G_1)$};
\node[state,minimum size=0.7cm,] (v1) at (-\xbias,0) {$1$};
\node[state,minimum size=0.7cm, below left of=v1] (v2) {$1$~~$2$};
\node[state,minimum size=0.7cm, below right of=v1] (v3) {$1$~~$3$};
\path[-] (v1) edge (v2) edge (v3);

\node[] (t2) at (\xbias,\textbias-2.1) {$\Tree(G_2)$};
\node[state,minimum size=0.7cm, ] (va) at (\xbias,0) {$a$};
\node[state,minimum size=0.7cm, below left of=va] (vb) {$a$~~$b$};
\node[state, minimum size=0.7cm, below right of=va] (vc) {$a$~~$c$};
\path[-] (va) edge (vb) edge (vc);

\renewcommand{\distone}{5cm*0.42}
\node[] (t) at (0,-\ybias+\textbias-3.2) {$\ProductTree(G)$};
\node[state, text width=1.5cm] (v1a) at (0,-\ybias) {~~~~$\Prod{1,a}$~~~\\ $\Prod{1,b}$  $\Prod{2,a}$ $\Prod{1,c}$ $\Prod{3,a}$};
\node[state, align=center, text width=0.7cm, below left of=v1a] (v3a) {$\Prod{1,a}$  $\Prod{1,b}$ $\Prod{3,a}$ $\Prod{3,b}$ };
\node[state, align=center, text width=0.7cm, left of=v3a] (v2a) {$\Prod{1,a}$ $\Prod{1,b}$  $\Prod{2,a}$ $\Prod{2,b}$};
\node[state, align=center, text width=0.7cm, below right of=v1a] (v1b) {$\Prod{1,a}$ $\Prod{1,c}$ $\Prod{2,a}$ $\Prod{2,c}$};
\node[state, align=center, text width=0.7cm, right of=v1b] (v1c) {$\Prod{1,a}$ $\Prod{1,c}$ $\Prod{3,a}$ $\Prod{3,c}$};
\path[-] (v1a) edge (v2a) edge (v3a) edge (v1b) edge (v1c);

\end{tikzpicture}
\caption{The tree-decomposition $\ProductTree(G)$ of a concurrent graph $G$
of two constant-treewidth graphs $G_1$ and $G_2$.}\label{fig:product_tree_dec}
\end{figure}
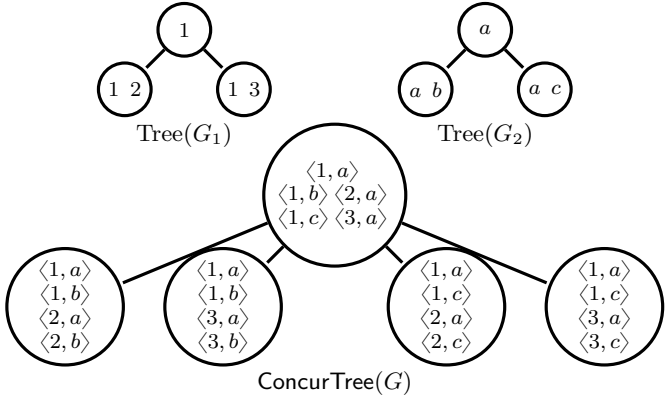

\smallskip
\begin{restatable}{remark}{rembagsize}\label{rem:bag_size}
Recall that for any bag $\Bag_j$ of a tree-decomposition $T_j$, we have $V_{T_j}(\Bag_j)=\bigcup_{\Bag'_j}\Bag'_j$,
where $\Bag'_j$ ranges over bags in $T_j(\Bag_j)$.
Then, for any two bags $\Bag_{i_1}$, $\Bag_{i_2}$, of tree-decompositions $T_{i_1}$ and $T_{i_2}$ respectively, we have
\[
V_{T_{i_1}}(\Bag_{i_1}) \times V_{T_{i_2}}(\Bag_{i_2}) = \bigcup_{\Bag'_{i_1}, \Bag'_{i_2}} \left(\Bag'_{i_1} \times \Bag'_{i_2}\right)
\]
where $\Bag'_{i_1}$ and $\Bag'_{i_2}$ range over bags in $T_{i_1}(\Bag_{i_1})$ and $T_{i_2}(\Bag_{i_2})$ respectively.
Since each tree-decomposition $T_i$ has constant width,
it follows that $|V_{T_{i_1}}(\Bag_{i_1}) \times V_{T_{i_2}}(\Bag_{i_2})|=O(|T_{i_1}(\Bag_{i_1})|\cdot |T_{i_2}(\Bag_{i_2})|)$.
Thus, the size of each bag $\Bag$ of $\ProductTree(G)$ constructed in \cref{eq:bag}
on some input $(T_i(\Bag_i))_i$ is $|\Bag|=O(\sum_i \prod_{j\neq i} n_j)$, where $n_i=|T_i(\Bag_i)|$.
%We will use this fact later in the complexity analysis of our algorithms.
\end{restatable}

In view of \cref{rem:bag_size}, the time and space required by $\ProductTree$
to operate on input $(T_i(\Bag_i))_{1\leq i\leq k}$ where $|T_i(\Bag_i)|=n_i$ is given, up to constant factors, by 

\begin{equation}
\Time(n_1,\dots, n_k) \leq \sum_{1\leq i\leq k} \prod_{j\neq i}n_j  + \sum_{\mathclap{(r_i)_i\in \Children^k}} \Time( n_{1,{r_1}},\dots,  n_{k,{r_k}})
\end{equation}
such that  for every $i$ we have that $\sum_{r_i\in \Children} n_{i,r_i} \leq n_i$.
In \cref{sec:product_tree_proofs} we establish that the solution to the above recurrence is $O(n^k)$, where $n_i\leq n$ for all $1\leq i\leq k$.

The following theorem establishes that $\ProductTree(G)$ is a tree-decomposition of $G$ constructed in $O(n^k)$ time and space.
Additionally, if every input tree-decomposition $T_i$ is $(\beta,\gamma)$-balanced, 
then the size of each bag $\Bag$ of $\ProductTree(G)$ decreases geometrically with its level $\Level(\Bag)$.
%See~\cite{TR} for a formal proof.
See \cref{sec:product_tree_proofs} for a formal proof.

\smallskip
\begin{restatable}{theorem}{producttree}\label{them:product}
Let $G=(V,E)$ be a concurrent graph of $k$ constant-treewidth graphs $(G_i)_{1\leq i\leq k}$ of $n$ nodes each.
Let a binary $(\alpha,\beta,\gamma)$ tree-decomposition $T_i$ for every graph $G_i$ be given, for some constant $\alpha$.
%where $\alpha=\Bagfactor$, $\beta=\Balfactor$ and $\gamma=\Levelfactor$ for some $\lambda\in \Nats$ and $\delta>0$.
$\ProductTree$ constructs a $2^{k}$-ary tree-decomposition $\ProductTree(G)$ of $G$ in $O(n^k)$ time and space,
with the following property.
For every $i\in \Nats$ and bag $\Bag$ at level $\Level(\Bag)\geq i\cdot \gamma$, we have $|\Bag|=O(n^{k-1}\cdot \beta^{i})$. 
\end{restatable}

%\smallskip\noindent{\bf Optimality of the width of $\ProductTree(G)$.}
%%It follows from \cref{them:product} that the width of $\ProductTree(G)$ is determined by the size of its root, which has size $O(n^{k-1})$.

\section{Concurrent Algebraic Paths}\label{sec:concurrent}

We now turn our attention to the core algorithmic problem of this paper,
namely answering semiring distance queries in a concurrent graph $G$
of $k$ constant-treewidth graphs $(G_i)_{1\leq i\leq k}$.
To this direction, we develop a data-structure $\ConcurAP$ (for \emph{concurrent algebraic paths}) which will preprocess $G$
and afterwards support single-source, pair, and partial pair queries on $G$.
%As some of the proofs are technical, they are presented only in the full version~\cite{TR}
As some of the proofs are technical, they are presented only in \cref{sec:concurrent_proofs}.

\smallskip\noindent{\bf Semiring distances on tree decompositions.}
The preprocessing and query of our data-structure exploits a key property of semiring distances on tree decompositions.
This property is formally stated in \cref{lem:u_v_dist}, and concerns any two nodes $u$, $v$ that appear
in some distinct bags $\Bag_1$, $\Bag_j$ of $\Tree(G)$.
Informally, the semiring distance $\Distance(u,v)$ can be written as the semiring multiplication of distances $\Distance(x_i, x_{i+1})$,
where $x_i$ is a node that appears in the $i$-th and $(i-1)$-th bags of the 
unique simple path $\Bag_1\Path\Bag_j$ in $\Tree(G)$.
\cref{fig:uv_dist} provides an illustration.

\smallskip
\begin{restatable}{lemma}{uvdist}\label{lem:u_v_dist}
Consider a graph $G=(V,E)$ with a weight function $\Weight:E\rightarrow \Sigma$, and a tree-decomposition $\Tree(G)$.
Let $u,v\in V$, and $P:\Bag_{1},\Bag_{2},\dots,\Bag_{j}$ be a simple path in $T$ such that $u\in \Bag_{1}$ and $v\in \Bag_{j}$. 
Let $A=\{u\} \times (\prod_{1< i\leq j} \left(\Bag_{i-1}\cap \Bag_i\right))\times \{v\}$. 
%Let $A=\{u\} \times \Bag_2\times\dots \times \Bag_{{j-1}}\times \{v\}$. 
Then $\Distance(u,v)=\bigoplus_{(x_1,\dots, x_{j+1})\in A}\bigotimes_{i=1}^{j} \Distance(x_i, x_{i+1})$.
\end{restatable}

\begin{figure}
\newcommand{\distone}{5cm*0.3}
\def \darkgreen {black!50!green}
\def \darkred {black!30!red}
\centering
\begin{tikzpicture}[->,>=stealth',shorten >=1pt,auto,node distance=\distone,
                    very thick,scale=1 ]
      
\tikzstyle{every state}=[fill=white,draw=black,text=black,font=\small , inner sep=0.05cm, minimum size=0.5cm]
\tikzstyle{invis}=[fill=white,draw=white,text=white,font=\small , inner sep=-0.05cm]
\tikzstyle{bag}=[draw=black, ultra thick, ellipse, minimum width=12mm, minimum height=12mm]
\tikzstyle{node}=[draw=black, fill=black, ultra thick, circle, minimum size=1mm, inner sep=0 ]
\tikzstyle{every state}=[fill=white,draw=black,text=black,font=\small , inner sep=0.05cm, minimum size=0.9cm]
\renewcommand{\distone}{5cm*0.3}

\def\xbias{2}
\def\ybias{3}
\def\textbias{0.6}

\renewcommand{\distone}{5cm*0.42}

\node[bag, ] (b1) at (0,0) {};
\node[bag, right of=b1] (b2) {};
\node[bag, right of=b2] (b3) {};
\node[bag, right of=b3] (b4) {};
\draw[-, very thick] (b1) to (b2) to (b3) to (b4);

\node[] (bt1) at (0,-0.85) {$\Bag_1$};
\node[ right of=bt1] (bt2) {$\Bag_2$};
\node[ right of=bt2] (bt3) {$\Bag_3$};
\node[ right of=bt3] (bt4) {$\Bag_4$};

\node[node] (n1) at (0.2,0.3) {};
\node[node, right of=n1] (n2) {};
\node[node, right of=n2] (n3) {};
\node[node, right of=n3] (n4) {};

\node[node] (m1) at (0.2,-0.3) {};
\node[node, right of=m1] (m2) {};
\node[node, right of=m2] (m3) {};
\node[node, right of=m3] (m4) {};

\draw[->,  thick] (n1) to (m1);
\draw[->,  thick] (n2) to (m2);
\draw[->,  thick] (n3) to (m3);
\draw[->,  thick] (n4) to (m4);

\node[] (t1) at (-0.2,0.3) {$u$};
\node[ right of=t1] (t2) {$x_2$};
\node[ right of=t2] (t3) {$x_3$};
\node[ right of=t3] (t4) {$x_4$};

\node[] (r1) at (-0.2,-0.3) {$x_2$};
\node[ right of=r1] (r2) {$x_3$};
\node[ right of=r2] (r3) {$x_4$};
\node[ right of=r3] (r4) {$v$};

\end{tikzpicture}
\caption{Illustration of \cref{lem:u_v_dist}.
If $P$  is the unique simple path $\Bag_1\Path \Bag_4$ in $\Tree(G)$,
then there exist (not necessarily distinct) $x_i\in \Bag_{i-1} \cap \Bag_i$ with $1<i\leq 4$ such that 
$\Distance(u,v)=\Distance(u, x_2) \otimes \Distance(x_2,x_3) \otimes \Distance(x_3,x_4) \otimes \Distance(x_4,v)$.
 }\label{fig:uv_dist}
\end{figure}
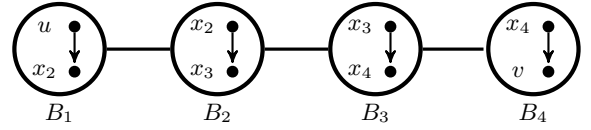

\smallskip\noindent{\bf Informal description of the preprocessing.}
The preprocessing phase of $\ConcurAP$ is handled by algorithm $\Concpreprocessalgo$,
which performs the following steps.
\begin{compactenum}
\item First, the \emph{partial expansion} $\PNode{G}$ of $G$ is constructed by introducing
a pair of strictly partial nodes $\PNode{u}^1$, $\PNode{u}^2$ for every strictly partial node $\PNode{u}$ of $G$,
and edges between strictly partial nodes and the corresponding nodes of $G$ that refine them.

\item Second, the concurrent tree-decomposition $T=\ProductTree(G)$ of $G$ is constructed, and modified
to a tree-decomposition $\PNode{T}$ of the partial expansion graph $\PNode{G}$.

\item Third, a standard, two-way pass of $\PNode{T}$ is performed to compute \emph{local distances}.
%Intuitively, local distances are computed for partial nodes that belong to the same bag.
In this step, for every bag $\PNode{B}$ in $\PNode{T}$ and all partial nodes $\PNode{u}, \PNode{v} \in \PNode{B}$,
the distance $\Distance(\PNode{u},\PNode{v})$ is computed (i.e., all-pair distances in $\PNode{\Bag}$).
Since we compute distances between nodes that are \emph{local} in a 
bag, this step is called local distance computation.
This information is used to handle (i)~single-source queries and (ii)~partial pair queries in which both nodes are strictly partial.

\item Finally, a top-down pass of $\PNode{T}$ is performed in which 
for every node $u$ and partial node $\PNode{v}\in \AncestV{\PNode{V}}_{\PNode{T}}(\PNode{\Bag}_u)$
(i.e., $\PNode{v}$ appears in some ancestor of $\PNode{\Bag}_u$)
the distances $\Distance(u, \PNode{v})$ and $\Distance(\PNode{v}, u)$ are computed.
This information is used to handle pair queries in which at least one node is a node of $G$
(i.e., not strictly partial).
\end{compactenum}

\smallskip\noindent{\bf Algorithm $\Concpreprocessalgo$.}
We now formally describe algorithm $\Concpreprocessalgo$ for preprocessing
the concurrent graph $G=(V,E)$ for the purpose of answering algebraic path queries.
%For any desired $0<\eps\leq 1$, choose a $\lambda\in \Nats$ and $\delta\in \Reals$ such that $\lambda\geq 4/\eps$
%and $\delta\leq \eps/18$.
%Let $\alpha=\Bagfactor$, $\beta=\Balfactor$ and $\gamma=\Levelfactor$.
For any desired $0<\eps\leq 1$, we choose appropriate constants $\alpha$, $\beta$, $\gamma$, 
which will be defined later for the complexity analysis.
On input $G=(V,E)$, where $G$ is a concurrent graph of $k$ constant-treewidth graphs
$(G_i=(V_i, E_i))_{1\leq i\leq k}$, and a weight function $\Weight:E\rightarrow \Sigma$, $\Concpreprocessalgo$ operates as follows:
\begin{compactenum}
\item\label{item:partial_exp} Construct the \emph{partial expansion} $\PNode{G}=(\PNode{V}, \PNode{E})$ of $G$
together with an extended weight function $\PNode{\Weight}: \PNode{E}\rightarrow \Sigma$ as follows.
\begin{compactenum}
\item The node set is $\PNode{V}= V \cup \{\PNode{u}^1, \PNode{u}^2: \exists u\in V \text{ s.t. } u\Strictrefines \PNode{u} \}$; i.e.,
$\PNode{V}$ consists of nodes in $V$ and two copies for every partial node $\PNode{u}$ that is strictly refined by a node $u$ of $G$.

\item The edge set is $\PNode{E} = E\cup \{(\PNode{u}^1, u), (u, \PNode{u}^2): \PNode{u}^1, \PNode{u}^2\in \PNode{V} \text { and } u\in V \text{ s.t. } u\Strictrefines \PNode{u}^1, \PNode{u}^2\}$,
i.e., along with the original edges $E$,
the first (resp. second) copy of every strictly partial node has outgoing (resp. incoming) edges to (resp. from) the nodes of $G$ that refine it.
\item For the weight function we have $\PNode{\Weight}(\PNode{u},\PNode{v})=\Weight(\PNode{u},\PNode{v})$ if $\PNode{u},\PNode{v}\in V$, and $\PNode{\Weight}(\PNode{u},\PNode{v})=\One$ otherwise. 
That is, the original weight function is extended with value $\One$ (which is neutral for semiring multiplication) 
to all new edges in $\PNode{G}$.
\end{compactenum}
\item\label{item:partial_tree_dec} Construct the tree-decomposition $\PNode{T}=({\PNode{V}_T, \PNode{E}_T})$ of $\PNode{G}$ as follows.
\begin{compactenum}

\item\label{item:tree_decomp} Obtain an $(\alpha, \beta, \gamma)$ tree-decomposition $T_i=\Tree(G_i)$ of every graph $G_i$ 
using \cref{them:tree_decomp}.

\item Construct the concurrent tree-decomposition $T=\ProductTree(G)$ of $G$ using $(T_i)_{1\leq i\leq k}$.

\item Let $\PNode{T}$ be identical to $T$, with the following exception:
For every bag $\Bag$ of $T$ and $\PNode{\Bag}$ the corresponding bag in $\PNode{T}$,
for every node $u\in \Bag$, insert in $\PNode{\Bag}$ all strictly partial nodes $\PNode{u}^1$, $\PNode{u}^2$ of $\PNode{V}$ that $u$ refines.
Formally, set $\PNode{\Bag}=  \Bag\cup \{\PNode{u}^1, \PNode{u}^2: \exists u\in \Bag \text{ s.t. } u\Strictrefines\PNode{u}\}$.
Note that also $u\in \PNode{\Bag}$.
\end{compactenum}
Observe that the root bag of $\PNode{T}$ contains all strictly partial nodes.
\item\label{item:ld} Perform the \emph{local distance computation} on $\PNode{T}$ as follows.
For every partial node $\PNode{u}$, maintain two map data-structures
$\Fwd_{\PNode{u}}, \Bwd_{\PNode{u}}:\PNode{\Bag}_{\PNode{u}}\rightarrow \Sigma$.
Intuitively, $\Fwd_{\PNode{u}}$ (resp. $\Bwd_{\PNode{u}}$) aims to store the forward 
(resp., backward) distance, i.e., distance from (resp., to) $\PNode{u}$ to (resp. from) 
vertices in $\PNode{\Bag}_{\PNode{u}}$.
Initially set $\Fwd_{\PNode{u}}(\PNode{v})=\PNode{\Weight}(\PNode{u},\PNode{v})$ and $\Bwd_{\PNode{u}}(\PNode{v})=\PNode{\Weight}(\PNode{v},\PNode{u})$
for all partial nodes $\PNode{v}\in \PNode{\Bag}_{\PNode{u}}$
(and $\Fwd_{\PNode{u}}(\PNode{v})=\Bwd_{\PNode{u}}(\PNode{v})=\Zero$ if $(\PNode{u}, \PNode{v})\not\in \PNode{E}$).
%At the end of this step, the values $\Fwd_{\PNode{u}}(\PNode{v})$ and $\Bwd_{\PNode{u}}(\PNode{v})$
%will equal the distances $\Distance(\PNode{u},\PNode{v})$ and $\Distance(\PNode{v},\PNode{u})$.
At any point in the computation, given a bag $\PNode{\Bag}$ we 
denote by $\Weightmap_{\PNode{\Bag}}:\PNode{\Bag}\times \PNode{\Bag}\rightarrow \Sigma$
a map data-structure such that for every pair of partial nodes $\PNode{u}, \PNode{v}$ with $\Level(\PNode{v})\leq \Level(\PNode{u})$ we have
$\Weightmap_{\PNode{\Bag}}(\PNode{u},\PNode{v})= \Fwd_{\PNode{u}}(\PNode{v})$ and $\Weightmap_{\PNode{\Bag}}(\PNode{v},\PNode{u})= \Bwd_{\PNode{u}}(\PNode{v})$.
\begin{compactenum}
\item\label{item:bottom_up} Traverse $\PNode{T}$ bottom-up, and for every bag $\PNode{\Bag}$,
execute an all-pairs algebraic path computation on $\PNode{G}\restr{\PNode{\Bag}}$
with weight function $\Weightmap_{\PNode{\Bag}}$.
This is done using classical algorithms for the transitive closure, e.g.~\cite{Lehmann77,Floyd62,Warshall62,Kleene56}.
For every pair of partial nodes $\PNode{u},\PNode{v}$ with $\Level(\PNode{v})\leq \Level(\PNode{u})$, 
set $\Bwd_{\PNode{u}}(\PNode{v})=\Distancecomp(\PNode{v},\PNode{u})$ and $\Fwd_{\PNode{u}}(\PNode{v})=\Distancecomp(\PNode{u},\PNode{v})$,
where $\Distancecomp(\PNode{u},\PNode{v})$ and $\Distancecomp(\PNode{v},\PNode{u})$ are the computed distances in $\PNode{G}\restr{\PNode{\Bag}}$.
\item\label{item:top_down} Traverse $\PNode{T}$ top-down, and for every bag $\PNode{\Bag}$
perform the computation of \cref{item:bottom_up}.
\end{compactenum}
\item\label{item:ancestors} Perform the \emph{ancestor distance computation} on $\PNode{T}$ as follows.
For every node $u$, maintain two map data-structures 
$\From_{u}, \To_{u}: \AncestV{\PNode{V}}_{\PNode{T}}(\PNode{\Bag}_u) \to \Sigma$
from partial nodes that appear in the ancestor bags of $\PNode{\Bag}_u$ to $\Sigma$.
These maps aim to capture distances between the node $u$ and nodes in the ancestor bags of $\PNode{\Bag}_{u}$
(in contrast to $\Fwd_{u}$ and $\Bwd_{u}$ which store distances only between $u$ and nodes in $\PNode{\Bag}_{u}$).
Initially, set $\From_{u}(\PNode{v}) = \Fwd_{u}(\PNode{v})$ and $\To_{u}(\PNode{v}) = \Bwd_{u}(\PNode{v})$ 
for every partial node $\PNode{v}\in \PNode{\Bag}_u$.
%At the end of this step, the values $\From_{u}(\PNode{v})$ and $\To_{u}({\PNode{v}})$
%will equal the distances $\Distance(u,\PNode{v})$ and $\Distance(\PNode{v}, u)$.
Given a pair of partial nodes $\PNode{u}, \PNode{v}$ with $\Level(\PNode{v})\leq \Level(\PNode{u})$ we denote by
$\Weightmap^{+}(\PNode{u},\PNode{v})= \From_{\PNode{u}}(\PNode{v})$ and $\Weightmap^{+}(\PNode{v},\PNode{u})= \To_{\PNode{u}}(\PNode{v})$.
Traverse $\PNode{T}$ via a DFS starting from the root, and for every encountered bag $\PNode{\Bag}$
with parent $\PNode{\Bag}'$, for every node $u$ such that $\PNode{\Bag}$ is the root bag of $u$,
for every partial node $\PNode{v}\in\AncestV{\PNode{V}}_{\PNode{T}}(\PNode{\Bag}_u)$, assign
\[
\From_{u}(\PNode{v}) = \bigoplus_{x\in \PNode{\Bag}\cap \PNode{\Bag}'} \Fwd_{u}(x) \otimes \Weightmap^{+}(x, \PNode{v}) \numberthis \label{eq:from}
\]
\[
\To_{u}(\PNode{v}) = \bigoplus_{x\in \PNode{\Bag}\cap \PNode{\Bag}'} \Bwd_{u}(x) \otimes \Weightmap^{+}(\PNode{v}, x) \numberthis \label{eq:to}
\]
If $\PNode{\Bag}$ is the root of $\PNode{T}$, simply initialize the maps $\From_{u}$ and $\To_{u}$
according to the corresponding maps $\Fwd_{u}$ and $\Bwd_{u}$ constructed from \cref{item:ld}.

\item Preprocess $\PNode{T}$ to answer LCA queries in $O(1)$ time~\cite{Tarjan84}.
\end{compactenum}
%The pseudocode of $\Concpreprocessalgo$ can be found in \cref{sec:algorithms}. 
The following claim states that the first (resp. second) copy of each strictly partial node inserted in \cref{item:partial_exp}
captures the distance from (resp. to) the corresponding strictly partial node of $\PNode{G}$.

\smallskip
\begin{restatable}{claim}{pnodedistances}\label{prop:pnode_distances}
For every partial node $\PNode{u}$ and strictly partial node $\PNode{v}$ we have
$\Distance(\PNode{u}, \PNode{v}) = \Distance(\PNode{u}, \PNode{v}^2)$ and
$\Distance(\PNode{v}, \PNode{u}) = \Distance(\PNode{v}^1, \PNode{u})$.
\end{restatable}

\smallskip\noindent{\bf Key novelty and insights.} 
The key novelty and insights of our algorithm are as follows:
\begin{compactenum}
\item A partial pair query can be answered by breaking it down to several pair queries. 
Instead, preprocessing the partial expansion of the concurrent graph
allows to answer partial pair queries directly. 
Moreover, the partial expansion does not increase the asymptotic complexity of
the preprocessing time and space.
%%asymptotic overhead in the preprocessing time and space.
%The key reason to consider the partial expansion of the graphs is to 
%answer partial queries faster than breaking them down to several queries. 

\item $\Concpreprocessalgo$ computes the transitive closure only during the local
distance computation in each bag (\cref{item:ld} above), instead of 
a global computation on the whole graph. 
The key reason of our algorithmic improvement lies on the fact that the local computation is cheaper than 
the global computation, and is also sufficient to handle queries fast.

\item The third key aspect of our algorithm is the strongly balanced tree decomposition,
which is crucially used in \cref{them:product} to construct a tree decomposition for the 
concurrent graph such that the size of the bags decreases geometrically along the levels. 
By using the cheaper local distance computation (as opposed to the transitive closure globally)
and recursing on a geometrically decreasing series we obtain the desired complexity bounds for our algorithm. 
Both the strongly balanced tree decomposition and the fast local distance computation play important roles in our algorithmic improvements.
\end{compactenum}

We now turn our attention to the analysis of $\Concpreprocessalgo$.
%The following lemma is proved in \cref{sec:concurrent_proofs}.
%We start with showing that $\PNode{T}$ is a tree decomposition of $\PNode{G}$.

\smallskip
\begin{restatable}{lemma}{partialtreedecomp}\label{lem:partial_tree_decomp}
$\PNode{T}$ is a tree decomposition of the partial expansion $\PNode{G}$.
\end{restatable}

In \cref{lem:preprocess_correctness} we establish that the forward and backward maps computed
by $\Concpreprocessalgo$ store the distances between nodes.

\smallskip
\begin{lemma}\label{lem:preprocess_correctness}
At the end of $\Concpreprocessalgo$, the following assertions hold:
\begin{compactenum}
\item For all nodes $u,v\in V$ such that $\PNode{\Bag}_u$ appears in $\PNode{T}(\PNode{\Bag}_v)$, 
we have $\From_u(v)=\Distance(u,v)$ and $\To_u(v)=\Distance(v,u)$.
\item For all strictly partial nodes $\PNode{v}\in \PNode{V}$ and nodes $u\in V$
we have  $\From_u({\PNode{v}^2})=\Distance(u, \PNode{v})$ and  $\To_u({\PNode{v}^1})=\Distance(\PNode{v}, u)$.
\item For all strictly partial nodes $\PNode{u},\PNode{v}\in \PNode{V}$ we have
$\Fwd_{\PNode{u}^1}(\PNode{v}^2)=\Distance(\PNode{u}, \PNode{v})$ and
$\Bwd_{\PNode{u}^2}(\PNode{v}^1)=\Distance(\PNode{v},\PNode{u})$.
\end{compactenum}
\end{lemma}
\begin{proof}
We describe the key invariants that hold during the traversals of $\PNode{T}$ by $\Concpreprocessalgo$
in \cref{item:bottom_up}, \cref{item:top_down} and \cref{item:ancestors} after the algorithm processes a bag $\PNode{\Bag}$.
All cases depend on \cref{lem:u_v_dist}.
\begin{compactenum}
\item[\em \cref{item:bottom_up}] For every pair of partial nodes $\PNode{u}$, $\PNode{v}\in \PNode{\Bag}$
such that $\Level(\PNode{v})\leq \Level(\PNode{u})$ we have $\Fwd_{\PNode{u}}(\PNode{v})=\bigoplus_{P_1}\otimes(P_{1})$ 
and $\Bwd_{\PNode{u}}(\PNode{v})=\bigoplus_{P_2}\otimes(P_{2})$ where $P_1$ and $P_2$ are $\PNode{u}\Path \PNode{v}$
and $\PNode{v}\Path\PNode{u}$ paths respectively that only traverse nodes in $\PNode{V}_{\PNode{T}}(\PNode{\Bag})$.
The statement follows by a straightforward induction on the levels processed by the algorithm in the bottom-up pass.
Note that if $\PNode{u}$ and $\PNode{v}$ are partial nodes in the root of $\PNode{T}$, the statement yields
$\Fwd_{\PNode{u}}(\PNode{v})=\Distance(u,v)$ and $\Bwd_{\PNode{u}}(\PNode{v})=\Distance(v,u)$.
\item[\em \cref{item:top_down}] The invariant is similar to the previous, except that $P_1$ and $P_2$
range over all $\PNode{u}\Path \PNode{v}$ and $\PNode{v}\Path\PNode{u}$ paths in $\PNode{G}$ respectively.
Hence now $\Fwd_{\PNode{u}}(\PNode{v})=\Distance(\PNode{u}, \PNode{v})$ and 
$\Bwd_{\PNode{u}}(\PNode{v})=\Distance(\PNode{v}, \PNode{u})$.
The statement follows by a straightforward induction on the levels processed by the algorithm in the top-down pass.
Note that the base case on the root follows from the previous item, where the maps $\Bwd$ and $\Fwd$ store actual distances.
\item[\em \cref{item:ancestors}] For every node $u\in \PNode{\Bag}$ and partial node $\PNode{v}\in \AncestV{\PNode{V}}_{\PNode{T}}(\PNode{\Bag})$ we have
$\From_{u}(\PNode{v})=\Distance(u, \PNode{v})$ and 
$\To_{u}(\PNode{v})=\Distance(\PNode{v},u)$.
The statement follows from \cref{lem:u_v_dist} and a straightforward induction on the length of the path from the root of $\PNode{T}$
to the processed bag $\PNode{\Bag}$.
\end{compactenum}
Statement~1 of the lemma follows from \cref{item:ancestors}.
Similarly for statement~2, together with the observation that every strictly partial node $\PNode{v}$ 
appears in the root of $\PNode{T}$, and thus $\PNode{v}\in \AncestV{\PNode{V}}_{\PNode{T}}(\PNode{\Bag}_u)$.
Finally, statement~3 follows again from the fact that all strictly partial nodes appear in the root bag of $\PNode{T}$.
The desired result follows.
\end{proof}

We now consider the complexity analysis, and we start with a technical lemma on recurrence relations.

\smallskip
\begin{restatable}{lemma}{concrecurrence}\label{lem:conc_recurrence}
Consider the recurrences in \cref{eq:rec1} and \cref{eq:rec2}.
%The solutions of \cref{eq:rec1} and \cref{eq:rec2} are, respectively
\begin{align*}
\Time_k(n) &\leq n^{3\cdot (k-1)} + 2^{\lambda\cdot k} \cdot \Time_k\left(n\cdot \left(\frac{1+\delta}{2}\right)^{\lambda-1}\right)
\numberthis \label{eq:rec1}
\end{align*}
\begin{align*}
\Space_k(n) &\leq n^{2\cdot (k-1)} + 2^{\lambda\cdot k} \cdot \Space_k\left(n\cdot \left(\frac{1+\delta}{2}\right)^{\lambda-1}\right)
\numberthis \label{eq:rec2}
\end{align*}
Then
\begin{compactenum}
\item $\Time_k(n)=O(n^{3\cdot (k-1)})$, and
\item (i)~$\Space_k(n)=O(n^{2\cdot(k-1)})$ if $k\geq 3$, and (ii)~$\Space_2(n)=O(n^{2+\eps})$.
\end{compactenum}
\end{restatable}

The proof of \cref{lem:conc_recurrence} is technical, and presented in \cref{sec:concurrent_proofs}.
The following lemma analyzes the complexity of $\Concpreprocessalgo$, and makes use of the above recurrences.
Recall that $\Concpreprocessalgo$ takes as part of its input a desired constant $0<\eps\leq 1$.
We choose a $\lambda\in \Nats$ and $\delta\in \Reals$ such that $\lambda\geq 4/\eps$ and $\delta\leq \eps/18$.
Additionally, we set  $\alpha=\Bagfactor$, $\beta=\Balfactor$ and $\gamma=\Levelfactor$,
which are the constants used for constructing an $(\alpha, \beta, \gamma)$ tree-decomposition
$T_i=\Tree(G_i)$ in \cref{item:tree_decomp} of $\Concpreprocessalgo$.

\smallskip
\begin{lemma}\label{lem:preproces_complexity}
$\Concpreprocessalgo$ requires $O(n^{2\cdot k-1})$ space and
\begin{inparaenum}
\item $O(n^{3\cdot(k-1)})$ time if $k\geq 3$, and
\item $O(n^{3+\eps})$ time if $k=2$.
\end{inparaenum}

\end{lemma}
\begin{proof}
We examine each step of the algorithm separately.
\begin{compactenum}
\item The time and space required for this step is bounded by the number of nodes introduced in the partial expansion $\PNode{G}$,
which is $2\cdot \sum_{i<k} {n \choose i}=O(n^{k-1})$.
\item By \cref{them:product}, $\ProductTree(G)$ is constructed in $O(n^k)$ time and space.
In $\PNode{T}$, the size of each bag $\PNode{\Bag}$ is increased by constant factor, hence this step requires $O(n^k)$ time and space.
\item In each pass, $\Concpreprocessalgo$ spends $|\PNode{\Bag}|^3$ time to perform an all-pairs algebraic paths computation in each bag $\PNode{\Bag}$ of $\PNode{T}$~\cite{Lehmann77,Floyd62,Warshall62,Kleene56}.
The space usage for storing all maps $\Fwd_{\PNode{u}}$ and $\Bwd_{\PNode{u}}$ for every node $\PNode{u}$ whose root bag is $\PNode{\Bag}$
is $O(|\PNode{\Bag}|^2)$, since there are at most $|\PNode{\Bag}|$ such nodes $\PNode{u}$, and each map has size $|\PNode{\Bag}|$.
By the previous item, we have $|\PNode{\Bag}|=O(|\Bag|)$, where $\Bag$ is the corresponding bag of $T$ before the partial expansion of $G$.
By \cref{them:product}, we have $|\Bag|=O(n^{k-1}\cdot \beta^{i})$, where $\Level(\Bag)\geq i\cdot \gamma=i\cdot \Levelfactor$,
and $\beta=\Balfactor$. Then, since $\PNode{T}$ is a full $2^k$-ary tree,
the time and space required for preprocessing every $\gamma=\lambda$ levels of $\PNode{T}$
is given by the following recurrences respectively (ignoring constant factors for simplicity).

\[
\Time_k(n) \leq n^{3\cdot (k-1)} + 2^{\lambda\cdot k} \cdot \Time_k\left(n\cdot \left(\frac{1+\delta}{2}\right)^{\lambda-1}\right)
\]
\[
\Space_k(n) \leq n^{2\cdot (k-1)} + 2^{\lambda\cdot k} \cdot \Space_k\left(n\cdot \left(\frac{1+\delta}{2}\right)^{\lambda-1}\right)
\]

By the analysis of \cref{eq:rec1} and \cref{eq:rec2} of \cref{lem:conc_recurrence}, we have that
$\Time_k(n)=O(n^{3\cdot (k-1)})$  and (i)~$\Space_k(n)=O(n^{2\cdot(k-1)})$ if $k\geq 3$, and (ii)~$\Space_2(n)=O(n^{2+\eps})$.

\item
We first focus on the space usage. Let $\PNode{\Bag}_u^i$ denote the ancestor bag of $\PNode{\Bag}_u$ at level $i$.
We have 
\begin{align*}
|\AncestV{\PNode{V}}_{\PNode{T}}(\PNode{\Bag}_u)| & =
\sum_i |\PNode{\Bag}_u^i|  \leq c_1\cdot \sum_i |\PNode{\Bag}_u^{\lfloor i/\gamma \rfloor }|\leq 
c_2\cdot   \sum_i  |\Bag_u^{\lfloor i/\gamma \rfloor }|\\
&\leq c_3 \cdot \sum_i \left(n^{k-1} \cdot \beta^i\right)  = O(n^{k-1})
\end{align*}
for some constants $c_1, c_2, c_3$.
The first inequality comes from expressing the size of all (constantly many) ancestors $\PNode{\Bag}_u^i$
with $\lfloor i/\gamma \rfloor =j$ as a constant factor the size of $\PNode{\Bag}_u^{\lfloor i/\gamma \rfloor }$.
The second inequality comes from \cref{item:partial_exp} of this lemma, which states that $O(|\PNode{\Bag}|)=O(|\Bag|)$ for every bag $\PNode{\Bag}$. The third inequality comes from \cref{them:product}.
By \cref{item:partial_tree_dec}, there are $O(n^k)$ such nodes $u$ in $\PNode{T}$, hence the space required is $O(n^{2\cdot k-1})$.

We now turn our attention to the time requirement.
For every bag $\PNode{\Bag}$, the algorithm requires $O(|\PNode{\Bag}|^2)$ time to iterate over all pairs of nodes
$u$ and $x$ in \cref{eq:from} and \cref{eq:to} to compute the values $\From_u(\PNode{v})$ and $\To_u(\PNode{v})$
for every $\PNode{v}\in \AncestV{\PNode{V}}_{\PNode{T}}(\PNode{\Bag})$. 
Hence the time required for all nodes $u$ and one partial node $\PNode{v}\in \AncestV{\PNode{V}}_{\PNode{T}}(\PNode{\Bag})$ to store the maps values $\From_u(\PNode{v})$ and $\To(\PNode{v})$ is given by the recurrence
\[
\Time_k(n) \leq n^{2\cdot (k-1)} + 2^{\lambda\cdot k} \cdot \Time_k\left(n\cdot \left(\frac{1+\delta}{2}\right)^{\lambda-1}\right)
\]
The analysis of \cref{eq:rec1} and \cref{eq:rec2} of \cref{lem:conc_recurrence}
gives $\Time_k(n)=O(n^{2\cdot (k-1)})$ for $k\geq 3$ and $\Time_2(n)=O(n^{2+\eps})$
(i.e., the above time recurrence is analyzed as the recurrence for $\Space_k$ 
of \cref{lem:conc_recurrence}).
From the space analysis we have that there exist $O(n^{k-1})$ partial nodes $\PNode{v}\in \AncestV{\PNode{V}}_{\PNode{T}}(\PNode{\Bag})$
for every node $u$ whose root bag is $\PNode{\Bag}$.
Hence the total time for this step is $O(n^{3\cdot (k-1)})$ for $k\geq 3$, and  $O(n^{3+\eps})$ for $k=2$.
\item This step requires time linear in the size of $\PNode{T}$~\cite{Tarjan84}.
\end{compactenum}
The desired result follows.
\end{proof}

\smallskip\noindent{\bf Algorithm $\Concqueryalgo$.}
In the query phase, $\ConcurAP$ answers distance queries using the algorithm $\Concqueryalgo$.
We distinguish three cases, according to the type of the query. 
\begin{compactenum}
\item\label{item:ss_query} \emph{Single-source query.} Given a source node $u$,
initialize a map data-structure $A:V \rightarrow \Sigma$, and initially set $A(v)= \Fwd_u(v)$ for all $v\in \PNode{\Bag}_u$,
and $A(v)=\Zero$ for all other nodes $v\in V\setminus \PNode{\Bag}_u$.
Perform a BFS on $\PNode{T}$ starting from $\PNode{\Bag}_{u}$, 
and for every encountered bag $\PNode{\Bag}$ and nodes $x,v\in \PNode{\Bag}$ with $\Level(v)\leq \Level(x)$,
set $A(v)=A(v)\oplus (A(x) \otimes \Fwd_x(v))$.
Return the map $A$.
\item\label{item:p_query} \emph{Pair query.} Given two nodes $u,v\in V$, find the LCA $\PNode{\Bag}$ of bags $\PNode{\Bag}_u$ and $\PNode{\Bag}_v$.
Return $\bigoplus_{x\in \PNode{\Bag}\cap V} (\From_u(x) \otimes \To_v(x))$.
\item\label{item:pp_query} \emph{Partial pair query.} Given two partial nodes $\PNode{u}, \PNode{v}$, 
\begin{compactenum}
\item\label{item:spp_query} If both $\PNode{u}$ and $\PNode{v}$ are strictly partial, return $\Fwd_{\PNode{u}^1}(\PNode{v}^2)$, else
\item If $\PNode{u}$ is strictly partial, return $\To_{v}(\PNode{u}^1)$, else
\item Return $\From_{u}(\PNode{v}^2)$.
\end{compactenum}
\end{compactenum}
%The pseudocode of $\Concqueryalgo$ can be found in \cref{sec:algorithms}. 

We thus establish the following theorem.

\smallskip
\begin{theorem}\label{them:concurrent}
Let $G=(V,E)$ be a concurrent graph of $k$ constant-treewidth graphs $(G_i)_{1\leq i\leq k}$,
and $\Weight:E\to \Sigma$ a weight function of $G$.
For any fixed $\eps>0$, 
the data-structure $\ConcurAP$ correctly answers single-source and pair queries and requires:
\begin{compactenum}
\item Preprocessing time\\ 
\begin{inparaenum}
\item $O(n^{3\cdot(k-1)})$ if $k\geq 3$, and~~
\item $O(n^{3+\eps})$ if $k=2$.
\end{inparaenum}
\item Preprocessing space $O(n^{2\cdot k-1})$.
\item Single-source query time\\
\begin{inparaenum}
\item $O(n^{2\cdot(k-1)})$ if $k\geq 3$, and~~
\item$O(n^{2+\eps})$ if $k=2$.
\end{inparaenum}
\item Pair query time $O(n^{k-1})$.
\item Partial pair query time $O(1)$.
\end{compactenum} 
\end{theorem}
\begin{proof}
The correctness of $\Concqueryalgo$ for handling all queries follows from \cref{lem:u_v_dist}
and the properties of the preprocessing established in \cref{lem:preprocess_correctness}.
The preprocessing complexity is stated in \cref{lem:preproces_complexity}.
The time complexity for the single-source query comes from the observation that $\Concqueryalgo$ spends
quadratic time in each encountered bag, and the result follows from the 
recurrence analysis of \cref{eq:rec2} in \cref{lem:conc_recurrence}.
The time complexity for the pair query follows from the $O(1)$ time to access the LCA bag $\PNode{\Bag}$ of $\PNode{\Bag}_u$ and $\PNode{\Bag}_v$,
and the $O(|\PNode{\Bag}|)=O(n^{k-1})$ time required to iterate over all nodes $x\in \PNode{\Bag}\cap V$.
Finally, the time complexity for the partial pair query follows from the $O(1)$ time lookup in the constructed maps $\Fwd$, $\From$ and $\To$.
\end{proof}

Note that a single-source query from a strictly partial node $\PNode{u}$
can be answered in $O(n^{k})$ time by breaking it down to $n^k$ partial pair queries.

The most common case in analysis of concurrent programs is that of two threads,
for which we obtain the following corollary.

\smallskip
\begin{corollary}\label{cor:concurrent}
Let $G=(V,E)$ be a concurrent graph of two constant-treewidth graphs $G_1, G_2$,
and $\Weight:E\to \Sigma$ a weight function of $G$.
For any fixed $\eps>0$,
the data-structure $\ConcurAP$ correctly answers single-source and pair queries and requires:
\begin{compactenum}
\item Preprocessing time $O(n^{3+\eps})$.
\item Preprocessing space $O(n^{3})$.
\item Single-source query time $O(n^{2+\eps})$.
\item Pair query time $O(n)$.
\item Partial pair query time $O(1)$.
\end{compactenum} 
\end{corollary}

\smallskip
\begin{remark}
In contrast to \cref{cor:concurrent}, the existing methods 
for handling even one pair query require hexic time and quartic space~\cite{Lehmann77,Floyd62,Warshall62,Kleene56}
by computing the transitive closure.
While our improvements are most significant for algebraic path queries, they imply improvements
also for special cases like reachability (expressed in Boolean semirings).
For reachability, the complete preprocessing requires quartic time, 
and without preprocessing every query requires quadratic time.
In contrast, with almost cubic preprocessing we can answer pair (resp., partial pair) queries
in linear (resp. constant) time.
\end{remark}

Note that \cref{item:ancestors} of $\Concpreprocessalgo$ is required for handling pair queries only.
By skipping this step, %%during $\Concpreprocessalgo$,
we can handle every (partial) pair query $\PNode{u},\PNode{v}$ similarly to the single source query from $\PNode{u}$, 
but restricting the BFS to the path $P:\PNode{\Bag}_{\PNode{u}}\Path \PNode{\Bag}_{\PNode{v}}$,
and spending $O(|\PNode{\Bag}|^2)$ time for each bag $\PNode{\Bag}$ of $P$.
Recall (\cref{them:product}) that the size of each bag $\Bag$ in $T$
(and thus the size of the corresponding bag $\PNode{\Bag}$ in $\PNode{T}$)
decreases geometrically every $\gamma$ levels. 
Then, the time required for this operation is $O(|\PNode{\Bag}'|^2)=O(n^2)$,
where $\PNode{\Bag}'$ is the bag of $P$ with the smallest level.
This leads to the following corollary.

\smallskip
\begin{corollary}\label{cor:ld}
Let $G=(V,E)$ be a concurrent graph of two constant-treewidth graphs $G_1, G_2$,
and $\Weight:E\to \Sigma$ a weight function of $G$.
For any fixed $\eps$,
the data-structure $\ConcurAP$ (by skipping \cref{item:ancestors} in $\Concpreprocessalgo$) 
correctly answers single-source and pair queries and requires:
\begin{compactenum}
\item Preprocessing time $O(n^{3})$.
\item Preprocessing space $O(n^{2+\eps})$.
\item Single-source query time $O(n^{2+\eps})$.
\item Pair and partial pair query time $O(n^2)$.
\end{compactenum} 
\end{corollary}

Finally, we can use $\ConcurAP$ to obtain the transitive closure of $G$
by performing $n^2$ single-source queries.
The preprocessing space is $O(n^{2+\eps})$ by \cref{cor:ld},
and the space of the output is $O(n^4)$, since there are $n^4$ pairs for the computed distances.
Hence the total space requirement is $O(n^4)$.
The time requirement is $O(n^{4+\eps})$, since by \cref{cor:ld}, every single-source query requires $O(n^{2+\eps})$ time.
We obtain the following corollary.

\smallskip
\begin{corollary}\label{cor:closure}
Let $G=(V,E)$ be a concurrent graph of two constant-treewidth graphs $G_1, G_2$,
and $\Weight:E\to \Sigma$ a weight function of $G$.
For any fixed $\eps>0$, the transitive closure of $G$ wrt $\Weight$ can be computed in $O(n^{4+\eps})$ time and $O(n^4)$ space.
\end{corollary}

%s\input{tradeoff}
\section{Conditional Optimality for Two Graphs}\label{sec:arbitrary_product}

In the current section we establish the optimality of \cref{cor:ld}
in handling algebraic path queries in a concurrent graph that consists of two 
constant-treewidth components.
%\smallskip\noindent{\bf Key idea.}
The key idea is to show that for any arbitrary graph (i.e., without the constant-treewidth restriction) 
$G$ of $n$ nodes, we can construct a concurrent graph $G'$ as a 
$2$-self-concurrent asynchronous composition of a constant-treewidth 
graph $G''$ of $2\cdot n$ nodes, such that semiring queries in $G$ coincide 
with semiring queries in $G'$.

\smallskip\noindent{\bf Arbitrary graphs as composition of two constant-treewidth graphs.}
We fix an arbitrary graph $G=(V, E)$ of $n$ nodes, and a weight function $\Weight: E\to \Sigma$.
Let $x_i$, $1\leq i\leq n$ range over the nodes $V$ of $G$, and construct a graph $G''=(V'', E'')$
such that $V''=\{x_i, y_i:~1\leq i\leq n\}$ and $E''=\{(x_i, y_i), (y_i,x_i):~1\leq i\leq n\}\cup \{(y_i, y_{i+1}), (y_{i+1}, y_i):~1\leq i <n\}$.

\smallskip
\begin{restatable}{claim}{twclaim}\label{claim:tw}
The treewidth of $G''$ is $1$.
\end{restatable}

Given $G''$, we construct a graph $G'$ as a $2$-self-concurrent asynchronous composition of $G''$.
Informally, a node $x_i$ of $G$ corresponds to the node $\Prod{x_i, x_i}$ of $G'$.
An edge $(x_i, x_j)$ in $G$ is simulated by two paths in $G'$.
\begin{compactenum}
\item The first path has the form $P_1:\Prod{x_i, x_i}\Path \Prod{x_i, x_j}$,
and is used to witness the weight of the edge in $G$, i.e., $\Weight(x_i, x_j)=\otimes(P_1)$.
It traverses a sequence of nodes, where the first constituent is fixed to $x_i$,
and the second constituent forms the path $x_i\to y_i\to y_{i'}\to \dots \to y_j \to x_j$.
The last transition will have weight equal to $\Weight(x_i, x_j)$, and the other transitions have weight $\One$.
\item The second path has the form $P_2:\Prod{x_i, x_j}\Path \Prod{x_j, x_j}$,
it has no weight (i.e., $\otimes(P_2)=\One$), and is used to reach the node $\Prod{x_j, x_j}$.
It traverses a sequence of nodes, where the second constituent is fixed to $x_j$,
and the first constituent forms the path $x_i\to y_i\to y_{i'}\to \dots \to y_j \to x_j$.
\end{compactenum}
Then the concatenation of $P_1$ and $P_2$ creates a path $P:\Prod{x_i, x_i}\Path\Prod{x_j, x_j}$ with $\otimes(P)=\otimes(P_1)\otimes \otimes(P_2)=\Weight(x_i, x_j)\otimes \One = \Weight(x_i, x_j)$.

\smallskip\noindent{\bf Formal construction.}
We construct a graph $G'=(V', E')$ as a $2$-self-concurrent asynchronous composition of $G''$,
by including the following edges.
\begin{compactenum}
\item {\em Black edges.} For all $1\leq i\leq n$ and $1\leq j < n$ we have $(\Prod{x_i,y_j}, \Prod{x_i,y_{j+1}}), (\Prod{x_i,y_{j+1}}, \Prod{x_i,y_{j}})\in E'$ , and
for all $1\leq i < n$ and $1\leq j \leq n$ we have $(\Prod{y_i,x_j}, \Prod{y_{i+1},x_j}), (\Prod{y_{i+1},x_{j}}, \Prod{y_i,x_{j}})\in E'$.
\item {\em Blue edges.} For all $1\leq i \leq n$ we have $(\Prod{x_i, x_i}, \Prod{x_i, y_i}), (\Prod{y_i, x_i}, \Prod{x_i, x_i})\in E'$ .
\item {\em Red edges.} For all $(x_i, x_j)\in E$ we have $(\Prod{x_i,y_j}, \Prod{x_i,x_j})\in E'$.
\item {\em Green edges.} For all $1\leq i, j \leq n$ with $i\neq j$ we have $(\Prod{x_i,x_j},\Prod{ y_i,x_j})\in E'$.
\end{compactenum}
Additionally, we construct a weight function such that $\Weight'(\Prod{x_i,y_j}, \Prod{x_i,x_j})=\Weight(x_i, x_j)$ for every 
red edge $(\Prod{x_i,y_j}, \Prod{x_i,x_j})$, and $\Weight'(u,v)=\One$ for every other edge $(u,v)$.
\cref{fig:product} provides an illustration of the construction.
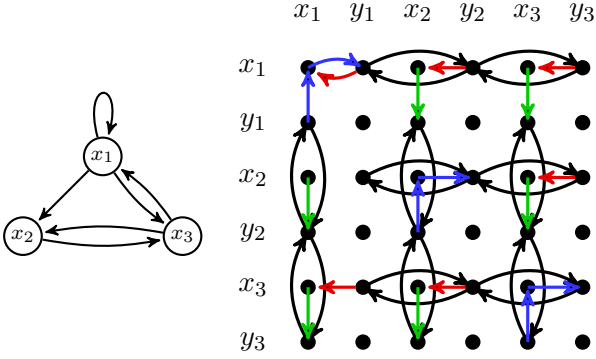
\begin{figure}
\newcommand{\distone}{5cm*0.3}
\def \darkgreen {black!20!green}
\def \darkred {black!10!red}
\def \darkblue {blue!80}
\centering
\begin{tikzpicture}[->,>=stealth',shorten >=1pt,auto,node distance=\distone,
                    thick,scale=1 ]
      
\tikzstyle{every state}=[fill=white,draw=black,text=black,font=\small , inner sep=0.05cm, minimum size=0.5cm]
\tikzstyle{invis}=[fill=white,draw=white,text=white,font=\small , inner sep=-0.05cm]

\def\bend{10}
\def\disttwo{0.73}
%\node[state] (v1) at (0,0) {$1$};
%\node[state,above of=v1] (v8) {$8$};
%\node[state,right of=v1] (v9) {$9$};

\node[state] (x1)	at	(-2,3.2)	{$x_1$};
\node[state, below left of=x1] (x2)	{$x_2$};
\node[state, below right of=x1] (x3)	{$x_3$};
\draw[->] (x1) to (x2);
\draw[->, bend right=\bend] (x1) to (x3);
\draw[->, bend right=\bend] (x2) to (x3);
\draw[->, bend right=\bend] (x3) to (x1);
\draw[->, bend right=\bend] (x3) to (x2);
\draw[->, loop above, looseness=15] (x1) to (x1);

\def\xbias{0}
\foreach \x in {1,...,6}
\foreach \y in {1,...,6} 
\node[minimum size=5, inner sep=0, circle, draw=black, fill=black]  (\x\y) at (\xbias + \disttwo*\x,\disttwo*\y) {};

\foreach \x in {1,3,5}
\node[]	at	(\xbias + \disttwo*\x, 7*\disttwo)	{\large $x_{\pgfmathparse{int((\x+1)/2)}\pgfmathresult}$};
\foreach \x in {2,4,6}
\node[]	at	(\xbias + \disttwo*\x, 7*\disttwo)	{\large $y_{\pgfmathparse{int((\x)/2)}\pgfmathresult}$};

\foreach \x in {6,4,2}
\node[]	at	(\xbias, \disttwo*\x)	{\large $x_{\pgfmathparse{int(4-(\x)/2)}\pgfmathresult}$};
\foreach \x in {5,3,1}
\node[]	at	(\xbias, \disttwo*\x)	{\large $y_{\pgfmathparse{int(4-(\x+1)/2)}\pgfmathresult}$};

\def\bendtwo{30}
\foreach \y in {2,4,6}
\foreach \x in {2,4}{
\draw[->, bend left=\bendtwo, very thick] ($({\xbias + \disttwo*(\x)},\disttwo*\y)$) to ($({\xbias + \disttwo*(\x+2)},\disttwo*\y)$);
\draw[->, bend left=\bendtwo, very thick] ($({\xbias + \disttwo*(\x+2)},\disttwo*\y)$) to ($({\xbias + \disttwo*(\x)},\disttwo*\y)$);
}
\foreach \x in {1,3,5}
\foreach \y in {1,3}{
\draw[->, bend left=\bendtwo, very thick] ($({\xbias + \disttwo*(\x)},\disttwo*\y)$) to ($({\xbias + \disttwo*(\x)},{\disttwo*(\y+2)})$);
\draw[->, bend left=\bendtwo, very thick] ($({\xbias + \disttwo*(\x)},{\disttwo*(\y+2)})$) to ($({\xbias + \disttwo*(\x)},\disttwo*\y)$);
}

\def\x{6}
\draw[->, draw=\darkblue, very thick, bend left=\bendtwo] ($({\xbias + \disttwo*(7-\x)},{\disttwo*(\x)})$) to ($({\xbias + \disttwo*(7-\x+1)},{\disttwo*(\x)})$);
\foreach \x in {4,2}
\draw[->, draw=\darkblue, very thick] ($({\xbias + \disttwo*(7-\x)},{\disttwo*(\x)})$) to ($({\xbias + \disttwo*(7-\x+1)},{\disttwo*(\x)})$);

\foreach \x in {6,4,2}
\draw[->, draw=\darkblue, very thick] ($({\xbias + \disttwo*(7-\x)},{\disttwo*(\x-1)})$)  to ($({\xbias + \disttwo*(7-\x)},{\disttwo*(\x)})$);

%\foreach \y in {6,4,2}{
\foreach \y in {1,3,5}{
\foreach \x in {1,3,5}{
\ifthenelse{\NOT{\x=\y}}{
\draw[->, draw=\darkgreen, very thick] ($({\xbias + \disttwo*(\x)},{\disttwo*(7-\y)})$)  to ($({\xbias + \disttwo*(\x)},{\disttwo*(7-\y-1)})$);
}
}
}

\draw[->, draw=\darkred, very thick] (46) to (36);
\draw[->, draw=\darkred, very thick] (66) to (56);
\draw[->, draw=\darkred, very thick] (64) to (54);
\draw[->, draw=\darkred, very thick] (22) to (12);
\draw[->, draw=\darkred, very thick] (42) to (32);
\draw[->, draw=\darkred, very thick, bend left=\bendtwo] (26) to (16);

\begin{comment}
\foreach \x in {1,3,5}
\foreach \xx in {1,3,5}
\ifthenelse{\NOT\x=\xx}{
\draw [->, draw=\darkgreen, very thick, bend left=\bendtwo] ($({\xbias + \disttwo*(\x)},{\disttwo*(7-\xx)})$) to ($({\xbias + \disttwo*(\x)},{\disttwo*(7-\x)})$);
}{};
\end{comment}

%\renewcommand{\distone}{5cm*0.38}
%\node[below of =v5] (ra) {\Large$\Rightarrow$};

\end{tikzpicture}
\caption{A graph $G$ (left), and $G'$ that is a 2-self-product of a graph $G''$ of treewidth $1$ (right).
The weighted edges of $G$ correspond to weighted red edges on $G'$.
The distance $\Distance(x_i, x_j)$ in $G$ equals the distance $\Distance(\Prod{x_i, x_i}, \Prod{x_j, x_j})=\Distance(\Prod{\bot, x_i}, \Prod{\bot, x_j})$ in $G'$.
 }\label{fig:product}
\end{figure}

\smallskip
\begin{restatable}{lemma}{pathlemma}\label{lem:paths}
For every $x_i, x_j\in V$, there exists a path $P:x_i\Path x_j$ with $\otimes(P)=z$ in $G$ iff
there exists a path $P':\Prod{x_i,x_i}\Path\Prod{x_j, x_j}$ with $\otimes(P')=z$ in $G'$.
\end{restatable}

\cref{lem:paths} implies that for every $x_i, x_j\in V$, we have 
$\Distance(x_i, x_j) = \Distance(\Prod{x_i,x_i}, \Prod{x_j,x_j})$,
i.e., pair queries in $G$ for nodes $x_i,x_j$ coincide with pair queries 
$(\Prod{x_i,x_i}, \Prod{x_j,x_j})$ in $G'$. 
Observe that in $G'$ we have $\Distance(\Prod{x_i,x_i}, \Prod{x_j,x_j}) = 
\Distance(\Prod{\bot,x_i}, \Prod{\bot,x_j})$, and hence pair queries 
in $G$ also coincide with partial pair queries in $G'$.

\smallskip
\begin{theorem}\label{them:arbitrary_product}
For every graph $G=(V,E)$ and weight function $\Weight:E\to \Sigma$
there exists a graph $G'=(V\times V, E')$ that is a $2$-self-concurrent asynchronous composition 
of a constant-treewidth graph, together with a weight function $\Weight':E'\to\Sigma$, such that for all 
$u,v\in V$, and $\Prod{u,u}, \Prod{v,v}\in V'$ we have 
$\Distance(u,v) = \Distance(\Prod{u,u},\Prod{v,v})=\Distance(\Prod{\bot,u},\Prod{\bot,v})$.
Moreover, the graph $G'$ can be constructed in quadratic time in the size of $G$.
\end{theorem}

This leads to the following corollary.
\smallskip
\begin{corollary}\label{cor:lowerbound}
Let $\Time_S(n)=\Omega(n^2)$ be a lower bound on the time required to answer a single algebraic paths query wrt to a semiring $S$
on arbitrary graphs of $n$ nodes. 
Consider any concurrent graph $G$ which is an asynchronous self-composition of two constant-treewidth graphs of $n$ nodes each.
For any data-structure $\mathsf{DS}$,
let $\Time_{\mathsf{DS}}(G,r)$ be the time required by $\mathsf{DS}$ to preprocess $G$ and answer $r$ pair queries.
We have $\Time_{\mathsf{DS}}(G,1)=\Omega(\Time_S(n))$.
\end{corollary}

\smallskip\noindent{\bf Conditional optimality of \cref{cor:ld}.}
Note that for $r=O(n)$ pair queries, \cref{cor:ld} yields that
the time spent by our data-structure $\ConcurAP$ for preprocessing $G$ and answering $r$ queries 
is $\Time_{\ConcurAP}(G,r)=O(n^3)$.
The long-standing (over five decades) upper bound for answering even one pair query for 
algebraic path properties in arbitrary graphs of $n$ nodes is $O(n^3)$.
\cref{them:arbitrary_product} implies that any improvement upon our results would yield the same improvement
for the long-standing upper bound, which would be a major breakthrough.

\smallskip\noindent{\bf Almost-optimality of \cref{them:concurrent} and \cref{cor:closure}.}
Finally, we highlight some almost-optimality results obtained by variants of $\ConcurAP$ for the 
case of two graphs.
By almost-optimality we mean that the obtained bounds are $O(n^{\eps})$ factor worse that optimal,
for any fixed $\eps>0$ arbitrarily close to $0$.
\begin{compactenum}
\item According to \cref{them:concurrent}, after $O(n^{3+\eps})$ preprocessing time,
single-source queries are handled in $O(n^{2+\eps})$ time, and partial pair queries in $O(1)$ time.
The former (resp. later) query time is almost linear (resp. exactly linear) in the size of the output.
Hence the former queries are handled almost-optimally, and the latter indeed optimally.
Moreover, this is achieved using $O(n^{3+\eps})$ preprocessing time, which is far less than
the $\Omega(n^4)$ time required for the transitive closure computation 
(which computes the distance between all $n^4$ pairs of nodes).

\item According to \cref{cor:closure}, the transitive closure can be computed 
in $O(n^{4+\eps})$ time, for any fixed $\eps>0$, and $O(n^4)$ space. 
Since the size of the output is $\Theta(n^4)$, the transitive closure is computed in 
almost-optimal time and optimal space.
\end{compactenum}

\section{Experimental Results}\label{sec:experiments}

In the current section we report on experimental evaluation of our algorithms, in particular 
of the algorithms of \cref{cor:closure}.
We test their performance for obtaining the transitive closure on various 
concurrent graphs. We focus on the transitive closure for a fair comparison with the
existing algorithmic methods, which compute the transitive closure even for a single query.
%We compare against transitive closure because the existing algorithmic methods 
%always compute the transitive closure.
Since the contributions of this work are algorithmic improvements for algebraic path properties, 
we consider the most fundamental representative of this framework, namely, the shortest path problem.
Our comparison is done against the standard Bellman-Ford algorithm,
which (i)~ has the best worst-case complexity for the problem, and (ii)~allows for practical improvements, such as early termination.

\smallskip\noindent{\bf Basic setup.}
We outline the basic setup used in all experiments. 
We use two different sets of benchmarks, and obtain the controlflow graphs 
of Java programs using Soot~\cite{Soot}, and use LibTW~\cite{LibTW} to obtain 
the tree decompositions of the corresponding graphs.
For every obtained  graph $G'$, we construct a concurrent graph $G$ as a 
$2$-self asynchronous composition of $G'$, and then assign random integer weights in the 
range $[-10^3,10^3]$, without negative cycles.
Although this last restriction does not affect the running time of our algorithms,
it allows for early termination of the Bellman-Ford algorithm
(and thus only benefits the latter).
 The $2$-self composition is a natural situation arising in practice, e.g. in concurrent data-structures
 where two threads of the same method access the data-structure.
 We note that the $2$-self composition is no simpler than
 the composition of any two constant-treewidth graphs,
 (recall that the lower-bound of \cref{sec:arbitrary_product} is established on a $2$-self composition).

\smallskip\noindent{\bf DaCapo benchmarks.}
In our first setup, we extract controlflow graphs of methods from the DaCapo suit~\cite{Blackburn06}.
The average treewidth of the input graphs is around $6$.
This supplies a large pool of $120$ concurrent graphs, for which we use
\cref{cor:closure} to compute the transitive closure. 
This allows us to test the scalability of our algorithms, as well as their practical 
dependence on input parameters.
Recall that our transitive closure time complexity is $O(n^{4+\eps})$, for any fixed $\eps>0$,
which is achieved by choosing a sufficiently large $\lambda\in \Nats$ and a sufficiently small $\delta\in \Reals$
when running the algorithm of \cref{them:tree_decomp}.
%when constructing the $(\alpha,\beta,\gamma)$ tree decompositions of the component graphs.
We compute the transitive closure for various $\lambda$.
In practice, $\delta$ has effects only for very large input graphs.
For this, we fix it to a large value ($\delta=1/3$) which can be proved to have no effect on the obtained running times.
\cref{tab:lambdas} shows for each value of $\lambda$, the percentage of cases for which that value 
is at most $5\%$ slower than the smallest time (among all tested $\lambda$) for each examined case.
We find that $\lambda=7$ works best most of the time.

\begin{table}
\centering
\small
\begin{tabular}{|c||c|c|c|c|c|c|c|c|}
\hline
$\mathbf{\lambda}$ & $2$ & $3$ & $4$ & $5$ & $6$ & $7$ & $8$ \\
\hline
$\mathbf{\%}$ &$6$ & $7$ & $16$ & $22$ & $25$ & $57$ & $17$  \\
\hline
\end{tabular}
\caption{ Percentage of cases for which the transitive closure of the graph $G$ for the given value of $\lambda$
is at most $5\%$ slower than the time required to obtain the transitive closure of $G$ for the best $\lambda$. }
\label{tab:lambdas}
\end{table}

\cref{fig:closure_times} shows the time required to compute the transitive closure on each concurrent graph $G$
by our algorithm (for $\lambda=7$) and the baseline Bellman-Ford algorithm.
We see that our algorithm significantly outperforms the baseline method.
%%Note that although the baseline method scales as a large degree polynomial 
%%(recall that the worst case is $O(n^6)$), 
Note that our algorithm seems to scale much better than its theoretical 
worst-case bound of $O(n^{4+\eps})$ of \cref{cor:closure}.

\begin{figure}
\centering
\includegraphics[scale=0.5]{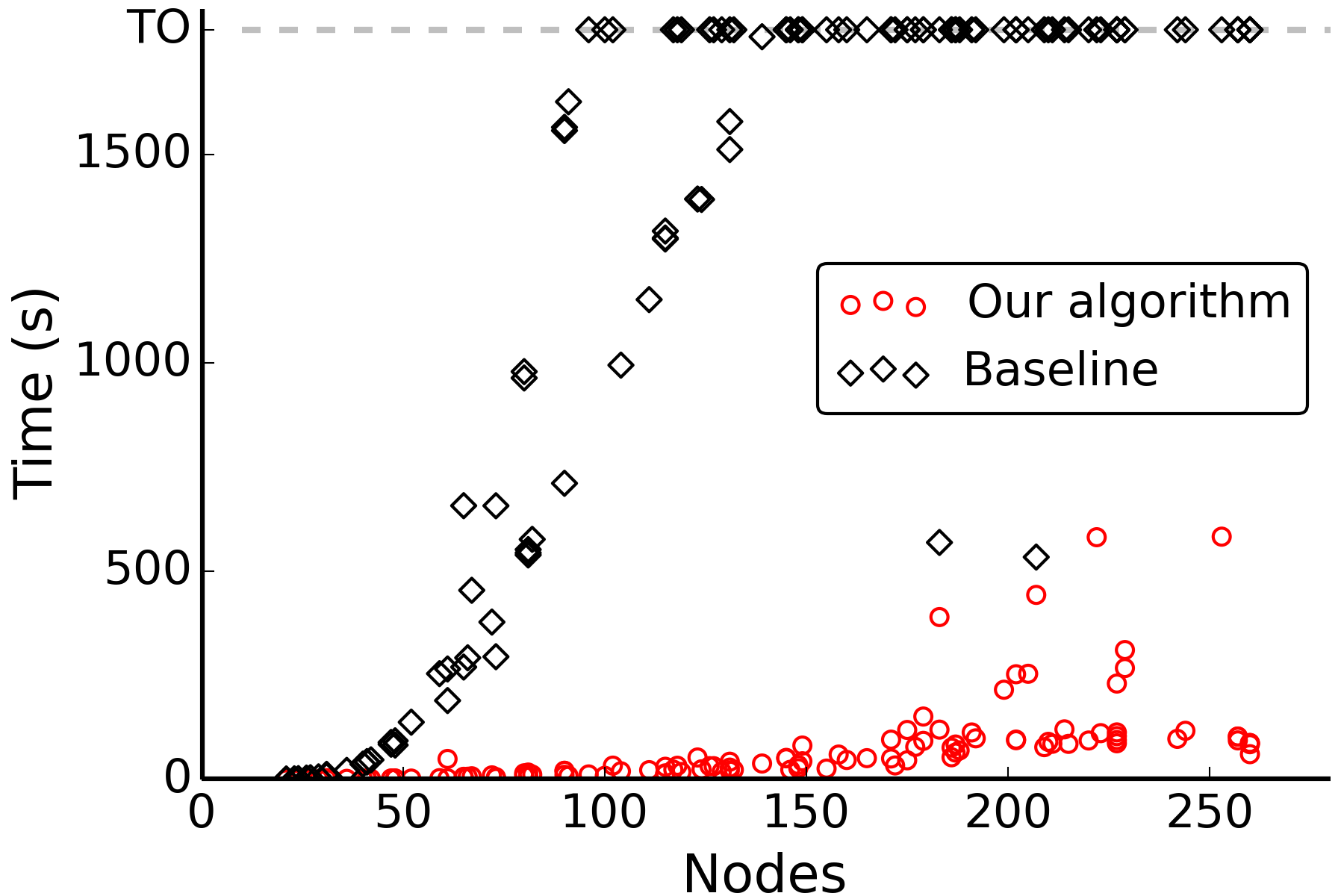}
\caption{Time required to compute the transitive closure on concurrent graphs of various sizes.
Our algorithm is run for $\lambda=7$.
$\mathsf{TO}$ denotes that the computation timed out after 30 minutes.}
\label{fig:closure_times}
\end{figure}

\smallskip\noindent{\bf Concurrency with locks.}
Our second set of experiments is on methods from containers of the $\mathsf{java.util.concurrent}$ library
that use locks as their synchronization mechanism.
%For every examined method, we obtained its controlflow graph $G'$,
%and constructed a concurrent graph $G$ that is the $2$-self asynchronous composition of $G'$.
The average treewidth of the input graphs is around $8$.
In this case, we expand the node set of the concurrent graph $G$ with the lock set $[3]^{\ell}$, 
where $\ell$ is the number of locks used by $G'$.
Intuitively, the $i$-th value of the lock set denotes which of the two components owns the $i$-th lock
(the value is $3$ if the lock is free).
Transitions to nodes that perform lock operations are only allowed wrt the lock semantics.
That is, a transition to a node of $G$ where the value of the $i$-th lock is
\begin{compactenum}
\item \emph{(Lock aquire):}
 $j\in [2]$, is only allowed from nodes where the value of that lock is $3$, and the respective graph $G_j$ is performing a lock operation
 on that edge.
\item \emph{(Lock release):}
$3$, is only allowed from nodes where the value of that lock is $j\in [2]$, and the respective graph $G_j$ is performing an unlock operation
on that edge.
\end{compactenum}

Similarly as before, we compare our transitive closure time with the standard Bellman-Ford algorithm.
\cref{tab:locks} shows a time comparison between our algorithms and the baseline method.
We observe that our transitive closure algorithm is significantly faster, and also scales better.

\begin{table}
\centering
\setlength\tabcolsep{5pt} % default value: 6pt
\begin{tabular}{|l|c|c|c|}
\hline
\multicolumn{1}{|c|}{\textbf{Java method}}& $\mathbf{n}$ & \textbf{$\mathbf{\mathsf{\mathbf{T_o}} (s)}$ } & \textbf{ $\mathbf{\mathsf{\mathbf{T_b}} (s)}$}\\
\hline
\hline
{ \footnotesize ArrayBlockingQueue: poll} & 19 & 19 & 60 \\
\hline
{ \footnotesize ArrayBlockingQueue: peek} & 20 & 20 & 81 \\
\hline
{ \footnotesize LinkedBlockingDeque: advance} & 25 & 29 & 195 \\
\hline
{ \footnotesize PriorityBlockingQueue: removeEQ} & 25 & 32 & 176 \\
\hline
{ \footnotesize ArrayBlockingQueue: init} & 26 & 47 & 249 \\
\hline
{ \footnotesize LinkedBlockingDeque: remove} & 26 & 49 & 290 \\
\hline
{ \footnotesize ArrayBlockingQueue: offer} & 26 & 56 & 304 \\
\hline
{ \footnotesize ArrayBlockingQueue: clear} & 28 & 33 & 389 \\
\hline
{ \footnotesize ArrayBlockingQueue: contains} & 32 & 205 & 881 \\
\hline
{ \footnotesize DelayQueue: remove} & 42 & 267 & 3792 \\
\hline
{ \footnotesize ConcurrentHashMap: scanAndLockForPut} & 46 & 375 & 2176 \\
\hline
{ \footnotesize ArrayBlockingQueue: next} & 46 & 407 & 3915 \\
\hline
{ \footnotesize ConcurrentHashMap: put} & 72 & 1895  & $>8$ h  \\
\hline
\end{tabular}
\caption{Time required for the transitive closure on $2$-self concurrent graphs extracted from methods of the
$\mathsf{java.util.concurrent}$ library.
Each constituent graph has $\mathbf{n}$ nodes.
 \textbf{$\mathbf{\mathsf{\mathbf{T_o}} (s)}$ } and \textbf{ $\mathbf{\mathsf{\mathbf{T_b}} (s)}$}
correspond to our method and the baseline method respectively.
%$\mathsf{TO}$ denotes timeout after $8$ hours.
}
\label{tab:locks}
\end{table}

\section{Conclusions}
We have considered the fundamental algorithmic problem
of computing algebraic path properties in a concurrent 
intraprocedural setting, where component graphs have constant treewidth.
We have presented algorithms that significantly improve the existing theoretical complexity of the problem,
and provide a variety of tradeoffs between preprocessing and query times for on-demand analyses.
Moreover, we have proved that further theoretical improvements over our algorithms must achieve major breakthroughs.
%Given that our theoretical results are almost tight in the intraprocedural concurrent setting, 
An interesting direction of future work is to extend our algorithms to the interprocedural setting.
However, in that case even the basic problem of reachability is undecidable, 
and other techniques and formulations are required to make the analysis tractable,
such as context-bounded formulations and regular approximations of interprocedural paths
~\cite{Qadeer05,Bouajjani05,Lal09}.
The effect of constant-treewidth components in such formulations is an interesting theoretical direction to pursue,
with potential for practical use.

\acks The research was partly supported by Austrian Science Fund (FWF) Grant No P23499- N23, 
FWF NFN Grant No S11407-N23 (RiSE/SHiNE), and ERC Start grant (279307: Graph Games).

\clearpage

{%\scriptsize
\bibliographystyle{abbrvnat}
\bibliography{bibliography}
}
\clearpage

\renewcommand{\baselinestretch}{1}

\appendix

\section*{\centering \Large \textsc{APPENDIX}}
\section{Modeling power}\label{sec:example}

The algebraic paths framework considered in this work has a rich expressive power,
as it can model a wide range of path problems arising in the static analysis of programs.

\smallskip\noindent{\bf Reachability.}
The simplest path problem asks whether there exists a path between two locations of a concurrent system.
The problem can be formulated on the boolean semiring $(\{\True, \False\}, \lor, \land,  \False, \True)$.

\smallskip\noindent{\bf Dataflow problems.}
A wide range of dataflow problems has an algebraic paths formulation, expressed as a ``meet-over-all-paths'' analysis~\cite{Kildall73}.
Perhaps the most well-known case is that of distributive flow functions considered in the IFDS framework~\cite{Reps95, Sagiv96}.
Given a finite domain $D$ and a universe $F$ of distributive dataflow functions $f:2^D\rightarrow 2^D$,
a weight function $\Weight:E\rightarrow F$ associates each edge of the controlflow graph with a flow function. 
The weight of a path is then defined as the composition of the flow functions along its edges,
and the dataflow distance between two nodes $u$, $v$ is the meet $\sqcap$ (union or intersection) of the
weights of all $u\Path v$ paths. 
The problem can be formulated on the meet-composition semiring $(F, \sqcap, \circ, \emptyset,  I)$, where $\circ$ is function composition and $I$ is the identity function. 
We note, however, that the IFDS/IDE framework considers interprocedural paths in sequential programs.
In contrast, the current work focuses on intraprocedural analysis of concurrent programs.
The dataflow analysis of concurrent programs has been a problem of intensive study (e.g.~\cite{Grunwald93, Knoop96, Farzan07, Chugh08, Kahlon09, De11}), where (part of) the underlying analysis is based on an algebraic, ``meet-over-all-paths'' approach.

\smallskip\noindent{\bf Weighted problems.}
The algebraic paths framework subsumes several problems on weighted graphs.
The most well-known such problem is the shortest path problem~\cite{Floyd62, Warshall62, Bellman58, Ford56, Johnson77},
phrased on the tropical semiring $(\Reals\cup \{-\infty, \infty\}, \inf, +, \infty,  0)$.
A number of other fundamental problems on weighted graphs can be reduced to various instances of the shortest-path problem,
e.g. the most probable path, the mean-payoff and the minimum initial credit problem~\cite{Viterbi67, L76, Karp78, Bouyer08, Chatterjee15C}.
Lately, path problems in weighted systems are becoming increasingly important in quantitative verification
and quantitative program analysis~\cite{Bouyer08, W08, CDH10, CHR13, Chatterjee15B}.

\smallskip\noindent{\bf Kleene algebras.}
Finally, a well-known family of closed semirings used in program analysis is that of Kleene algebras~\cite{Fernandes07}.
A common instance is when edges of the controlflow graph are annotated with observations or actions.
In such case the set of observations/actions the system makes in a path between two nodes of the controlflow graph forms a regular language \cite{Dwyer04, Farzan13, Cerny15}. Kleene algebras have also been used as algebraic relaxations of interprocedurally valid paths in sequential and concurrent systems~\cite{Yan11, Bouajjani03}.

\subsection*{Modeling example}
\begin{figure*}
\renewcommand{\thealgocf}{}
\removelatexerror
\centering
\begin{tikzpicture}[thick, >=latex,
pre/.style={<-,shorten >= 1pt, shorten <=1pt, thick},
post/.style={->,shorten >= 1pt, shorten <=1pt,  thick},
und/.style={very thick, draw=gray},
bag/.style={ellipse, minimum height=7mm,minimum width=14mm,draw=gray!80, line width=1pt, inner sep=1},
internal/.style={circle,draw=black!80, inner sep=1, minimum size=4.5mm, very thick},
entry/.style={isosceles triangle, shape border rotate=-90, isosceles triangle stretches, minimum width=8mm, minimum height=3.6mm, draw=black!80, inner sep=0},
exit/.style={isosceles triangle, shape border rotate=90, isosceles triangle stretches, minimum width=8mm, minimum height=3.6mm, draw=black!80, inner sep=0},
call/.style={isosceles triangle, shape border rotate=0, isosceles triangle stretches, minimum width=8mm, minimum height=3.6mm, draw=black!80, inner sep=0},
return/.style={isosceles triangle, shape border rotate=180, isosceles triangle stretches, minimum width=8mm, minimum height=3.6mm, draw=black!80, inner sep=0},
%exit/.style={circle,draw=black!80, inner sep=2, minimum size=4pt},
%call/.style={circle,draw=black!80, inner sep=2, minimum size=4pt},
%return/.style={circle,draw=black!80, inner sep=2, minimum size=4pt},
virt/.style={circle,draw=black!50,fill=black!20, opacity=0}]

\node	[text width = 5.3cm] at (0,0)	{%

\begin{algorithm}[H]
\small
%\TitleOfAlgo{$\Dotvectoralgo$}
\small
\SetInd{0.6em}{0.0em}
\SetAlgoNoLine
\SetAlgorithmName{Method}{method}
\DontPrintSemicolon
%\setstretch{1.05}
\caption{$\Philosophers$}
%\KwIn{$x,y\in \Reals^n$}
%\KwOut{The dot product $x^\top y$}
\While{$\True$}{
\While{$\Fork$ not $\Mine$ or $\Knife$ not $\Mine$}{
\If{$\Fork$ is free}{
$\Lock(\ell)$\\
$\Attain(\Fork)$\\
$\Unlock(\ell)$\\
}
\If{$\Knife$ is free}{
$\Lock(\ell)$\\
$\Attain(\Knife)$\\
$\Unlock(\ell)$\\
}
}
$\Dine(\Fork, \Knife)$\tcp{for some time}
$\Lock(\ell)$\\
$\Release(\Fork)$\\
$\Release(\Knife)$\\
$\Unlock(\ell)$\\
$\Sleep()$\tcp{for some time}
}

\end{algorithm}
};

\newcommand{\globalxdisposition}{4.75}
\newcommand{\globalydisposition}{3.4}
\newcommand{\ynodestep}{-0.8}
\newcommand{\xnodestep}{0.5}

\node	[internal]		(x1)	at	(\globalxdisposition+0*\xnodestep,\globalydisposition+0*\ynodestep)		{$1$};
\node	[internal]		(x20)	at	(\globalxdisposition+3*\xnodestep,\globalydisposition+0*\ynodestep)		{$20$};
\node	[internal]		(x2)	at	(\globalxdisposition+1*\xnodestep,\globalydisposition+1*\ynodestep)		{$2$};
\node	[internal]		(x13)	at	(\globalxdisposition-1*\xnodestep,\globalydisposition+2*\ynodestep)		{$13$};
\node	[internal]		(x14)	at	(\globalxdisposition-1*\xnodestep,\globalydisposition+3*\ynodestep)		{$14$};
\node	[internal]		(x15)	at	(\globalxdisposition-1*\xnodestep,\globalydisposition+4*\ynodestep)		{$15$};
\node	[internal]		(x16)	at	(\globalxdisposition-1*\xnodestep,\globalydisposition+5*\ynodestep)		{$16$};
\node	[internal]		(x17)	at	(\globalxdisposition-1*\xnodestep,\globalydisposition+6*\ynodestep)		{$17$};
\node	[internal]		(x18)	at	(\globalxdisposition-1*\xnodestep,\globalydisposition+7*\ynodestep)		{$18$};
\node	[internal]		(x19)	at	(\globalxdisposition-1*\xnodestep,\globalydisposition+8*\ynodestep)		{$19$};

\node	[internal]		(x3)	at	(\globalxdisposition+2*\xnodestep,\globalydisposition+2*\ynodestep)		{$3$};
\node	[internal]		(x4)	at	(\globalxdisposition+2*\xnodestep,\globalydisposition+3*\ynodestep)		{$4$};
\node	[internal]		(x5)	at	(\globalxdisposition+2*\xnodestep,\globalydisposition+4*\ynodestep)		{$5$};
\node	[internal]		(x6)	at	(\globalxdisposition+2*\xnodestep,\globalydisposition+5*\ynodestep)		{$6$};
\node	[internal]		(x7)	at	(\globalxdisposition+2*\xnodestep,\globalydisposition+6*\ynodestep)		{$7$};

\node	[internal]		(x8)	at	(\globalxdisposition+5*\xnodestep,\globalydisposition+2*\ynodestep)		{$8$};
\node	[internal]		(x9)	at	(\globalxdisposition+5*\xnodestep,\globalydisposition+3*\ynodestep)		{$9$};
\node	[internal]		(x10)	at	(\globalxdisposition+5*\xnodestep,\globalydisposition+4*\ynodestep)		{$10$};
\node	[internal]		(x11)	at	(\globalxdisposition+5*\xnodestep,\globalydisposition+5*\ynodestep)		{$11$};
\node	[internal]		(x12)	at	(\globalxdisposition+5*\xnodestep,\globalydisposition+6*\ynodestep)		{$12$};

\draw [->, thick] (x1) to 	(x2);

\draw [->, thick] (x2) to 	(x13);
\draw [->, thick] (x2) to 	(x3);

\draw [->, thick] (x13) to 	(x14);
\draw [->, thick] (x14) to 	(x15);
\draw [->, thick] (x15) to 	(x16);
\draw [->, thick] (x16) to 	(x17);
\draw [->, thick] (x17) to 	(x18);
\draw [->, thick] (x18) to 	(x19);
\draw [->, thick] (x1) to 	(x20);

\draw [->, thick] (x3) to 	(x4);
\draw [->, thick] (x4) to 	(x5);
\draw [->, thick] (x5) to 	(x6);
\draw [->, thick] (x6) to 	(x7);

\draw [->, thick] (x8) to 	(x9);
\draw [->, thick] (x9) to 	(x10);
\draw [->, thick] (x10) to 	(x11);
\draw [->, thick] (x11) to 	(x12);

\draw [->, thick] (x7) .. controls ([xshift=1cm] x7) and ([xshift=-1cm] x8) .. (x8);
\draw [->, thick] (x12) .. controls ([xshift=-1cm] x12) and ([xshift=1.5cm] x2) .. (x2);
\draw [->, thick] (x19) .. controls ([xshift=-1cm] x19) and ([xshift=-1.5cm] x1) .. (x1);
\draw [->, thick] (x3) .. controls ([xshift=-0.7cm] x3) and ([xshift=-0.7cm] x7) .. (x7);
\draw [->, thick] (x8) .. controls ([xshift=0.7cm] x8) and ([xshift=0.7cm] x12) .. (x12);

\renewcommand{\globalxdisposition}{10.5}
\renewcommand{\globalydisposition}{3.3}
\renewcommand{\ynodestep}{-0.9}
\renewcommand{\xnodestep}{0.6}

\node	[bag]		(b1)	at	(\globalxdisposition + 1*\xnodestep,\globalydisposition+0*\ynodestep)		{$1, 19, 20$};
\node	[bag]		(b2)	at	(\globalxdisposition + 1*\xnodestep,\globalydisposition+1*\ynodestep)		{$1, 2, 19$};
\node	[bag]		(b13)	at	(\globalxdisposition + -1*\xnodestep,\globalydisposition+2*\ynodestep)		{$2, 13, 19$};
\node	[bag]		(b14)	at	(\globalxdisposition + -1*\xnodestep,\globalydisposition+3*\ynodestep)		{$13, 14, 19$};
\node	[bag]		(b15)	at	(\globalxdisposition + -1*\xnodestep,\globalydisposition+4*\ynodestep)		{$14, 15, 19$};
\node	[bag]		(b16)	at	(\globalxdisposition + -1*\xnodestep,\globalydisposition+5*\ynodestep)		{$15, 16, 19$};
\node	[bag]		(b17)	at	(\globalxdisposition + -1*\xnodestep,\globalydisposition+6*\ynodestep)		{$16, 17, 19$};
\node	[bag]		(b18)	at	(\globalxdisposition + -1*\xnodestep,\globalydisposition+7*\ynodestep)		{$17, 18, 19$};

\node	[bag]		(b7)	at	(\globalxdisposition + 3*\xnodestep,\globalydisposition+2*\ynodestep)		{$2, 7, 12$};

\node	[bag]		(b3)	at	(\globalxdisposition + 2*\xnodestep,\globalydisposition+3*\ynodestep)		{$2, 3, 7$};
\node	[bag]		(b4)	at	(\globalxdisposition + 2*\xnodestep,\globalydisposition+4*\ynodestep)		{$3, 4, 7$};
\node	[bag]		(b5)	at	(\globalxdisposition + 2*\xnodestep,\globalydisposition+5*\ynodestep)		{$4, 5, 7$};
\node	[bag]		(b6)	at	(\globalxdisposition + 2*\xnodestep,\globalydisposition+6*\ynodestep)		{$5, 6, 7$};

\node	[bag]		(b8)	at	(\globalxdisposition + 5*\xnodestep,\globalydisposition+3*\ynodestep)		{$7, 8, 12$};
\node	[bag]		(b9)	at	(\globalxdisposition + 5*\xnodestep,\globalydisposition+4*\ynodestep)		{$8, 9, 12$};
\node	[bag]		(b10)	at	(\globalxdisposition + 5*\xnodestep,\globalydisposition+5*\ynodestep)	   {$9, 10, 12$};
\node	[bag]		(b11)	at	(\globalxdisposition + 5*\xnodestep,\globalydisposition+6*\ynodestep)	   {$10, 11, 12$};

\draw [-, very thick] (b1) to 	(b2);
\draw [-, very thick] (b2) to 	(b13);
\draw [-, very thick] (b13) to 	(b14);
\draw [-, very thick] (b14) to 	(b15);
\draw [-, very thick] (b15) to 	(b16);
\draw [-, very thick] (b16) to 	(b17);
\draw [-, very thick] (b17) to 	(b18);

\draw [-, very thick] (b2) to 	(b7);
\draw [-, very thick] (b7) to 	(b8);
\draw [-, very thick] (b8) to 	(b9);
\draw [-, very thick] (b9) to 	(b10);
\draw [-, very thick] (b10) to 	(b11);

\draw [-, very thick] (b7) to 	(b3);
\draw [-, very thick] (b3) to 	(b4);
\draw [-, very thick] (b4) to 	(b5);
\draw [-, very thick] (b5) to 	(b6);

\end{tikzpicture}
\caption{A concurrent program (left), its controlflow graph (middle), and a tree decomposition of the controlflow graph (right).}\label{fig:fig_exampleAp}
\end{figure*}

\cref{fig:fig_exampleAp} illustrates the introduced notions in a small example of the well-known $k$ dining philosophers problem.
For the purpose of the example, $\Lock$ is considered a blocking operation.
Consider the case of $k=2$ threads being executed in parallel. 
The graphs $G_1$ and $G_2$ that correspond to the two threads have nodes of the form $(i,\ell)$,
where $i\in[20]$ is a node of the controlflow graph, and $\ell\in [3]$ denotes the thread that controls the lock
($\ell=3$ denotes that $\ell$ is free, whereas $\ell=i\in [2]$ denotes that it is acquired by thread $i$).
The concurrent graph $G$ is taken to be the asynchronous composition of $G_1$ and $G_2$,
and consists of nodes $\Prod{x,y}$, where $x$ and $y$ is a node of $G_1$ and $G_2$ respectively,
such that $x$ and $y$ agree on the value of $\ell$ (all other nodes can be discarded).
For brevity, we represent nodes of $G$ as triplets $\Prod{x,y,\ell}$ where now $x$ and $y$ are nodes in the controlflow graphs 
$G_1$ and $G_2$ (i.e., without carrying the value of the lock), and $\ell$ is the value of the lock. 
A transition to a node $\Prod{x,y, \ell}$ in which one component $G_i$ performs a $\Lock$ is allowed only from a node where $\ell=3$,
and sets $\ell=i$ in the target node (i.e., $\Prod{x,y,i}$).
Similarly, a transition to a node $\Prod{x,y, \ell}$ in which one component $G_i$ performs an $\Unlock$ is allowed from a node where $\ell=i$, and sets $\ell=3$ in the target node(i.e., $\Prod{x,y,3}$).

Suppose that we are interested in determining (1)~whether the first thread can execute $\Dine(\Fork, \Knife)$
without owning $\Fork$ or $\Knife$, and (2)~whether a deadlock can be reached in which each thread owns one resource.
These questions naturally correspond to partial pair and pair queries respectively, 
as in case (1)~we are interested in a local property of $G_1$, whereas in case (2)~we are interested in a global property of $G$.
We note, however, that case (1)~still requires an analysis on the concurrent graph $G$.
In each case, the analysis requires a set of datafacts $D$, along with dataflow functions $f:2^D\to 2^D$
that mark each edge. These functions are distributive, in the sense that $f(A)=\bigcup_{a\in A}f(a)$.

\smallskip\noindent{\bf Local property as a partial pair query.}
Assume that we are interested in determining whether the first thread can execute $\Dine(\Fork, \Knife)$
without owning $\Fork$ or $\Knife$.
A typical datafact set is $D=\{\Fork, \Knife,\Null \}$, where each datafact denotes that the corresponding resource must be owned by the first thread.
The concurrent graph $G$ is associated with a weight function $\Weight$ of dataflow functions $f:2^D\to 2^D$.
The dataflow function $\Weight(e)$ along an edge $e$ behaves as follows on input datafact $F$
(we only describe the case where $F=\Fork$, as the other case is symmetric).
\begin{compactenum}
\item If $e$ transitions to a node in which the second thread acquires $\Fork$ or the first thread releases $\Fork$, then $\Weight(e)(\Fork)\to \Null$
(i.e., $\Fork$ is removed from the datafacts).
\item Else, if $e$ transitions to a node in which the first thread acquires $\Fork$, then $\Weight(e)(\Null)\to \Fork$
(i.e., $\Fork$ is inserted to the datafacts).
%\item Else, $\Weight(e)$ does not modify the datafact.
\end{compactenum}
Similarly for the $F=\Knife$ datafact. The ``meet-over-all-paths'' operation is set intersection.
Then the question is answered by testing whether $\Distance(\Prod{1,1,3}, \Prod{14,\bot,3})=\{\{\Fork, \Knife\}\}P$, i.e., by performing a \emph{partial pair query}, in which the node of the second thread is unspecified.

\smallskip\noindent{\bf Global property as a pair query.}
Assume that we are interested in determining whether the two threads can cause a deadlock.
Because of symmetry, we look for a deadlock in which the first thread may hold the $\Fork$,
and the second thread may hold the $\Knife$.
A typical datafact set is $D=2^{\{ \Fork, \Knife \}}$. For a datafact $F\in D$ we have 
\begin{compactenum}
\item $\Fork\in F$ if $\Fork$ may be acquired by the \emph{first} thread.
\item $\Knife\in F$ if $\Knife$ may be acquired by the \emph{second} thread.
\end{compactenum}
The concurrent graph $G$ is associated with a weight function $\Weight$ of dataflow functions $f:2^D\to 2^D$.
The dataflow function $\Weight(e)$ along an edge $e$ behaves as follows on input datafact $F$.
\begin{compactenum}
\item If $e$ transitions to a node in which the second thread acquires $\Fork$ or the first thread releases $\Fork$,
then $\Weight(e)(F)\to F\setminus\{\Fork\}$ (i.e., the first thread no longer owns $\Fork$).
\item If $e$ transitions to a node in which the first thread acquires $\Fork$, then $\Weight(e)(F)\to F\cup \{\Fork\}$
(i.e., the first thread now owns $\Fork$).
\item If $e$ transitions to a node in which the first thread acquires $\Knife$ or the second thread releases $\Knife$,
then $\Weight(e)(F)\to F\setminus\{\Knife\}$ (i.e., the second thread no longer owns $\Knife$).
\item If $e$ transitions to a node in which the second thread acquires $\Knife$, then $\Weight(e)(F)\to F\cup \{\Knife\}$
(i.e., the second thread now owns $\Knife$).
\end{compactenum}

The ``meet-over-all paths'' operation is set union.
Then the question is answered by testing whether $\{\Fork, \Knife\}\in \Distance(\Prod{1,1,3}, \Prod{2,2,3})$, i.e., by performing a \emph{pair query},
and finding out whether the two threads can start the $while$ loop with each one holding one resource.
Alternatively, we can answer the question by performing a single-source query from $\Prod{1,1,3}$ and finding out whether 
there exists any node in the concurrent graph $G$ in which every thread owns one resource
(i.e., its distance contains $\{\Fork, \Knife\}$).

\section{Details of \cref{sec:prelim_tree_dec}}\label{sec:tree_decomp}

Given constants $0<\delta\leq 1$ and $\lambda \geq 2$, throughout this section we fix
\[
\alpha=\Bagfactor; \quad \beta=\Balfactor;\quad  \gamma=\Levelfactor
\]

We show how given a graph $G$ of treewidth $t$
and a tree-decomposition $\Tree'(G)$ of $b$ bags and width $t$,
we can construct in $O(b\cdot \log b)$ time and $O(b)$ space a $(\alpha, \beta, \gamma)$ tree-decomposition with $b$ bags.
That is, the resulting tree-decomposition has width at most $\alpha \cdot (t+1)$, 
and for every bag $\Bag$ and descendant $\Bag'$ of $\Bag$ that appears $\gamma$ levels below,
we have that $|T(\Bag')|\leq \beta\cdot |T(\Bag)|$ (i.e., the number of bags in $T(\Bag')$ is at most $\beta$ times
as large as that in $T(\Bag)$).
The result is established in two steps.

\smallskip\noindent{\bf Tree components and operations $\Splitalgo$ and $\Mergealgo$.}
Given a tree-decomposition  $T=(V_T, E_T)$, a  \emph{component} of $T$ is a subset of bags of $T$.
The \emph{neighborhood} $\Nhood(\Comp)$ of $\Comp$
is the set of bags in $V_T\setminus \Comp$ that have a neighbor in $\Comp$, i.e.
\[
\Nhood(\Comp) = \{\Bag\in V_T\setminus \Comp:~(\{\Bag\}\times \Comp)\cap E_T\neq \emptyset\}
\]

Given a component $\Comp$, we define the operation $\Splitalgo$ as 
$\Splitalgo(\Comp)=(\mathcal{\Seplist}, \mathcal{\Complist})$, where
$\mathcal{\Seplist}\subseteq \Comp$ is a list of bags $(\Bag_1, \dots \Bag_{\Nsep})$
and $\mathcal{\Complist}$ is a list of sub-components $(\Comp_1,\dots \Comp_r )$
such that removing each bag $\Bag_i$ from $\Comp$ splits $\Comp$ into the subcomponents
$\mathcal{\Complist}$, and for every $i$ we have $|\Comp_i|\leq \frac{\delta}{2}\cdot |\Comp|$.
Note that since $\Comp$ is a component of a tree, we can find a single separator bag
that splits $\Comp$ into sub-components of size at most $\frac{|\Comp|}{2}$.
Applying this step recursively for $\log(2/\delta)$ levels yields the desired separator set $\mathcal{X}$.
For technical convenience, if this process yields less than $\Nsep$ bags, we repeat some of these bags
until we have $\Nsep$ many.

Consider a list of components $\mathcal{\Complist}=(\Comp_1,\dots \Comp_r)$, and let $z=\sum_i |\Comp_i|$.
Let $j$ be the largest integer such that $\sum_{i=1}^j|\Comp_i|\leq \frac{z}{2}$.
We define the operation $\Mergealgo(\mathcal{\Complist})=(\ov{\Comp}_1, \ov{\Comp}_2)$,
where $\ov{\Comp}_1=\bigcup_{i=1}^j\Comp_i$ and $\ov{\Comp}_2=\bigcup_{i=j+1}^r\Comp_i$.
The following claim is trivially obtained.

\smallskip
\begin{claim}\label{claim:merge}
If $|\Comp_i|<\frac{\delta}{2}\cdot z$ for all $i$, then $|\ov{\Comp}_1|\leq |\ov{\Comp}_2|\leq \frac{1+\delta}{2}\cdot z$.
\end{claim}
\begin{proof}
By construction, $\frac{1-\delta}{2}\cdot z< |\ov{\Comp}_1|\leq \frac{1}{2}\cdot z$,
and since $\ov{\Comp}_1$ and $\ov{\Comp}_2$ partition $\Complist$, we have
$\ov{\Comp}_1+\ov{\Comp}_2=1$. The result follows.
\end{proof}

\smallskip\noindent{\bf Construction of a $(\beta, \gamma)$-balanced rank tree}.
In the following, we consider that $T_G=\Tree'(G)=(V_T, E_T)$ is a tree-decomposition of $G$
and has $|V_T|=b$ bags. 
Given the parameters $\lambda \in \Nats$ with $\lambda \geq 2$ and $0<\delta<1$,
we use the following algorithm $\Rankalgo$ to construct a tree of bags $\Ranktree$.
$\Rankalgo$ operates recursively on inputs $(\Comp, \ell)$ where $\Comp$ is a component of $T_G$
and $\ell\in \{0\}\cup [\Levelfactor-1]$, as follows.
\begin{compactdesc}
\item[1. If $|\Comp|\cdot \frac{\delta}{2}\leq 1$], construct a bag $\RBag=\bigcup_{\Bag\in\Comp}\Bag$, and return $\RBag$. 
%%Otherwise,
\item[2. Else, if $\ell> 0$]\label{item:spliC1}, let $(\Seplist, \Complist)=\Splitalgo(\Comp)$.
Construct a bag $\RBag=\bigcup_{\Bag_i\in \Seplist}\Bag_i$, and let $(\ov{\Comp}_1, \ov{\Comp}_2)=\Mergealgo(\mathcal{\Complist})$.
Call $\Rankalgo$ recursively on input $(\ov{\mathcal{\Comp}}_1, (\ell+1)\mod \Levelfactor)$ and $(\ov{\mathcal{\Comp}}_2, (\ell+1)\mod \Levelfactor)$, and let $\RBag_1$, $\RBag_2$ be the returned bags. Make $\RBag_1$ and $\RBag_2$ the left and right child of $\RBag$.
\item[3. Else,]\label{item:spliC2} if $\ell=0$ , if $|\Nhood(\Comp)| > 1$,
find a bag $\Bag$ whose removal splits $\Comp$ into connected components $\ov{\Comp}_1, \ov{\Comp}_2$
with $|\Nhood(\ov{\Comp}_i)\cap \Nhood(\Comp)|\leq \frac{|\Nhood(\Comp)|}{2}$.
Call $\Rankalgo$ recursively on input $(\ov{\Comp}_1, (\ell+1)\mod \Levelfactor)$ and $(\ov{\Comp}_2, (\ell+1)\mod \Levelfactor)$, and let $\RBag_1$, $\RBag_2$ be the returned bags. Make $\RBag_1$ and $\RBag_2$ the left and right child of $\RBag$.
Finally, if $|\Nhood(\Comp)| \leq 1$, call $\Rankalgo$ recursively on input $(\Comp, (\ell-1)\mod \Levelfactor)$.
\end{compactdesc}
%A description of $\Rankalgo$ in pseudocode can be found in \cref{sec:algorithms}.

In the following we use the symbols $\Bag$ and $\RBag$ to refer to bags of $T_G$ and $\Ranktree$ respectively.
Given a bag $\RBag$, we denote by $\Comp(\RBag)$ the input component of $\Rankalgo$ when $\RBag$ was constructed,
and define the \emph{neighborhood} of $\RBag$ as $\Nhood(\RBag)=\Nhood(\Comp(\RBag))$.
Additionally, we denote by $\Baghood(\RBag)$ the set of separator bags $\Bag_1,\dots \Bag_r$ of $\Comp$ that were used to construct $\RBag$.
It is straightforward that $\Baghood(\RBag_1)\cap \Baghood(\RBag_2) = \emptyset$ for every distinct
$\RBag_1$ and $\RBag_2$.

\smallskip
\begin{claim}\label{prop:nhood_parent}
Let $\RBag$ and $\RBag'$ be respectively a bag and its parent in $\Ranktree$. 
Then $\Nhood(\RBag)\subseteq \Nhood(\RBag')\cup \Baghood(\RBag')$,
and thus $|\Nhood(\RBag)|\leq |\Nhood(\RBag')|+\Nsep$.
\end{claim}
\begin{proof}
Every bag in $\Nhood(\Comp(\RBag))$ is either a bag in $\Nhood(\Comp(\RBag'))$,
or a separator bag of $\Comp(\RBag')$, and thus a bag of $\Baghood(\Bag')$.
\end{proof}

\begin{figure}
\newcommand\irregularcircle[2]{% radius, irregularity
  \pgfextra {\pgfmathsetmacro\len{(#1)+rand*(#2)}}
  +(0:\len pt)
  \foreach \a in {10,25,...,350}{
    \pgfextra {\pgfmathsetmacro\len{(#1)+rand*(#2)}}
    -- +(\a:\len pt)
  } -- cycle
}

\newcommand\irregularellipse[3]{% a,b, irregularity
  \pgfextra {\pgfmathsetmacro\b{(#1)+rand*(#3)}}
  \pgfextra {\pgfmathsetmacro\c{(#2)+rand*(#3)}}
  \pgfextra {\pgfmathsetmacro\len{\b*\c/ ( (\c * cos(0))^2 + (\b * sin(0 ))^2)^(1/2)}}
  +(0:\len pt)
  \foreach \a in {10,20,...,350}{
    \pgfextra {\pgfmathsetmacro\b{(#1)+rand*(#3)}}
    \pgfextra {\pgfmathsetmacro\c{(#2)+rand*(#3)}}
    \pgfextra {\pgfmathsetmacro\len{\b*\c/ ( (\c * cos(\a))^2 + (\b * sin(\a ))^2)^(1/2)}}
    -- +(\a:\len  pt)
  } -- cycle
}

\centering
\begin{tikzpicture}
\pgfmathsetseed{1234}
\coordinate (c) at (0,0);
%\draw[very thick,rounded corners=1mm, fill=red!30] (c) \irregularcircle{1.6cm}{2mm};
%\draw[very thick,rounded corners=1mm, fill=gray!30] (c) \irregularcircle{1.2cm}{1.5mm};
\draw[very thick,rounded corners=1mm, fill=red!30] (c) \irregularellipse{38mm}{13mm}{1.2mm};
\draw[very thick,rounded corners=1mm, fill=gray!30] (c) \irregularellipse{23mm}{9mm}{1.2mm};
\node[] at (-3,0) {$\Nhood(\Comp)$};
\node[circle, draw=black, thick, fill=white] at (0,0) (bag)	{$\Bag$};
\node[] at (1,0.1) {$\ov{\Comp}_2$};
%\node[] at (1.3,-0.2) {$\Comp_3$};
\node[] at (-1.2,-0.1) {$\ov{\Comp}_1$};
%\draw[thick, decoration={random steps,segment length=2mm,amplitude=0.5mm}, decorate,] (bag) -- (2,1.2);
\draw[thick, decoration={random steps,segment length=2mm,amplitude=0.5mm}, decorate,] (bag) -- (1.1,-1.2);
\draw[thick, decoration={random steps,segment length=2mm,amplitude=0.5mm}, decorate,] (bag) -- (-1.8,1.15);
\end{tikzpicture}
\caption{Illustration of one recursive step of $\Rankalgo$ on a component $\Comp$ (gray).
$\Comp$ is split into two sub-components $\ov{\Comp}_1$ and $\ov{\Comp}_2$
by removing a list of bags $\Seplist=(\Bag_i)_i$.
Once every $\lambda$ recursive calls, $\Seplist$ contains one bag, 
such that the neighborhood $\Nhood(\ov{\Comp}_i)$ of each $\ov{\Comp}_i$ is at most half the size of $\Nhood(\Comp)$
(i.e., the red area is split in half).
In the remaining $\lambda-1$ recursive calls, $\Seplist$ contains $\Nsep$ bags,
such that the size of each $\ov{\Comp}_i,$
is at most $\frac{1+\delta}{2}$ fraction the size of $\Comp$.
(i.e., the gray area is split in almost half).
}\label{fig:nhoodAp}
\end{figure}

Note that every bag $\Bag$ of $T_G$ belongs in $\Baghood(\RBag)$ of some bag $\RBag$ of $\Ranktree$,
and thus the bags of $\Ranktree$ already cover all nodes and
edges of $G$ (i.e., properties C1 and C2 of a tree decomposition).
In the following we show how $\Ranktree$ can be modified to also satisfy condition C3,
i.e., that every node $u$ appears in a contiguous subtree of $\Ranktree$.
Given a bag $\RBag$, we denote by $\Nhoodnodes(\RBag)=\RBag\cup \bigcup_{\Bag\in \Nhood(\RBag)}\Bag$,
i.e., $\Nhoodnodes(\RBag)$ is the set of nodes of $G$ that appear in $\RBag$ and its neighborhood.
In the sequel, to distinguish between paths in different trees,
given a tree of bags $T$ (e.g. $T_G$ or $\Ranktree$) and bags $B_1$, $B_2$ of $T$,
we write $B_1\Path_T B_2$ to denote the unique simple path from $B_1$ to $B_2$ in $T$.
%In \cref{lem:gap} we will show the crucial property that for every pair of bags $\RBag_1$ and $\RBag_2$ such that 
%$\RBag_1$ is ancestor of $\RBag_2$ in $\Ranktree$, every node $u\in (\RBag_1\cap\RBag_2)$ is in the neighborhood $\Nhoodnodes(\RBag)$
%of every bag $\RBag$ of the path $\RBag_1\Path_{\Ranktree} \RBag_2$.

We say that a pair of bags $(\RBag_1$, $\RBag_2)$ \emph{form a gap} of some node $u$ in a tree of bags $T$ (e.g., $\Ranktree$)
if $u\in \RBag_1\cap \RBag_2$
and for the unique simple path $P:\RBag_1\Path_T \Bag_2$ we have that $|P|\geq 2$ (i.e., there is at least one intermediate bag in $P$)
and for all intermediate bags $\RBag$ in $P$ we have $u\not\in \RBag$.
The following crucial lemma shows that if $\RBag_1$ and $\RBag_2$ form a gap of $u$ in $\Rankhtree$,
then for every intermediate bag $\RBag$ in the path $P:\RBag_1\Path_{\Ranktree} \RBag_2$,
$u$ must appear in some bag of $\Nhood(\RBag)$.

\smallskip
\begin{lemma}\label{lem:gap}
For every node $u$, and pair of bags $(\RBag_1$, $\RBag_2)$ that form a gap of $u$ in $\Ranktree$, such that $\RBag_1$ is an ancestor of $\RBag_2$,
for every intermediate bag $\RBag$ in $P:\RBag_1\Path_{\Ranktree} \RBag_2$ in $\Ranktree$, 
we have that $u\in \Nhoodnodes(\RBag)$.
\end{lemma}
\begin{proof}
Fix any such a bag $\RBag$, and since $\RBag_1$ and $\RBag_2$ form a gap of $u$,
there exist bags $\Bag_1\in \Baghood(\RBag_1)$ and $\Bag_2\in \Baghood(\RBag_2)$
with $u\in \Bag_1\cap \Bag_2$.
Let $\Bag^r$ be the rightmost bag  of the path $P_1:\Bag_1\Path_{T_G}\Bag_2$ that has been chosen as a separator
when $\RBag$ was constructed. Note that $\Bag_1$ has been chosen as such a separator, therefore $\Bag^r$ is well defined.
We argue that $\Bag^r\in \Nhood(\RBag)$, which implies that $u\in \Nhoodnodes(\RBag)$.
This is done in two steps.

\begin{compactenum}
\item Since $\RBag_2$ is a descendant of $\RBag$, we have that $\Bag_2\in \Comp(\RBag)$,
i.e., $\Bag_2$ is a bag of the component when $\RBag$ was constructed.
\item By the choice of $\Bag^r$, for every intermediate bag $\Bag^i$ in the path $\Bag^r\Path_{T_G}\Bag_2$
we have $\Bag^i\in \Comp(\RBag)$. Hence $\Bag^r$ is incident to the component where $\Bag_2$ belonged
at the time $\RBag$ was constructed.
\end{compactenum}
These two points imply that $\Bag^r\in \Nhood(\RBag)$, as desired.
\cref{fig:gap} provides an illustration of the argument.
\end{proof}

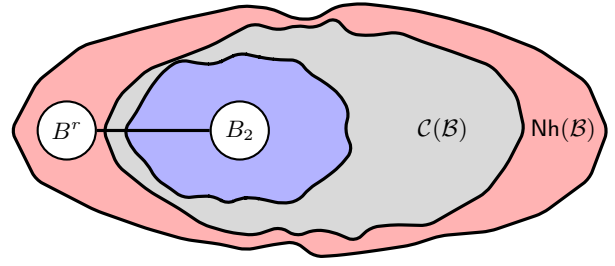
\begin{figure}
\newcommand\irregularcircle[2]{% radius, irregularity
  \pgfextra {\pgfmathsetmacro\len{(#1)+rand*(#2)}}
  +(0:\len pt)
  \foreach \a in {10,25,...,350}{
    \pgfextra {\pgfmathsetmacro\len{(#1)+rand*(#2)}}
    -- +(\a:\len pt)
  } -- cycle
}

\newcommand\irregularellipse[3]{% a,b, irregularity
  \pgfextra {\pgfmathsetmacro\b{(#1)+rand*(#3)}}
  \pgfextra {\pgfmathsetmacro\c{(#2)+rand*(#3)}}
  \pgfextra {\pgfmathsetmacro\len{\b*\c/ ( (\c * cos(0))^2 + (\b * sin(0 ))^2)^(1/2)}}
  +(0:\len pt)
  \foreach \a in {10,20,...,350}{
    \pgfextra {\pgfmathsetmacro\b{(#1)+rand*(#3)}}
    \pgfextra {\pgfmathsetmacro\c{(#2)+rand*(#3)}}
    \pgfextra {\pgfmathsetmacro\len{\b*\c/ ( (\c * cos(\a))^2 + (\b * sin(\a ))^2)^(1/2)}}
    -- +(\a:\len  pt)
  } -- cycle
}

\centering
\begin{tikzpicture}
\pgfmathsetseed{3456}
\coordinate (c) at (0,0);
\coordinate (c2) at (0,0);
\coordinate (c3) at (-1,0);
%\draw[very thick,rounded corners=1mm, fill=red!30] (c) \irregularcircle{1.6cm}{2mm};
%\draw[very thick,rounded corners=1mm, fill=gray!30] (c) \irregularcircle{1.2cm}{1.5mm};
\draw[very thick,rounded corners=1mm, fill=red!30] (c2) \irregularellipse{40mm}{16mm}{1.2mm};
\draw[very thick,rounded corners=1mm, fill=gray!30] (c) \irregularellipse{28mm}{14mm}{1.2mm};
\draw[very thick,rounded corners=1mm, fill=blue!30] (c3) \irregularellipse{15mm}{10mm}{1.2mm};
%\node[] at (-3,0) {$\Nhood(\Comp)$};
\node[circle, draw=black, thick, fill=white, minimum size=22] at (-3.3,0) (bagr)	{$\Bag^r$};
\node[circle, draw=black, thick, fill=white, minimum size=22] at (-1,0) (bag)	{$\Bag_2$};
\node[] at (1.7,0) {$\Comp(\RBag)$};
\node[] at (3.3,0) {$\Nhood(\RBag)$};
%\node[] at (-1.2,-0.1) {$\ov{\Comp}_1$};
\draw[very thick] (bagr) -- (bag);
%\draw[thick, decoration={random steps,segment length=2mm,amplitude=0.5mm}, decorate,] (bag) -- (1.1,-1.4);
%\draw[thick, decoration={random steps,segment length=2mm,amplitude=0.5mm}, decorate,] (bag) -- (-1.8,1.35);

%\def\xbias{-6}
%\node[circle, draw=black, thick, fill=white, minimum size=22] at (\xbias,1.5) (bag1)	{$\Bag_1$};
%\node[circle, draw=black, thick, fill=white, minimum size=22] at (\xbias,-0.5) (bagr)	{$\Bag^r$};
%\node[circle, draw=black, thick, fill=white, minimum size=22] at (\xbias,-2) (bag2)	{$\Bag_2$};
%\draw[very thick] (bag1) -- (bagr) -- (bag2);

%\def\xbias{-12}
%\node[circle, draw=black, thick, fill=white, minimum size=22] at (\xbias,1.5) (bag1)	{$\RBag_1$};
%\node[circle, draw=black, thick, fill=white, minimum size=22] at (\xbias,0) (bagr)	{$\RBag$};
%\node[circle, draw=black, thick, fill=white, minimum size=22] at (\xbias,-1.5) (bag2)	{$\RBag_2$};
%\draw[very thick] (bag1) -- (bagr) -- (bag2);

\end{tikzpicture}
\caption{
Illustration of \cref{lem:gap}. Since $\Bag_2$ belongs to $\Comp(\RBag)$ and the blue sub-component has not been split yet,
the bag $\Bag^r$ is in the neighborhood of the blue sub-component, and thus in the neighborhood of $\Comp(\RBag)$.
}\label{fig:gap}
\end{figure}

\smallskip\noindent{\bf Turning the rank tree to a tree decomposition.}
\cref{lem:gap} suggests a way to turn the rank tree $\Ranktree$ to a tree-decomposition.
Let $\Rankhtree=\Rep(\Ranktree)$ be the tree obtained by replacing each bag $\Bag$ of $\Ranktree$ 
with $\Nhoodnodes(\Bag)$. 
For a bag $\RBag$ in $\Ranktree$ let $\RHBag$ be the corresponding bag in $\Rankhtree$  and vice versa.

\smallskip
\begin{claim}\label{claim:gap_ancestor}
If there is a pair of bags $\RHBag_1$, $\RHBag_2$ that form a gap of some node $u$ in $\Rankhtree$,
then there is a pair of bags $\RHBag'_1$, $\RHBag'_2$ that also form a gap of $u$,
and $\RHBag'_1$ is ancestor of $\RHBag'_2$.
\end{claim}
\begin{proof}
Assume that neither  of $\RHBag_1$, $\RHBag_2$ is ancestor of the other.
\begin{compactenum}
\item If for some $i\in\{1,2\}$ there is no bag $\Bag_i\in \Baghood(\RBag_i)$ such that $u\in \Bag_i$,
then $u\in \Nhoodnodes(\RBag_i)\setminus \RBag_i$ and hence there is an ancestor $\RBag'_i$ of $\RBag_i$ such that $u\in \Nhoodnodes(\RBag'_i)$.
Thus $\RHBag'_i$ and $\RHBag_i$ form a gap of $u$ in $\Rankhtree$.
\item Else, there exists a $\Bag_1\in \Baghood(\RBag_1)$ and $\Bag_2\in \Baghood(\RBag_2)$ such that $u\in \Bag_1\cap \Bag_2$.
Let $\Bag$ be first bag in the path $\Bag_1\Path_{T_G}\Bag_2$ that was chosen as a separator we have $\Bag\in \Baghood(\RBag)$
for some ancestor $\RBag$ of $\RBag_1$ and $\RBag_2$, therefore $u\in \Nhoodnodes(\RBag)$.
\end{compactenum}
It follows that there exists an ancestor $\RHBag'_i$ of some $\RHBag_i$ so that the two form a gap of $u$ in $\Rankhtree$.
\end{proof}

The following lemma states that $\Rankhtree$ is a tree decomposition of $G$. 

\smallskip
\begin{lemma}\label{lem:balance_correctness}
$\Rankhtree=\Rep(\Ranktree)$ is a tree-decomposition of $G$.
\end{lemma}
\begin{proof}
It is straightforward to see that the bags of $\Rankhtree$ cover all nodes and edges of $G$
(properties C1 and C2 of the definition of tree-decomposition),
because for each bag $\RBag$, we have that $\RBag\subseteq \widehat{\RBag}$.
It remains to show that every node $u$ appears in a contiguous subtree of $\Rankhtree$
(i.e., that property C3 is satisfied).

Assume towards contradiction otherwise, and by \cref{claim:gap_ancestor} it follows that there exist bags
$\RHBag_1$ and $\RHBag_2$ in $\Rankhtree$ that form a gap of some node $u$ such that $\RHBag_1$ is an ancestor of $\RHBag_2$.
Let $\widehat{P}:\RHBag_1\Path_{\Rankhtree} \RHBag_2$ be the path between them,
and $P:\RBag_1\Path_{\Ranktree} \RBag_2$ the corresponding path in $\Ranktree$.
By \cref{lem:gap} we have $u\not \in \RBag_1\cap \RBag_2$, otherwise for every intermediate bag $\RBag\in \widehat{P}$
we would have $u\in \Nhoodnodes(\RBag)$ and thus $u\in \RHBag$.
Additionally, we have $u\in \RBag_2$, otherwise by \cref{prop:nhood_parent}, we would have $u\in \Nhoodnodes(\RBag'_2)$,
where $\RBag'_2$ is the parent of $\RBag_2$, and thus $u\in \RHBag'_2$, contradicting the assumption that $\RHBag_1$ and $\RHBag_2$ form a gap of $u$. Hence $u\not \in \RBag_1$.
A similar argument as that of \cref{claim:gap_ancestor} shows that there exists an ancestor $\RBag'_1$ of $\RBag_1$
such that $u\in \RBag'_1$, and WLOG, take $\RBag'_1$ to be the lowest ancestor of $\RBag_1$ with this property.
Then $\RBag'_1$ is also an ancestor of $\RBag_2$, and $\RBag'_1$ and $\RBag_2$ form a gap of $u$ in $\Ranktree$.
Since $\RBag_1$ is an intermediate bag in $\RBag'_1\Path_{\Ranktree} \RBag_2$, by \cref{lem:gap} we have that $u\in \Nhoodnodes(\RBag)$,
thus $u\in \RHBag$. We have thus arrived at a contradiction, and the desired result follows.
\end{proof}

\smallskip\noindent{\bf Properties of the tree-decomposition $\Rankhtree$.}
\cref{lem:balance_correctness} states that $\Rankhtree$ obtained by replacing each bag of $\Ranktree$
with $\Nhoodnodes(\Bag)$ is a tree-decomposition of $G$. 
The remaining of the section focuses on showing that $\Rankhtree$ is a $(\alpha, \beta, \gamma)$
tree-decomposition of $G$, and that it can be constructed in $O(b\cdot \log  b)$ time and $O(b)$ space.

\smallskip
\begin{lemma}\label{lem:rank_assertions}
The following assertions hold:
\begin{compactenum}
\item \label{item:rank} Every bag $\RHBag$ of $\Rankhtree$ is $(\beta, \gamma)$-balanced.
\item \label{item:nhood} For every bag $\RHBag$ of $\Rankhtree$, we have $|\RHBag|\leq \alpha\cdot (t+1)$.
\end{compactenum}
\end{lemma}
\begin{proof}
We prove each item separately.
\begin{compactenum}
\item For every bag $\RBag$ constructed by $\Rankalgo$,in at least $\gamma-1$ out of every $\gamma$ levels,
Item~2 of the algorithm applies, and by \cref{claim:merge}, the recursion proceeds on components $\ov{\Comp}_1$ and $\ov{\Comp}_2$
that are at most $\frac{1+\delta}{2}$ factor as large as the input component $\Comp$ in that recursion step.
Thus $\RBag$ is $(\beta, \gamma)$-balanced in $\Ranktree$, and hence $\RHBag$ is 
$(\beta, \gamma)$-balanced in $\Rankhtree$.
\item 
It suffices to show that for every bag $\RBag$, we have $|\Nhood(\RBag)|\leq \alpha-1 = 2\cdot (\Nsep)\cdot \lambda-1$.
Assume towards contradiction otherwise.
Let $\RBag$ be the first bag that $\Rankalgo$ assigned a rank such that $|\Nhood(\RBag)|\geq 2\cdot (\Nsep)\cdot \lambda$.
Let $\RBag'$ be the lowest ancestor of $\RBag$ in $\Ranktree$ that was constructed by $\Rankalgo$ on some input $(\Comp, \ell)$
with $\ell=1$, and let $\RBag''$ be the parent of $\RBag'$ in $\Ranktree$
(note that $\RBag'$ can be $\RBag$ itself).
By Item~3 of $\Rankalgo$, it follows that $|\Nhood(\RBag')|\leq\lfloor\frac{|\Nhood(\RBag'')|}{2}\rfloor+1$.
Note that $\RBag'$ is at most $\lambda-1$ levels above $\RBag$ (as we allow $\RBag'$ to be $\RBag$).
By \cref{prop:nhood_parent}, the neighborhood of a bag can increase by at most $(\Nsep)$ from the neighborhood of its parents,
hence $\Nhood(\RBag')\geq (\Nsep)\cdot (\lambda+1)$.
The last two inequalities lead to $|\Nhood(\RBag'')|\geq 2\cdot (\Nsep) \cdot  \lambda$, which contradicts our choice of $\Bag$.
\end{compactenum}
The desired result follows.
\end{proof}

\smallskip\noindent{\bf A minimal example.}
\cref{fig:TreeG2Rank2BalTreeG} illustrates an example of $\Rankhtree$ constructed out of a tree-decomposition $\Tree'(G)$.
First, $\Tree'(G)$ is turned into a  binary and balanced tree $\Ranktree$ and then into a binary and balanced tree $\Rankhtree$. 
If the numbers are pointers to bags, such that $\Tree'(G)$ is a tree-decomposition for $G$, then $\Rankhtree$ is a binary and balanced tree-decomposition of $G$.
The values of $\lambda$ and $\delta$ are immaterial for this example, as $\Rankhtree$ becomes 
perfectly balanced (i.e., $(1/2,1)$-balanced).
\begin{figure*}
\centering
\newcommand{\distone}{3.3cm*0.3}
\begin{tikzpicture}[thick,scale=0.7, node distance=\distone]
\tikzstyle{every state}=[very thick, fill=white,draw=black,text=black,font=\small , inner sep=-0.05cm, minimum size=8mm]

\node[state] (v1) {1};
\xdef\lv{1}
\foreach \v in {2,...,4} {
\node[state,above right of=v\lv] (v\v) {\v};
\draw (v\lv) to (v\v);
\xdef\lv{\v}
}

\foreach \v in {5,...,7} {
\node[state,below right of=v\lv] (v\v) {\v};
\draw (v\lv) to (v\v);
\xdef\lv{\v}
}

\node[above right of=v7] (arrow1) {\Large $\Rightarrow$};

\node[state,below right of=arrow1] (v1') {1};
\node[state,above right of=v1'] (v2') {2};
\node[state,above right of=v2',node distance=2cm] (v4') {4};
\node[state,below right of=v4',node distance=2cm] (v6') {6};
\node[state,below right of=v6'] (v7') {7};
\node[state,below left of=v6'] (v5') {5};
\node[state,below right of=v2'] (v3') {3};

\foreach \x/\y in {1/2,2/3,2/4,4/6,6/5,6/7} {
\draw (v\x ') to (v\y ');
}

\node[above right of=v7'] (arrow2) {\Large $\Rightarrow$};

\node[state,below right of=arrow2] (v1h) {1,2};
\node[state,above right of=v1h] (v2h) {2,4};
\node[state,above right of=v2h,node distance=2cm] (v4h) {4};
\node[state,below right of=v4h,node distance=2cm] (v6h) {4,6};
\node[state,below right of=v6h] (v7h) {7,6};
\node[state,below left of=v6h] (v5h) {4,5,6};
\node[state,below right of=v2h] (v3h) {2,3,4};

\foreach \x/\y in {1/2,2/3,2/4,4/6,6/5,6/7} {
\draw (v\x h) to (v\y h);
}
\renewcommand{\distone}{5.5cm*0.3}
\node[left of=v4]{\large $\Tree(G)$};
\node[left of=v4']{\large $\Ranktree$};
\node[left of=v4h]{\large $\widehat{\Ranktree}$};

\end{tikzpicture}
\caption{Given the tree-decomposition $\Tree(G)$ on the left, the graph in the middle is the corresponding $\Ranktree$ and the one on the right is the corresponding tree-decomposition $\widehat{\Ranktree}=\Rep(\Ranktree)$
after replacing each bag $\Bag$ with $\Nhoodnodes(\Bag)$.\label{fig:TreeG2Rank2BalTreeG}}
\end{figure*}  

\smallskip
\balancedtreedec*
\begin{proof}
By~\cite{Bodlaender96} an initial tree-decomposition $\Tree'(G)$ of $G$ with width $t$ and $b=O(n)$ bags can be constructed in $O(n)$ time.
\cref{lem:balance_correctness} and \cref{lem:rank_assertions} prove that the constructed $\Rankhtree$ is a
$(\alpha, \beta, \gamma)$ tree-decomposition of $G$.
The time and space complexity come from the construction of $\Ranktree$ by the recursion of $\Rankalgo$.
It can be easily seen that every level of the recursion processes disjoint components $\Comp_i$ of $\Tree'(G)$ in $O(|\Comp_i|)$
time, thus one level of the recursion requires $O(b)$ time in total. There are $O(\log b)$ such levels,
since every $\Levelfactor$ levels, the size of each component has been reduced to at most a factor $\Balfactor$.
Hence the time complexity is $O(b\cdot \log b)=)(n\cdot \log n)$.
The space complexity is that of processing a single level of the recursion, hence $O(b)=O(n)$.
\end{proof}

\section{Details of \cref{sec:product_tree_dec}}\label{sec:product_tree_proofs}

\smallskip
\begin{restatable}{lemma}{concurtree}\label{lem:product_tree_dec}
$\ProductTree(G)$ is a tree decomposition of $G$.
\end{restatable}
\begin{proof}
We show that $T$ satisfies the three conditions of a tree decomposition.
\begin{compactenum}
\item[C1] For each node $u=\Prod{u_i}_{1\leq i\leq k}$, let $j=\arg\min_i\Level(u_i)$. 
Then $u\in \Bag$, where $\Bag$ is the bag constructed by step~1 of $\ProductTree$
when it operates on the tree-decompositions $(T_i(\Bag_{i}))_{1\leq i\leq k}$ with $\Bag_{j}=\Bag_{u_j}$,
the root bag of $u_j$ in $T_j$.
\item[C2] Similarly, for each edge $(u,v)\in E$ with $u=\Prod{u_i}_{1\leq i\leq k}$ and $v=\Prod{v_i}_{1\leq i\leq k}$,
let $j=\arg\min_i(\max(\Level(u_i), \Level(v_i)))$. Then $(u,v)\in \Bag$, where $\Bag$ is a bag similar to C1. 
\item[C3] For any node $u=\Prod{u_i}_{1\leq i\leq k}$ and path $P:\Bag\Path \Bag'$ with $u\in \Bag\cap \Bag'$,
let $\Bag''$ be any bag of $P$.
Since at least one of $\Bag$, $\Bag'$ is a descendant of $\Bag''$, we have 
$V_T(\Bag)\subseteq V_T(\Bag'')$ or $V_T(\Bag')\subseteq V_T(\Bag'')$,
and because $u\in \Bag\cap \Bag'$, if $\Bag''$ was constructed on input $(T_i(\Bag''_i))_{1\leq i\leq k}$,
we have $u_i\in V_{T_i}(\Bag''_i)$.
Let $(T_i(\Bag_i))_{1\leq i\leq k}$ and $(T_i(\Bag'_i))_{1\leq i\leq k}$ be the inputs to the algorithm
when $\Bag$ and $\Bag'$ were constructed, and it follows that for some $1\leq j\leq k$
we have $u_j\in \Bag_j\cap \Bag'_j$.
Then $\Bag''_j$ is an intermediate bag in the path $P_j:\Bag_j\Path \Bag'_j$ in $T_j$,
thus $u_j\in \Bag''_j$ and hence $u\in \Bag''$.
\end{compactenum}
The  desired result follows.
\end{proof}

%\smallskip
%\rembagsize*

\smallskip
\begin{restatable}{lemma}{recconc}\label{lemm:rec-conc}
Consider the following recurrence.
\begin{equation}\label{eq:recP}
\Time(n_1,\dots, n_k) \leq \sum_{1\leq i\leq k} \prod_{j\neq i}n_j  + \sum_{\mathclap{(r_i)_i\in \Children^k}} \Time( n_{1,{r_1}},\dots,  n_{k,{r_k}})
\end{equation}
such that  for every $i$ we have that 
$
\sum_{(r_i)_i\in \Children^k} n_{i,r_i} \leq n_i. 
$
Then we have 
\begin{equation}\label{eq:solution1}
\Time(n_1,\dots, n_k)\leq   \prod_{1\leq i\leq k}n_i -  \sum_{1\leq i\leq k}\prod_{j\neq i}n_j.
\end{equation}
%%\end{lemma}
\end{restatable}
\begin{proof}
Indeed, substituting \cref{eq:solution1} to the recurrence \cref{eq:recP} we have
\begin{align*}
\Time(n_1,\dots, n_k) \leq & \sum_{1\leq i\leq k} \prod_{j\neq i}n_j  + \\
& \sum_{(r_i)_i\in \Children^k} \left( \prod\limits_{1\leq i\leq k} n_{i,{r_i}} - \sum_{1\leq i\leq k}\prod_{j\neq i}n_{j,{r_j}}  \right)\\
 = & \sum_{1\leq i\leq k} \prod_{j\neq i}n_j + X - Y \numberthis \label{eq:recp2}
\end{align*}
where 
\[
X = \sum_{(r_i)_i\in \Children^k} \left( \prod\limits_{1\leq i\leq k}n_{i,{r_i}} \right) 
\]
and
\[
Y = \sum_{(r_i)_i\in \Children^k} \left( \sum_{1\leq i\leq k}\prod_{j\neq i}n_{j,{r_j}} \right)
\]

We compute $X$ and $Y$ respectively.
\begin{align*}
X =& \sum_{(r_i)_i\in \Children^k} \left( \prod\limits_{1\leq i\leq k} n_{i,{r_i}} \right)\\
=&\sum_{r_1\in \Children}n_{1,{r_1}}\cdot \left( \sum_{r_2\in \Children} n_{2,{r_2}}\cdot \left(  \dots \sum_{r_k\in \Children} n_{k,{r_k}} \right) \right)\\
\leq  & \prod\limits_{1\leq i\leq k}n_i\numberthis \label{eq:X}
\end{align*}
by factoring out every $n_{i,r_i}$ of the sum. Similarly,
\begin{align*}
Y =& \sum_{(r_i)_i\in \Children^k} \left( \sum_{1\leq i\leq k}\prod_{j\neq i}n_{j,{r_j}} \right)\\
=& \sum_{1\leq i\leq k}   \left(\sum_{(r_i)_i\in \Children^k} \prod_{j\neq i}n_{j,{r_j}} \right)\\
=& 2 \cdot \sum_{1\leq i\leq k} \left( \sum_{r_1\in \Children} n_{1,{r_1}}\cdot \dots \right. 
 \left( \sum_{r_{i-1}\in \Children} n_{{i-1},{r_{i-1}}} \cdot \right.\\
&  \left( \sum_{r_{i+1}\in \Children} n_{{i+1},{r_{i+1}}} \cdot \right.
 \dots \left.\left.\left. \left(  \sum_{r_k\in \Children} n_{k,{r_k}} \right) \right) \right)\right) \\
\geq & 2\cdot \sum_{1\leq i\leq k}\prod_{j\neq i} n_j \numberthis \label{eq:Y}
\end{align*}
as the inner sum in the second line is independent of $i$, and $r_i\in \Children$.

Substituting inequalities~\cref{eq:X} and \cref{eq:Y} to \cref{eq:recp2} we obtain
\begin{align*}
\Time(n_1,\dots,n_k)\leq &   \sum_{1\leq i\leq k} \prod_{j\neq i}n_j+X-Y\\
 \leq & c_1\cdot  \prod_{1\leq i\leq k}n_i - c_2 \cdot \sum_{1\leq i\leq k}\prod_{j\neq i}n_j
\end{align*}
for appropriate choices of the constants $c_1$, $c_2$, and thus $\Time(n_1,\dots,n_k)= O(n_1\cdot \ldots \cdot n_k)=O(n^k)$, as desired.
\end{proof}

\smallskip
\begin{lemma}\label{lem:product_tree_compl}
$\ProductTree$ requires $O(n^k)$ time and space.
\end{lemma}
\begin{proof}
It is easy to verify that $\ProductTree(G)$ performs a constant number of operations per node per bag in the
returned tree decomposition. 
Hence we will bound the time taken by bounding the size of $\ProductTree(G)$.
Consider a recursion step of $\ProductTree$ on input $(T_i(\Bag_i))_{1\leq i\leq k}$.
Let $n_i=|T_i(\Bag_i)|$ for all $1\leq i\leq k$, and $n_{i,{r_i}}=|T_i(\Bag_{i,{r_i}})|$, $r_i\in \Children$,
where $\Bag_{i,{r_i}}$ is the $r_i$-th child of $\Bag_i$.
In view of \cref{rem:bag_size}, the time required by $\ProductTree$ on this input is given by the recurrence
in \cref{eq:recP}, up to a constant factor.
The desired result follows from \cref{lemm:rec-conc}.
\end{proof}

\smallskip
\producttree*
\begin{proof}
\cref{lem:product_tree_dec} proves the correctness and \cref{lem:product_tree_compl} the complexity.
Here we focus on bounding the size of a bag $\Bag$ with $\Level(\Bag)\leq i\cdot \gamma$.
Let $(T_i(\Bag_i))_{1\leq i\leq k}$ be the input on $\ProductTree$ when it constructed $\Bag$ using \cref{eq:bag}
and $n_i=|T_i(\Bag_i)|$. 
Observe that $\Level(\Bag)=\Level(\Bag_i)$ for all $i$, and since each $T_i$ is $(\beta, \gamma)$-balanced,
we have that $n_i\leq O(n\cdot \beta^i)$.
Since each $T_i$ is $\alpha$-approximate, $|\Bag_i|=O(1)$ for all $i$.
It follows from \cref{eq:bag} and \cref{rem:bag_size} that $|\Bag|=O(n^{k-1}\cdot \beta^{i})$.
\end{proof}

\section{Details of \cref{sec:concurrent}}\label{sec:concurrent_proofs}

\smallskip
\uvdist*
\begin{proof}
By~\cite[Lemma~1]{CIPG15}, for every bag $\Bag_i$ with $i>1$ and path $P:u\Path v$, 
there exists a node $x_i\in \Bag_{i-1}\cap \Bag_{i}$.
Denote by $P_{x,y}$ a path $x\Path y$ in $G$.
Then
\begin{align*}
\Distance(u,v) =& \bigoplus_{P_{u,v}} \otimes(P_{u,v}) \\
=& \bigoplus_{x_i\in \Bag_{i-1}\cap \Bag_i} \left(\bigoplus_{P_{u,x_i}} \otimes(P_{u,x_i}) \otimes  \bigoplus_{P_{x_i,v}} \otimes(P_{x_i,v}) \right)\\
= & \bigoplus_{x_i\in \Bag_{i-1}\cap \Bag_i} \left(\Distance(u,x_i) \otimes  \Distance(x_i,v) \right)
\end{align*}
and the proof follows an easy induction on $i$.
\end{proof}

\smallskip
\pnodedistances*
\begin{proof}
By construction, for every node $v\in V$ that strictly refines $\PNode{v}$ (i.e., $v\Strictrefines \PNode{v}$), we have
$\PNode{\Weight}(\PNode{v}^1, v)=\Distance(\PNode{v}^1,v)=\One$ and
$\PNode{\Weight}(v, \PNode{v}^2)=\Distance(v, \PNode{v}^2)=\One$,
i.e., every such $v$ can reach (resp. be reached from) $\PNode{v}^2$ (resp. $\PNode{v}^1$)
without changing the distance from $\PNode{u}$.
The claim follows easily.
\end{proof}

\smallskip
\partialtreedecomp*
\begin{proof}
By \cref{them:product}, $\ProductTree(G)$ is a tree decomposition of $G$.
To show that $\PNode{T}$ is a tree decomposition of the partial expansion $\PNode{G}$,
it suffices to show that the conditions C1-C3 are met for every pair of nodes $\PNode{u}^1$, $\PNode{u}^2$
that correspond to a strict partial node $\PNode{u}$ of $\PNode{G}$.
We only focus on $\PNode{u}^1$, as the other case is similar.
\begin{compactenum}
\item[C1] This condition is met, as $\PNode{u}^1$ appears in every bag of $\PNode{T}$
that contains a node $u$ that refines $\PNode{u}^1$.
\item[C2] Since every node $\PNode{u}^1$ is connected only to nodes $u$ of $G$ that refine $\PNode{u}$,
this condition is also met.
\item[C3] First, observe that $\PNode{u}^1$ appears in the root bag $\PNode{\Bag}$ of $\PNode{T}$.
Then, for every simple path $P:\PNode{\Bag}\Path \PNode{\Bag}'$ from the root to some leaf bag $\PNode{\Bag}'$,
if $\PNode{\Bag}''$ is the first bag in $P$ where $\PNode{u}^1$ does not appear,
then some non-$\bot$ constituent of $u$ does not appear in bags of $\PNode{T}_{\PNode{\Bag}''}$,
hence neither does $\PNode{u}^1$. Thus, $\PNode{u}^1$ appears in a contiguous subtree of $\PNode{T}$.
\end{compactenum}
The desired result follows.
\end{proof}

\smallskip
\concrecurrence*
\begin{proof}
We analyze each recurrence separately. First we consider \cref{eq:rec1}.
Note that
\[
\left(n \cdot \left(\frac{1+\delta}{2}\right)^{\lambda-1}\right)^{3\cdot (k-1)} = \left(\frac{1+\delta}{2}\right)^{3\cdot(\lambda-1)\cdot (k-1)} \cdot n^{3\cdot (k-1)} \numberthis \label{eq:time_geom}
\]
and
\[
2^{\lambda \cdot k} \cdot \left(\frac{1+\delta}{2}\right)^{3\cdot (\lambda-1)\cdot (k-1)} = 
\frac{(1+\delta)^{3\cdot(\lambda-1)\cdot(k-1)}} {2^{2\cdot k\cdot \lambda +3\cdot (k+\lambda-1)}}
\numberthis \label{eq:time_leq1}
\]
and since $\log(1+\delta)=\frac{\ln (1+\delta)}{\ln 2} < \frac{\delta}{\ln 2}<2\cdot \delta$, we have 
\[
(1+\delta)^{3\cdot(\lambda-1)\cdot(k-1)} = 2^{\log(1+\delta) \cdot 3\cdot(\lambda-1)\cdot(k-1)} <  2^{6\cdot \delta\cdot(\lambda-1)\cdot(k-1)}
\]
Hence the expression in \cref{eq:time_leq1} is bounded by $2^x$ with
\begin{align*}
x &\leq 6\cdot \delta\cdot(\lambda-1)\cdot(k-1) - 2\cdot k\cdot \lambda +3\cdot (\lambda + k - 1)\\
&= -2\cdot \lambda\cdot k\cdot (1-3\cdot \delta ) + 3\cdot (\lambda + k -1)\cdot (1-2\cdot \delta)
\end{align*}
Let 
$
f(k) = -2\cdot \lambda\cdot k\cdot (1-3\cdot \delta ) + 3\cdot (\lambda + k -1)\cdot (1-2\cdot \delta)
$
and note that since $\lambda\geq \frac{4}{\eps} \geq 4$ and $\delta\leq \frac{\eps}{18}\leq \frac{1}{18}$, 
$f(k)$ is decreasing, and thus maximized for $k=2$, for which we obtain
$f(2)= -4\cdot \lambda\cdot (1-3\cdot \delta) + 3\cdot \lambda \cdot (1-2\cdot \delta) = -\lambda\cdot (1-6\cdot \delta) <0$
as $\delta \leq \frac{1}{18}$.
It follows that there exists a constant $c<1$ for which 
\[
2^{\lambda \cdot k}\cdot \Time_k\left(n\cdot \left(\frac{1+\delta}{2}\right)^{\lambda-1}\right) \leq c \cdot n^{3\cdot (k-1)}
\]
which yields that \cref{eq:rec1} follows a geometric series, and thus $\Time_k(n)=O(n^{3\cdot (k-1)})$.

We now turn our attention to \cref{eq:rec2}. 
When $k\geq 3$, an analysis similar to \cref{eq:rec1} yields the bound $O(n^{2\cdot (k-1)})$.
When $k=2$, since $\eps>0$, we write \cref{eq:rec2} as
\begin{align*}
\Space_2(n) &\leq n^{2+\eps} + 2^{2\cdot \lambda} \cdot \Space_2\left(n\cdot \left(\frac{1+\delta}{2}\right)^{\lambda-1}\right) \numberthis \label{eq:space2}
\end{align*}
Similarly as above, we have 
\[
\left(n \cdot \left(\frac{1+\delta}{2}\right)^{\lambda-1}\right)^{2+\eps} = \left(\frac{1+\delta}{2}\right)^{(2+\eps)\cdot(\lambda-1)} \cdot  n^{2+\eps} \numberthis \label{eq:space_geom}
\]
and
\[
2^{2\cdot \lambda} \cdot \left(\frac{1+\delta}{2}\right)^{(2+\eps)\cdot (\lambda-1)} = 
\frac{(1+\delta)^{(2+\eps)\cdot(\lambda-1)}} {2^{-2 + \eps\cdot(\lambda-1)}}
\numberthis \label{eq:space_leq1}
\]
and since $\log(1+\delta)=\frac{\ln (1+\delta)}{\ln 2} < \frac{\delta}{\ln 2}<2\cdot \delta$, we have 
\[
(1+\delta)^{(2+\eps)\cdot(\lambda-1)}< 2^{2\cdot \delta\cdot (2+\eps)\cdot(\lambda-1)}
\]

Hence the expression in \cref{eq:space_leq1} is bounded by $2^x$ with
\begin{align*}
x &\leq 2\cdot \delta\cdot (2+\eps)\cdot(\lambda-1) +2 - \eps\cdot(\lambda-1)\\
&= (\lambda-1)\cdot (2\cdot \delta\cdot (2+\eps) - \eps) + 2\\
&\leq (\lambda-1)\cdot \frac{4\cdot \eps + 2\cdot \eps^2 - 18\cdot \eps}{18} + 2\\
&\leq (1-\lambda)\cdot \eps\cdot \frac{2}{3} + 2\\
&\leq -(4-\eps)\cdot \frac{2}{3} + 2\leq 0
\end{align*}
since $\delta\leq \frac{\eps}{18}$ and $\lambda\geq \frac{4}{\eps}$ and $\eps\leq 1$.
It follows that there exists a constant $c<1$ for which 
\[
2^{2\cdot \lambda }\cdot \Space_2(n) \leq c \cdot n^{2+\eps}
\]
which yields that \cref{eq:space2} follows a geometric series, and thus $\Space_2(n)=O(n^{2+\eps})$.

\end{proof}
\section{Details of \cref{sec:arbitrary_product}}\label{sec:arbitrary_product_proofs}

\smallskip
\twclaim*
\begin{proof}
Observe that if we (i)~ignore the direction of the edges and (ii)~remove multiple appearances of the same edge, we obtain a tree.
It is known that trees have treewidth~1.
\end{proof}

\smallskip
\pathlemma*
\begin{proof}
Recall that only red edges contribute to the weights of paths in $G'$.
We argue that there is path $\PNode{P}:\Prod{x_i, x_i}\Path \Prod{x_j, x_j}$ in $G'$ that traverses a single red edge
iff there is an edge $(x_i,x_j)$ in $G$ with $\otimes(\PNode{P})=\Weight(x_i, x_j)$.
\begin{compactenum}
\item  Given the edge $(x_i, x_j)$, the path $\PNode{P}$ is formed by traversing
the red edge $(\Prod{x_i, y_j},\Prod{x_i, x_j})$ as
\begin{align*}
& \Prod{x_i, x_i} \to \Prod{x_i, y_i} \Path  \Prod{x_i, y_j} \to \Prod{x_i, x_j} \to\\
& \Prod{y_i, x_j} \Path \Prod{y_i, x_j} \to \Prod{x_j, y_j}
\end{align*}
Since $\Weight((\Prod{x_i, y_j},\Prod{x_i, x_j}))=\Weight(x_i, x_j)$ and all other edges of $\PNode{P}$
have weight $\One$, we have that $\otimes(\PNode{P})=\Weight(x_i, x_j)$.
\item Every path $\PNode{P}$ that traverses a red edge $\Prod{x_{i'}, y_{j'}}\to \Prod{x_{i'}, x_{j'}}$ has to traverse a blue edge 
to $\Prod{x_{j'}, x_{j'}}$. Then $x_{j'}$ must be $x_j$, otherwise $\PNode{P}$ 
will traverse a second red edge before reaching $\Prod{x_j, x_j}$.
\end{compactenum}
The result follows easily from the above.
\end{proof}

\section{Formal Pseudocode of Our Algorithms}\label{sec:algorithms}

The current section presents formally (in pseudocode) the algorithms that appear in the main
part of the paper.

\setcounter{algocf}{0}

\begin{algorithm}
\small
\SetAlgoNoLine
\DontPrintSemicolon
\caption{$\Rankalgo$}\label{algo:rank}
\KwIn{A component $\Comp$ of $T$, a natural number $\ell\in [\lambda]$}
\KwOut{A rank tree $\Ranktree$}
\BlankLine
Assign $\mathcal{T}\gets$ an empty tree\\
\uIf{$|\Comp|\cdot \frac{\delta}{2}\leq 1$}{
Assign $\RBag\gets \bigcup_{\Bag\in\Comp}\Bag$ and make $\RBag$ the root of $\mathcal{T}$
}
\uElseIf{$\ell>0$}
{\label{line:flag_true}
Assign $(\Seplist, \Complist)\gets \Splitalgo(\Comp)$\\
Assign $\RBag\gets \bigcup_{\Bag_i\in \Seplist}\Bag_i$\\
Assign $(\ov{\Comp}_1, \ov{\Comp}_2)\gets \Mergealgo(\mathcal{\Complist})$\\
Assign $\mathcal{T}_1 \gets \Rankalgo(\ov{\mathcal{\Comp}}_1, (\ell+1)\mod \Levelfactor)$\\
Assign $\mathcal{T}_2 \gets \Rankalgo(\ov{\mathcal{\Comp}}_2, (\ell+1)\mod \Levelfactor)$\\
Make $\RBag$ the root of $\mathcal{T}$ and $\mathcal{T}_1$ and $\mathcal{T}_2$ its left and right subtree\\
}
\Else{
\label{line:flag_false}
\eIf{$|\Nhood(\Comp)| > 1$}{
Let $\Bag\gets$ a bag of $\Comp$ whose removal splits $\Comp$ to $\ov{\Comp}_1$, $\ov{\Comp}_2$
with $|\Nhood(\ov{\Comp}_i)\cap \Nhood(\Comp)|\leq \frac{|\Nhood(\Comp)|}{2}$\\
Assign $\RBag\gets \Bag$\\
Assign $\mathcal{T}_1 \gets \Rankalgo(\ov{\Comp}_1, (\ell+1)\mod \Levelfactor)$\\
Assign $\mathcal{T}_2 \gets \Rankalgo(\ov{\Comp}_2, (\ell+1)\mod \Levelfactor)$\\
Make $\RBag$ the root of $\mathcal{T}$ and $\mathcal{T}_1$ and $\mathcal{T}_2$ its left and right subtree\\
}{
Assign $\mathcal{T}\gets \Rankalgo(\Comp, (\ell-1)\mod \Levelfactor)$\\
}
}
\Return $\mathcal{T}$
\end{algorithm}

\begin{algorithm}
\small
\SetAlgoNoLine
\DontPrintSemicolon
\caption{$\mathsf{ConcurTree}$}\label{algo:ConcurTree}
\KwIn{Tree-decompositions $T_i=(V_{T_i}, E_{T_i})_{1\leq i\leq k}$ with root bags $(\Bag_i)_{1\leq i\leq k}$.} 
\KwOut{A tree decomposition $T$ of the concurrent graph}
\BlankLine
Assign $\Bag \leftarrow \emptyset$\\
Assign $T \leftarrow$ a tree with the single bag $\Bag$ as its root\\
\For{$i \in [k]$}
{
Assign $\Bag \leftarrow \Bag \cup \left( \prod_{1 \leq j < i}  V_{T_j}(\Bag_j) \times \Bag_i \times \prod_{i < j \leq k} V_{T_j}(B_j)  \right)$
}
\If{ none of the $\Bag_i$'s is a leaf in its respective $T_i$}
{
    \For{ every sequence of bags $\Bag'_1, \ldots, \Bag'_{k}$ such that each $\Bag'_i$ is a child of $\Bag_i$ in $T_i$}
    {
        Assign $T'_i \leftarrow \mathsf{ConcurTree}(T_1(\Bag'_1), \ldots, T_k(\Bag'_k))$\\
        Add $T'_i$ to $T_i$, setting the root of $T'_i$ as a new child of $B$
    }
}
\Return T
\end{algorithm}

\newcommand{\wtbar}{\PNode{\mathsf{wt}}}
\newcommand{\wtcurly}{\boldsymbol{wt}}
\begin{algorithm}
\small
\SetAlgoNoLine
\DontPrintSemicolon
\caption{$\mathsf{ConcurPreprocess}$ \cref{item:partial_exp}}\label{algo:ConcurPreprocess1}
\KwIn{Graphs $(G_i=V_i, E_i)_{1\leq i\leq k}$, a concurrent graph $G(V,E)$ of $G_i$'s and a weight function $\mathsf{wt}: E \rightarrow \Sigma$} 
%\KwOut{}
\BlankLine

\tcc{Construct the partial expansion $\PNode{G}$ of $G$}
Assign $\PNode{V} \leftarrow V$\\
Assign $\PNode{E} \leftarrow E$\\
Create a map $\wtbar: \PNode{E} \rightarrow \Sigma$\\
Assign $\wtbar \leftarrow \mathsf{wt}$\\

    \ForEach{ $u'\in \prod_i(V_i\cup \{\bot\})$ }
    {
	    Let $u\in V$ such that  $u \sqsubset u'$\\
        Assign  $\PNode{V} \leftarrow \PNode{V} \cup \left\{ \PNode{u}^1, \PNode{u}^2 \right\}$\\
        Assign $\PNode{E} \leftarrow \PNode{E} \cup \left\{ (\PNode{u}^1, u), (u, \PNode{u}^2) \right\}$\\
        Set $\wtbar(\PNode{u}^1, u) \leftarrow \PNode{\mathbf{1}}$\\
        Set $\wtbar(u, \PNode{u}^2) \leftarrow \PNode{\mathbf{1}}$
    }

\Return $\PNode{G}=(\PNode{V}, \PNode{E})$ and $\PNode{\Weight}$
%Assign $\PNode{G} \leftarrow$ a graph with vertex set $\PNode{V}$ and edge set $\PNode{E}$\\

\end{algorithm}

\begin{algorithm}
\small
\SetAlgoNoLine
\DontPrintSemicolon
\caption{$\mathsf{ConcurPreprocess}$  \cref{item:partial_tree_dec}}\label{algo:ConcurPreprocess2}
\KwIn{A tree-decomposition $T=\Tree(G)=(V_T, E_T)$ and the partial expansion $\PNode{G}=(\PNode{V}, \PNode{E})$} 
%\KwOut{}
\BlankLine

\tcc{Construct the tree-decomposition $\PNode{T}$ of $\PNode{G}$}

%\For{$i \in [k]$}{
%    Assign $T_i \leftarrow \mathsf{Tree}(G_i)$
%}

%Assign $T \leftarrow \mathsf{ConcurTree}(k, T_1, \ldots, T_k)$\\
Assign $\PNode{V}_{\PNode{T}} \leftarrow \emptyset$

\ForEach{bag $\Bag\in V_T$}
{
Assign $\PNode{\Bag}\leftarrow \Bag$
    \ForEach{$u \in \Bag$}{
        \ForEach{$\PNode{u}\in \PNode{V}$ such that $u \Strictrefines \PNode{u}$}
        {
            Assign $\PNode{\Bag} \leftarrow \PNode{\Bag} \cup \left \{ \PNode{u}^1, \PNode{u}^2 \right \}$
        }
    }
Assign $\PNode{V}_{\PNode{T}}\leftarrow \PNode{V}_{\PNode{T}}\cup \{\PNode{\Bag}\}$
}

%Assign $\PNode{E}_{\PNode{T}} \leftarrow E_{T}$\\
\Return $\PNode{T}=(\PNode{V}_{\PNode{T}},E_{\PNode{T}})$

\end{algorithm}

\begin{algorithm}
\small
\SetAlgoNoLine
\DontPrintSemicolon
\caption{$\mathsf{ConcurPreprocess}$ \cref{item:ld} }\label{algo:ConcurPreprocess3}
\KwIn{The partial expansion tree-decomposition $\PNode{T}=(\PNode{V}_{\PNode{T}},\PNode{E}_{\PNode{T}})$, and weight function $\PNode{\Weight}$} 
\tcc{Local distance computation}
\ForEach{partial node $\PNode{u}$}
{
    Create two maps $\Fwd_{\PNode{u}}, \Bwd_{\PNode{u}} : \PNode{\Bag}_{\PNode{u}} \rightarrow \Sigma$ \\
    \For{$\PNode{v} \in \PNode{\Bag}_{\PNode{u}}$}
    {
        Assign $\Fwd_{\PNode{u}}(\PNode{v}) \leftarrow \wtbar(\PNode{u}, \PNode{v})$\\
        Assign $\Bwd_{\PNode{u}}(\PNode{v}) \leftarrow \wtbar(\PNode{v}, \PNode{u})$
    }
    
}

\ForEach{bag $\PNode{\Bag}$ of $\PNode{T}$ in bottom-up order}
{
    %Create a matrix $d'$\\
    Assign $d' \leftarrow$ the transitive closure of $\PNode{G}[\PNode{\Bag}]$ wrt $\Weightmap_{\PNode{\Bag}}$\\
    \ForEach{$\PNode{u}, \PNode{v} \in \PNode{\Bag}$}
    {
        \If{$\Level(\PNode{v}) \leq \Level(\PNode{u})$}
        {
            Assign $\Bwd_{\PNode{u}}(\PNode{v}) \leftarrow d' (\PNode{v}, \PNode{u})$\\
            Assign $\Fwd_{\PNode{u}}(\PNode{v}) \leftarrow d' (\PNode{u}, \PNode{v})$\\
        }
    }
}

\ForEach{bag $\PNode{\Bag}$ of $\PNode{T}$ in top-down order}
{
   % Create a matrix $d'$\\
    Assign $d' \leftarrow$ the transitive closure of $\PNode{G}[\PNode{\Bag}]$ wrt $\Weightmap_{\PNode{\Bag}}$\\
    \ForEach{$\PNode{u}, \PNode{v} \in \PNode{\Bag}$}
    {
        \If{$\Level(\PNode{v}) \leq \Level(\PNode{u})$}
        {
            Assign $\Bwd_{\PNode{u}}(\PNode{v}) \leftarrow d' (\PNode{v}, \PNode{u})$\\
            Assign $\Fwd_{\PNode{u}}(\PNode{v}) \leftarrow d' (\PNode{u}, \PNode{v})$\\
        }
    }
}

\end{algorithm}

\newcommand{\fwdp}{\Fwd^{+}}
\newcommand{\bwdp}{\Bwd^{+}}
\newcommand{\vancestor}{\PNode{\mathcal{V}}}
\newcommand{\wtcurlyp}{\wtcurly^{+}}

\begin{algorithm}
\small
\SetAlgoNoLine
\DontPrintSemicolon
\caption{$\mathsf{ConcurPreprocess}$ \cref{item:ancestors} }\label{algo:ConcurPreprocess4}
\KwIn{The partial expansion tree-decomposition $\PNode{T}=(\PNode{V}_{\PNode{T}},\PNode{E}_{\PNode{T}})$ and maps $\Fwd_{\PNode{u}},\Bwd_{\PNode{u}}:\PNode{\Bag}_u\to \Sigma$ for every partial node $\PNode{u}$} 
\tcc{Ancestor distance computation}

\ForEach{node $u \in V$}
{
    Create two maps $\From_u, \To_u : \vancestor_{\PNode{T}}(\PNode{\Bag}_u) \rightarrow \Sigma$
}

\ForEach{bag $\PNode{\Bag}$ of $\PNode{T}$ in DFS order starting from the root}
{
    {
        Let $\PNode{\Bag}'$ be the parent of $\PNode{\Bag}$\\
        \ForEach{node $u \in \PNode{B} \cap V$ such that $\PNode{B}$ is the root of $u$} 
        {
                \ForEach{$\PNode{v} \in \AncestV{\PNode{V}}_{\PNode{T}}(\PNode{\Bag}_u)$}
                {
                    Assign $\From_u(\PNode{v}) \leftarrow  \underset{x \in \PNode{\Bag} \cap \PNode{\Bag}'}{\bigoplus} \Fwd_u (x) \otimes \Weightmap^{+}(x, \PNode{v})$\\
                    Assign $\To_u(\PNode{v}) \leftarrow \underset{x \in \PNode{\Bag} \cap \PNode{\Bag}'}{\bigoplus} \Bwd_u (x) \otimes \Weightmap^{+}(\PNode{v}, x)$\\
                }
        }
        
    }
}

%\tcc{Part 5}
%Preprocess $\PNode{T}$ for answering least common ancestor (LCA) queries
\end{algorithm}

\begin{algorithm}
\small
\SetAlgoNoLine
\DontPrintSemicolon
\caption{$\Concqueryalgo$ Single-source query}\label{algo:singlesourcequery}
\KwIn{A source node $u \in V$} 
\KwOut{A map $A: V \rightarrow \Sigma$ that contains distances of vertices from $u$}
\BlankLine

Create a map $A: V \rightarrow \Sigma$\\
\For{$v \in V$}
{
    Assign $A(v) \gets \PNode{\textbf{0}}$
}

\For{every bag $\PNode{\Bag}$ of $\PNode{T}$ in BFS order starting from $\PNode{\Bag}_u$}
{
    \For{$x, v \in \PNode{\Bag} \cap V$}
    {
        \If{$\Level(v) \leq \Level(x)$}
        {
            Assign $A(v) \gets A(v) \oplus A(x) \otimes \Fwd_x(v)$
        }
    }
}

\Return $A$

\end{algorithm}

\begin{algorithm}
\small
\SetAlgoNoLine
\DontPrintSemicolon
\caption{$\Concqueryalgo$ Pair query}\label{algo:pairquery}
\KwIn{Two nodes $u, v \in V$} 
\KwOut{The distance $\Distance(u,v)$}
\BlankLine

Let $\PNode{\Bag} \gets $ the LCA of $\PNode{B}_u$ and $\PNode{B}_v$ in $\PNode{T}$ \\

Assign $d \gets \PNode{\mathbf{0}}$

\For{$x \in \PNode{\Bag} \cap V$}
{
    Assign $d \gets d \oplus \fwdp_u(x) \otimes \bwdp_v(x)$
}

\Return $d$

\end{algorithm}

\begin{algorithm}
\small
\SetAlgoNoLine
\DontPrintSemicolon
\caption{$\Concqueryalgo$ Partial pair query}\label{algo:partialpairquery}
\KwIn{Two partial nodes $\PNode{u}, \PNode{v} \in \PNode{V}$, at least one of which is strictly partial} 
\KwOut{The distance $\Distance(\PNode{u}, \PNode{v})$}
\BlankLine

\uIf{both $\PNode{u}$ and $\PNode{v}$ are strictly partial}
{
    \Return $\Fwd_{\PNode{u}^1} (\PNode{v}^2)$
}
\uElseIf{$\PNode{u}$ is strictly partial}
{
    \Return $\bwdp_v(\PNode{u}^1)$
}
\Else{
    \Return $\fwdp_u (\PNode{v}^2)$
}

\end{algorithm}

\end{document}